\documentclass[letterpaper, 12pt]{article}
\usepackage{natbib}
\usepackage[T1]{fontenc}
\usepackage[active]{srcltx}
\usepackage{booktabs}
\usepackage{babel}
\usepackage{float}
\usepackage{amsthm}
\usepackage{amsmath}
\usepackage{amssymb}
\usepackage{graphicx, graphics}
\usepackage{setspace}
\usepackage{esint}
\usepackage{subfig}
\usepackage{caption}
\usepackage{multirow}

\usepackage{algorithm}
\usepackage{algorithmic}
\makeatletter
\newenvironment{breakablealgorithm}
{
		\begin{center}
			\refstepcounter{algorithm}
			\hrule height.8pt depth0pt \kern2pt
			\renewcommand{\caption}[2][\relax]{
				{\raggedright\textbf{\fname@algorithm~\thealgorithm} ##2\par}%
				\ifx\relax##1\relax 
				\addcontentsline{loa}{algorithm}{\protect\numberline{\thealgorithm}##2}%
				\else 
				\addcontentsline{loa}{algorithm}{\protect\numberline{\thealgorithm}##1}%
				\fi
				\kern2pt\hrule\kern2pt
			}
		}{
		\kern2pt\hrule\relax
	\end{center}
}

\usepackage[usenames]{color}
\usepackage[normalem]{ulem}
\definecolor{mygreen}{rgb}{0,.5,0}

\newcommand{\tk}{t+\frac{k}{M}}

\newcommand{\be}{\begin{equation}}
\newcommand{\ee}{\end{equation}}
\newcommand{\bee}{\begin{equation*}}
\newcommand{\eee}{\end{equation*}}

\newtheorem{proposition}{Proposition}[section]
\newtheorem{lemma}{Lemma}[section]

\newtheorem{remark}{Remark}[section]

\theoremstyle{plain}
  
  \theoremstyle{plain}

  \theoremstyle{plain}
  \newtheorem*{assumption*}{\protect\assumptionname}
  \theoremstyle{plain}

  \theoremstyle{remark}
  \newtheorem*{rem*}{\protect\remarkname}
  \theoremstyle{plain}

\usepackage{fullpage}
 \usepackage{pdfsync}
\usepackage[title,titletoc]{appendix}
\usepackage[bookmarks=true, colorlinks=true, linkcolor=blue, citecolor=blue, breaklinks=true]{hyperref}
\hypersetup{colorlinks, urlcolor=blue, linkcolor=blue, citecolor=blue}

  \providecommand{\assumptionname}{Assumption}
  \providecommand{\lemmaname}{Lemma}
  \providecommand{\propositionname}{Proposition}
  \providecommand{\remarkname}{Remark}
\providecommand{\theoremname}{Theorem}
\providecommand{\corollaryname}{Corollary}

\begin{document}

\title{Reinforcement Learning for Financial Index Tracking}

\date{November 10, 2024}

\author{Xianhua Peng\thanks{HSBC Business School, Peking University, University Town, Nanshan District, Shenzhen, 518055, China. Email: xianhuapeng@pku.edu.cn.} \thanks{Sargent Institute of Quantitative Economics and Finance, HSBC Business School, Peking University, University Town, Nanshan District, Shenzhen, 518055, China.}  \and Chenyin Gong\thanks{Business School, Hong Kong University of Science and Technology, Clear Water Bay, Kowloon, Hong Kong, China. Email: chenyin.gong@connect.ust.hk.}
	\and 
	Xue Dong He\thanks{Department of System Engineering and Engineering Management, The Chinese University of Hong Kong, Shatian, Hong Kong, China, Email: xdhe@se.cuhk.edu.hk.}
		}

\maketitle

\begin{abstract}  

We propose the first discrete-time infinite-horizon dynamic formulation of the financial index tracking problem under both return-based tracking error and value-based tracking error. The formulation overcomes the limitations of existing models by incorporating the intertemporal dynamics of market information variables not limited to prices, allowing exact calculation of transaction costs, accounting for the tradeoff between overall tracking error and transaction costs, allowing effective use of data in a long time period, etc. The formulation also allows novel decision variables of cash injection or withdraw. We propose to solve the portfolio rebalancing equation using a Banach fixed point iteration, which allows to accurately calculate the transaction costs specified as nonlinear functions of trading volumes in practice. We propose an extension of deep reinforcement learning (RL) method to solve the dynamic formulation. Our RL method resolves the issue of data limitation resulting from the availability of a single sample path of financial data by a novel training scheme. A comprehensive empirical study based on a 17-year-long testing set demonstrates that the proposed method outperforms a benchmark method in terms of tracking accuracy and has the potential for earning extra profit through cash withdraw strategy. 





\end{abstract}

\baselineskip 18pt

\section{Introduction}
An index fund aims to replicate the performance of a designated financial index. Index funds including index mutual funds and index ETF have gained popularity in recent years. At the end of 2021, index domestic equity funds held 16\% of US stock market capitalization,  for the first time surpassing the 14\% share held by active domestic equity funds \citep[][p. 30]{ICI-2022}. 
There are mainly two kinds of index tracking strategies used in practice: full replication and representative sampling. The former 
holds a portfolio with the same weights as those of the index portfolio; the latter invests in a portfolio with possibly smaller number of components and different weights from those of the index portfolio. 
The full replication method has the practical drawback of high transaction costs resulting from frequent trading, trading of illiquid assets, expensive data of index portfolio weights, and index revisions that necessitate substantial holding rebalance \citep[see, e.g.,][]{beasley2003evolutionary, benidis2017sparse}. 
Such drawback can be mitigated by representative sampling method. Many top US index equity funds ranked by asset under management (AUM) employ representative sampling strategies, such as the Vanguard Total Stock Market Index Fund (with AUM over \$1.2 trillion), the Fidelity 500 Index Fund (with AUM about \$400 billion), the iShares Core S\&P 500 ETF (with AUM about \$300 billion) and the iShares Russell 2000 ETF (with AUM about \$53 billion). 

The index fund manager needs to rebalance the tracking portfolio from time to time in order to minimize the index tracking error, as a fixed tracking portfolio cannot accurately track an index over a long time horizon due to change of the prices or composition of the index over time. Portfolio rebalancing incurs transaction costs that reduce the fund return, which affects the index tracking performance. Index tracking strategy therefore should take into account the tradeoff between index tracking error and transaction costs over a long time horizon.

There has been a large literature developing index tracking methods. To the best of our knowledge, all papers except \citet{yao2006tracking} formulate the index tracking problem in a static setting, which study how to determine the tracking portfolio weights at the beginning of a given short time period. The portfolio weights are obtained using the market prices/returns of stocks and index during a historical short time period right ahead of the given time period. 
The portfolio is then rebalanced to the obtained portfolio weights at the beginning and/or during the given time period. The static formulations of index tracking face six major challenges. $(i)$ The static formulations 
rely on the assumption that the joint distribution of the cross-sectional daily returns of stocks and index during the aforementioned two time periods are the same. Such an assumption may not hold if the historical time period is too long, so these formulations can only use the data during a short historical time period, wasting a large amount of data further earlier than the given time period. $(ii)$ The portfolio weights determined by the static formulations are fixed numbers that only depend on the data of the historical time period and hence are not responsive to the market change in the given time period. $(iii)$ The static formulations cannot directly incorporate transaction costs at more than one time point. In fact, most formulations completely ignore transaction costs; only very few formulations take into account the proportional and fixed transaction cost incurred at the beginning of the given time period under restrictive settings, but not those during the given time period; no formulations in the literature have incorporated non-proportional transaction costs specified in practice. $(iv)$ Static formulations cannot take into account the tradeoff between tracking error and transaction costs in a long time horizon that includes multiple time periods. $(v)$ As transaction costs cannot be incorporated directly, most static formulations try to indirectly incorporate transaction costs by restricting or penalizing the number of stocks in the tracking portfolio, based on the intuition that smaller number of stocks in the tracking portfolio may lead to lower transaction costs. However, such an intuition does not have solid empirical or theoretical support, and may not hold when multiple periods are considered. In fact, with the cardinality constraint, it is likely that one subset of stocks are selected into the tracking portfolio in one period and then another different subset of stocks are selected in the next period; this may result in high transaction costs. $(vi)$ The static formulation cannot utilize market information other than prices or returns, such as trading volume, market sentiment, etc. 
\citet{yao2006tracking} propose a continuous-time parametric stochastic control formulation of the index tracking problem that does not address some of the aforementioned challenges including ignoring transaction costs and only using price data.

In this paper, we propose a discrete-time infinite-horizon dynamic formulation of the index tracking problem to address the aforementioned challenges. The dynamic formulation minimizes the discounted cumulative tracking error over an infinite number of time periods; at the beginning of each time period, the fund manager determines the portfolio weights based on the current market state and then rebalances the portfolio to the determined portfolio weights
during the time period at a given frequency such as daily or only once at the beginning of the period. The dynamic formulation overcomes the challenges in the following aspects. 
 $(i)$ The formulation does not need the assumption that the joint distribution of the cross-sectional daily returns of stocks and index during the historical time period is the same as that in the given time period, so it allows the utilization of market data in a much longer historical time period. $(ii)$ The portfolio weights in our formulation are dynamic in the sense that the weights or the parameters of their conditional distribution at any rebalancing time are a function of the market state variables at the rebalancing time. The function is learned from data in the historical time period but the weights depend on the market information at the rebalancing time, so the weights are more responsive to the current market information. $(iii)$ Our formulation directly incorporates and accurately calculates the transaction costs under various specifications, particularly those specified as nonlinear functions of trading volume. $(iv)$ Our formulation considers the transaction costs at all rebalancing time during an infinite time horizon, so it exactly takes into account the tradeoff between tracking error and transaction costs over a long (infinite) time period. $(v)$ Our formulation allows to directly incorporate the exact transaction costs into the minimization of tracking error, so we do not need to impose cardinality constraints or penalty terms in our formulation. $(vi)$ The formulation allows to utilize any market information other than prices or returns by including market information into state variables. 

In our formulation, we consider two specifications of tracking error: the return tracking error based on the differences between daily returns of the tracking portfolio and those of the index, and the value tracking error based on the differences between daily values of the tracking portfolio and those of the index. The value tracking error essentially measures the differences between cumulative returns over all investment horizons that may be concerned by fund investors. 

In addition, we propose a dynamic formulation that allows index fund manager to inject or withdraw cash from the tracking portfolio for better index tracking performance under the value tracking error. In the formulation, the amount of cash injection or withdraw at each rebalancing time is a decision variable that is dependent on the current market information. This introduces novel and flexible index fund management strategies. 

The key to incorporate transaction costs in our dynamic formulations is to solve the rebalancing equation at each rebalancing time, which 
presents that the portfolio value right after rebalancing should equal the portfolio value before rebalancing subtracting the transaction cost that is specified as a nonlinear equation of trading volume in terms of both dollar value and number of shares. We propose a Banach fixed point iteration algorithm to solve the rebalancing equation accurately and efficiently, which leads to accurate calculation of the exact transaction cost and the number of shares to be traded at the rebalancing time. The algorithm can be applied in general portfolio management problems that need to rebalance portfolios. 



The dynamic formulations of the index tracking problem are difficult to solve due to high dimension of state variables and decision variables. To track a broad market index such as the S\&P 500 index, the tracking portfolio typically contains hundreds of stocks, which leads to state variables with dimension in the order of tens of thousands or higher. 
We propose an extension of deep reinforcement learning (RL) method to solve the dynamic formations of index tracking problem. The RL method has several advantages. First, the RL method allows high dimensional state variables that can include as much market information as possible, not just limited to prices or returns. Second, the state variables of the RL method include market data over a time period such as one year instead of a single time point, so the RL method allows model-free learning of the joint dynamics of daily prices/returns of stocks and index in both cross-sectional and intertemporal dimension; in contrast, the static methods only utilize the joint distribution of daily prices/returns of stocks and index in the cross-sectional dimension. The model-free learning mitigates model misspecification error in existing continuous-time stochastic control methods based on parameter models. Third, the RL method allows a unified approach to determine both portfolio weights and other decision variables such as cash injection and withdraw. Fourth, the RL method allows any specification of tracking errors. Fifth, the RL method learns the optimal policy, which is a function of or a distribution conditional on the current state variable. Therefore, the RL method learns the relation between the portfolio weights and the current market state. 
Because the relation may stay the same even over a long time period, the RL method can effectively utilize the data in a very long time period. 
In contrast, the static methods solve for only the numerical values of the portfolio weights based on the data in the historical time period. In order for the portfolio weights to be applicable in the current time period, the historical time period has to be short in the static methods.  
 

Our proposed RL method extends the existing RL methods used in computer games or robotics mainly in two aspects. First, we propose a new training scheme for the RL method to address the issue of data limitation, one of the major difficulties of an RL method for financial index tracking. In the application of RL for computer games and robotics, 
unlimited data can be generated for training by playing games or operating robots as many times as needed, but there is only one sample path of the time series of financial data in the context of index tracking. The time period of index tracking problems is typically a quarter or half a year, so basically we need a large number of quarterly or semiyearly data of training trajectories of the interaction between the RL agent and financial market. However, we only have a sample path of at most a few decades of training data. In addition, without a good parametric model of the joint dynamics of the high dimensional state variable involving a large number of assets and market information variables, training data cannot be obtained from simulation. We propose a new training scheme that generate enough quarterly or semiyearly data from daily data by randomly selecting the starting date of the first time period during training. Second, our RL method needs to estimate the value function at the terminal state, i.e., the state at the end of a finite length episode, because the index tracking problem is an infinite-horizon problem in which the value function at any state is never zero. In contrast, the training trajectory of an RL method for a computer game problem reaches a terminal state after a finite number of steps, and the value function at the terminal state is simply equal to zero.

We carry out comprehensive empirical analysis of the out-of-sample performance of our RL method on a rolling-window basis. We are the first to test the out-of-sample performance of an index tracking strategy using a 17-year-long testing set, which is much longer than those in the existing literature. We test the out-of-sample performance of our RL method for three stock indices with different weighting schemes: the price-weighted Dow Jones Industrial Average (DJIA) index, the value-weighted S\&P 500 index, and the equally-weighted S\&P 500 Equal Weight Index (EWI). The empirical results show that: $(i)$ Our RL method achieves smaller mean return tracking error than a benchmark method on the testing set for return tracking of the indices. $(ii)$ For value tracking of indices, our RL method not only achieves one magnitude smaller mean value tracking error than the benchmark method, but also produces a positive mean cash withdraw, implying that the strategy of cash injection and withdraw may earn extra profit for the fund manager.



In summary, the contributions of this paper are fivefold. $(i)$ We are the first to propose a discrete-time infinite-horizon dynamic formulation of the index tracking problem. In the formulation, the portfolio weights at each rebalancing time are contingent on the state variables at that time. The formulation explicitly takes into account the transaction costs at all rebalancing time so it exactly accounts for the tradeoff between tracking error and transaction costs in the long run. $(ii)$ We are the first to incorporate cash injection and withdraw decision variables in the formulation of index tracking problem, which allows more flexible index tracking strategies. $(iii)$ We are the first to propose a fast Banach fixed  point iteration algorithm to exactly solve the rebalancing equation and calculate the transaction cost under various specification of transaction costs, especially those used in practice. The algorithm can be used in general portfolio management problems in which rebalancing is needed. $(iv)$ We are the first to use a RL based method to solve the index tracking problem. Our RL method extends the existing RL methods by a novel training scheme overcoming the data limitation issue and taking into account the value function at the terminal state of a finite training episode. $(v)$ We are the first to carry out out-of-sample test of index tracking methods by using a 17-year-long testing set. The empirical results demonstrate that the proposed method achieves smaller tracking error than a benchmark method for three indices with different weighting schemes and may earn extra profit through cash withdraw and injection.

\subsection{Literature Review}


\subsubsection{Static Formulations}

Most existing literature formulate the index tracking problem in a static setting, i.e., a one-period setting. The optimal portfolio weights at the beginning of a given time period are determined by the return data during a historical time period right ahead of the given time period. Then the portfolio is rebalanced to the determined portfolio weights at the beginning of the time period. 
Most papers do not explicitly incorporate or calculate transaction costs in their problem formulations; instead, they implicitly take into account transaction costs by choosing a sparse portfolio with a smaller number of stocks than that of the index. The sparsity consideration is based on the intuition that trading a smaller number of liquid stocks may reduce transaction costs. The intuition does not have empirical or theoretical support, and may not hold if the total transaction costs during multiple time periods are considered. There are three ways to choose a sparse portfolio in the literature. The first is to select a given small number of stocks from a large number of stocks based on some heuristic selection criteria and stock characteristics such as beta and size \citep[see, e.g.,][]{Oh-2005}. The second is to use a constraint on the cardinality of tracking portfolio \citep[see, e.g.,][]{beasley2003evolutionary, canakgoz2009mixed, mutunge2018minimizing, kim2020index, zheng2020index, Li2022}. The third is to penalize the cardinality in the objective function \citep[see, e.g.,][]{jansen2002optimal, benidis2017sparse}. 

The sparsity constraint or penalty leads to mixed integer programming problems. \citet{gilli2002threshold}, \citet{beasley2003evolutionary}, \citet{maringer2007index}, \citet{chiam2013dynamic}, and \citet{andriosopoulos2014performance} provide heuristic solutions for index tracking problems that are formulated as mixed-integer nonlinear programming problems. These heuristic methods implement search schemes to evolve an initial population of solutions through reproduction and mutation until a satisfactory solution is reached. 
\citet{gaivoronski2005optimal} propose mixed-integer programming formulations that consider different measures of tracking error. 
\citet{canakgoz2009mixed} and \citet{mezali2013quantile} present mixed-integer linear programming (MILP) formulations for the index tracking problem from a regression-based view. 
\citet{strub2018optimal} present a MILP formulation based on an index value tracking objective function. 
\citet{benidis2017sparse} propose a majorization-minimization (MM) method to minimize the sum of tracking error and a cardinality penalty term. The MM method is based on approximating the $\ell_0$ norm function by a differentiable function. 
\citet{Li2022} propose to solve two constrained optimization problems that minimize the tracking error subject to an upper bound of the number of selected assets and minimize the number of selected assets subject to an upper bound on the tracking error, respectively.

Recently a few literature apply deep learning methods to solve the index tracking problem in a static setting. 
\citet{heaton2016deep}, \citet{ouyang2019index}, \citet{kim2020index} propose to select a given number of assets from the candidate assets using deep autoencoder. \citet{zheng2020index} propose to solve the index tracking problem with cardinality constraint using a stochastic neural network method. 

Both return tracking error and value tracking error have been proposed in the literature. For example, \citet{maringer2007index}, \citet{benidis2017sparse}, and \citet{Li2022} consider the root mean square error (MSE) of the portfolio returns and the index returns as their tracking error functions. \citet{beasley2003evolutionary} measure return tracking error in terms of general $q$-norm. 
\citet{jansen2002optimal} and \citet{mutunge2018minimizing} 
define the tracking error as the variance of the difference of the portfolio return and the index return. 
\citet{gaivoronski2005optimal} and \citet{strub2018optimal} use value tracking error that measures the differences between portfolio values and index values in their objective functions.

In the static formulations based on return tracking errors, the decision variables are the portfolio weights at the beginning of the testing time period. The weights are obtained by using daily returns during the historical time period right ahead of the testing time period, assuming that the weights are constant over the historical time period. 
However, when the out-of-sample tracking error is calculated during the testing period, the portfolio is only rebalanced once at the beginning of the testing period and then the shares of stocks are constant during the testing period. Because the portfolio weights change daily due to price change, the reported our-of-sample tracking error is not the true error corresponding to the weights determined at the beginning of the testing period. This results in inconsistency between the model training and testing. In contrast, our method can incorporate both the constant-weights and constant-shares assumption and ensure consistency between in-sample training and out-of-sample testing. 

\subsubsection{Dynamic Formulations}

Dynamic formulations differ from static formulations mainly in three aspects: $(i)$ Dynamic formulations take into account the tracking error in an infinite time horizon or multiple time periods. $(ii)$ Dynamic formulations determine the optimal control policy, i.e., the functional relation between the tracking portfolio weights and the current market state variable. $(iii)$ In dynamic formulations, the tracking portfolio weights during the testing time period depend on the functional relation determined from the historical time period and the current market state; in contrast, in static formulations, the portfolio weights during the testing time period are determined only by market information in the historical time period.



In the literature, there has been only one paper that studies index tracking in a dynamic setting. \citet{yao2006tracking} propose to formulate the financial index tracking problem as a continuous-time stochastic linear quadratic control problem. In the formulation, the portfolio weights are the feedback controls as a function of the current market information and are solved by semidefinite programming. The formulation is based on a multi-dimensional geometric Brownion motion model of the stock dynamics which is subject to model misspecification error. The formulation does not take into account transaction costs. In addition, the formulation needs to know the numbers of shares of component stocks in the index and assumes those numbers are fixed, which only hold during a time period when there is no index rebalancing. For example, the assumption may only hold during each quarter between the quarterly rebalancing of the S\&P 500 Equal Weight Index. 

Our dynamic discrete-time formulation of the index tracking problem complements \citet{yao2006tracking} in several aspects. First, our formulation is model-free and data-driven which avoid model misspecification error and parameter estimation error. Second,  we allow the exact calculation and incorporation of 
transaction costs in the formulation. Third, our method does not need the information of the composition of indices and 
makes no assumption on whether the composition is fixed or not, so it can be applied to track different types of indices
over a long (infinite) time horizon. 
Forth, our method allows tracking portfolio weights to depend on market information other than prices, such as trading volumes and market sentiment. Fifth, our method allows flexible index tracking strategies such as injection or withdrawn of cash from the tracking portfolio (see Section \ref{subsec:vb_form}).


\subsubsection{Transaction Costs}\label{subsubsec:TC}
Transaction costs are incurred when the portfolio is rebalanced. In practice, the transaction costs specified by brokers are nonlinear functions of the trading volume in terms of dollar value and the number of shares traded; 
see, e.g., the formula of transaction cost in Equation \eqref{eq:tc} of Section \ref{subsec:transaction_cost}. 
Furthermore, government regulators charge regulatory fees (see, e.g., Equation \eqref{eq:tc_with_reg_fee} in Appendix \ref{appendix:regulatory_fees}), which further complicates the specification of transaction costs. 

Transaction costs are handled in three different ways in the literature: $(i)$ ignored completely; 
$(ii)$ financed by an external account; $(iii)$ borne by the portfolio itself (self-financing). 
The majority of studies, such as \citet{jansen2002optimal}, \citet{dose2005clustering}, \citet{mutunge2018minimizing}, \citet{heaton2016deep}, \citet{ouyang2019index}, \citet{kim2020index}, \citet{kwak2021neural}, and \citet{Li2022}, ignore transaction costs in their problem formulations, which deviates from the practice and results in inaccurate results of index tracking performance. 

In some other studies such as 
\citet{Guastaroba2012} and \citet{benidis2017sparse}, transaction costs are assumed to be paid out of a separate account. Under this circumstance, transaction costs are not controlled by objective functions at all, 
so it is not clear how much transaction costs are incurred by the index tracking strategies. 


Only very few studies incorporate proportional and/or fixed transaction costs as being paid by the tracking portfolio itself, but only one of them can accurately compute transactions costs in the restrictive case of holding cash in the tracking portfolio; these studies consider only the transaction cost at the \emph{single} time point of the \emph{beginning} of time period but \emph{no} transaction costs \emph{during} the time period. In this self-financing case, the rebalancing equation holds at the rebalancing time.
\citet{beasley2003evolutionary} include proportional transaction costs in the problem formulation and propose a evolutionary heuristic to solve the problem, but the heuristic cannot guarantee the rebalancing equation to hold exactly and hence the transaction costs are not exactly calculated.
\citet{canakgoz2009mixed} include variables representing proportional transaction costs in their mixed-integer linear programming problem formulation, but the values of those variables obtained from solving the problem
are over-estimates of the actual transaction costs, so the actual rebalancing equation does not hold, and their method cannot calculate the actual transaction costs. 
\citet{strub2018optimal} include proportional and fixed transaction costs in the formulation of a mixed integer linear programming problem, but that formulation requires holding cash in the tracking portfolio. 
Their value-tracking objective function can only incorporate the tracking error at the beginning of each time period but not the daily tracking error within each time period. 


There has been no paper in the literature that can incorporate the transaction cost specified in practice as Equation \eqref{eq:tc}. The major difficulty lies in solving the rebalancing equation, which in this case is nonlinear with respect to the portfolio value and number of shares of stocks after rebalancing. So the portfolio value after rebalancing and the transaction cost need to be determined simultaneously, which is nontrivial to solve.



In this paper, we assume the transaction costs are paid by the tracking portfolio itself so that the index tracking strategy can take into account both tracking error and transaction costs. We propose a method that can exactly and accurately incorporate the transaction cost specified in practice. The method solves the rebalancing equation exactly and hence allows calculation of the exact transaction cost and the portfolio value after rebalancing. 
Furthermore, our method for calculating transaction costs is fast and converges in only several iterations (see more detailed discussion in Section \ref{subsec:transaction_cost}). The fast calculation is essential in generating sample trajectories for the reinforcement learning method, as the portfolio rebalancing may take place on a daily basis. 




\subsubsection{Reinforcement Learning for Finance}

%

Our paper contributes to an emerging literature that use RL for problems in finance 
such as optimal execution \citep[see, e.g.,][]{nevmyvaka2006reinforcement, hendricks2014reinforcement, ning2021double}, and active portfolio selection \citep[see, e.g.,][]{almahdi2017adaptive, pendharkar2018trading, yu2019model, cong2020alphaportfolio}. 
These studies aim to solve different problems from index tracking, and their RL methods do not apply to the index tracking problem. 
We are the first to use reinforcement learning for solving the index tracking problem. In particular, we are the first to propose a method to exactly solve the rebalancing equation and calculate the transaction cost. 
For a broader survey on recent advances in RL in finance, readers are referred to \citet{hambly2023recent}.

The rest of the paper is organized as follows. Section \ref{sec:problem_formulation} proposes two dynamic formulations for financial index tracking in terms of returns and values, respectively. Section \ref{sec:RL_for_IT} presents a RL method to solve these two financial index tracking problems. Section \ref{sec:experiments} reports empirical results followed by a conclusion section.

\section{A Dynamic Formulation of the Index Tracking Problem}\label{sec:problem_formulation}

In this section, we propose the formulation of the index tracking problem as a discrete-time infinite-horizon stochastic control problem. Our formulation is different from existing ones mainly in two aspects: $(i)$ We deduct the exact value of transaction costs from the portfolio value at each portfolio rebalancing time so that the optimization problem directly takes into account transaction costs. We will show that the transaction costs can be exactly calculated in Section \ref{subsec:transaction_cost}. $(ii)$ We dynamically control not only the portfolio weights but also the amount of cash inflow and outflow of the tracking portfolio, which enables better reduction of the value tracking error. 

We discuss the dynamic setting of the problem in Section \ref{subsec:setting_problem}. In Section \ref{subsec:formulation_return_based} and Section \ref{subsec:vb_form}, we consider two formulations of the index tracking problem in terms of the return based tracking error and the value based tracking error, respectively. We discuss the flexibility of the formulation in Section \ref{subsec:fis_form} and its extension for index enhancement problems in Section \ref{subsec:index_enhancement}.



\subsection{The Setting of the Problem}\label{subsec:setting_problem}

We consider a general discrete-time infinite-horizon problem starting from time 0. 
Suppose the fund manager intends to track the performance of a stock index by a tracking portfolio consisting of $N$ stocks and a cash position. The $N$ stocks may not necessarily be component stocks of the stock index. The fund manager can rebalance the tracking portfolio dynamically for better tracking performance.
Suppose each time period $[t, t+1)$, $t = 0, 1, 2, \dots$, has $M$ trading days denoted as $t + \frac{k}{M}$, $k = 0, \dots, M-1$. 
For example, if the time unit is a quarter, then each time period has $M=66$ trading days. 
For each stock $i = 1, \dots, N$, let $p_{i,t+\frac{k}{M}}$ denote the price of one share of stock $i$ at time $t + \frac{k}{M}$, and $I_{t+\frac{k}{M}}$ denote the value (price) of the index at time $t + \frac{k}{M}$. 
Let $x_{i,(t+\frac{k}{M})-}$ and $x_{i,t+\frac{k}{M}}$ (resp., $V_{i,(t+\frac{k}{M})-}$ and $V_{i,t+\frac{k}{M}}$) denote the number of shares (resp., dollar value) of stock $i$  in the tracking portfolio right before and right after the rebalancing time $t + \frac{k}{M}$, respectively. Hence, $V_{i,(t+\frac{k}{M})-} = p_{i,(t+\frac{k}{M})-} x_{i,(t+\frac{k}{M})-}$ and $V_{i,t+\frac{k}{M}} = p_{i,t+\frac{k}{M}} x_{i,t+\frac{k}{M}}$ for $i = 1, \ldots, N$.
Let $H_{(t+\frac{k}{M})-}$ and $H_{t+\frac{k}{M}}$ denote the value of cash held in the tracking portfolio right before and after the rebalancing time $t+\frac{k}{M}$. Let $V_{(t+\frac{k}{M})-}$ and 
$V_{t+\frac{k}{M}}$ denote the value of tracking portfolio right before and after time $t+\frac{k}{M}$.  
Hence, 
$V_{(t+\frac{k}{M})-} = \sum_{i=1}^{N} p_{i,t+\frac{k}{M}} x_{i,(t+\frac{k}{M})-} + H_{(t+\frac{k}{M})-}$ and 
$V_{t+\frac{k}{M}} = \sum_{i=1}^{N} p_{i,t+\frac{k}{M}} x_{i,t+\frac{k}{M}}+H_{t+\frac{k}{M}}$. In general, $V_{(t+\frac{k}{M})-} \neq V_{(t+\frac{k}{M})}$ whenever rebalancing is carried out because of the transaction cost. 
The tracking portfolio weights right before and after the rebalancing at time $t + \frac{k}{M}$ are denoted as $w_{(t+\frac{k}{M})-} = \left(w_{1, (t+\frac{k}{M})-}, \dots, w_{N, (t+\frac{k}{M})-}\right)$ and $w_{t+\frac{k}{M}} = \left(w_{1, t+\frac{k}{M}}, \dots, w_{N, t+\frac{k}{M}}\right)$, respectively, where $w_{i, (t+\frac{k}{M})-}$ and $w_{i,t+\frac{k}{M}}$ are defined as
$w_{i, (t+\frac{k}{M})-} = \frac{V_{i,(t+\frac{k}{M})-}}{V_{(t+\frac{k}{M})-}}$ and $w_{i,t+\frac{k}{M}} = \frac{V_{i,t+\frac{k}{M}}}{V_{t+\frac{k}{M}}}$.

Finally, $r_{t+\frac{k+1}{M}}^{\text{TP}} = \frac{V_{(t+\frac{k+1}{M})-}}{V_{(t+\frac{k}{M})-}} - 1$ and $r_{t+\frac{k+1}{M}}^{\text{I}} = \frac{I_{t+\frac{k+1}{M}}}{I_{t+\frac{k}{M}}} - 1$ are the simple return of the tracking portfolio and that of the stock index during the single rebalancing period from $t+\frac{k}{M}$ to $t+\frac{k+1}{M}$, respectively. Note that $r_{t+\frac{k+1}{M}}^{\text{TP}}$ is calculated in terms of the value of tracking portfolio right before the rebalancing times, so it is affected by the transaction cost of rebalancing at time $t+\frac{k}{M}$ but not $t+\frac{k+1}{M}$. Similar definition of this before-rebalancing return can be found, for example, in \citet{chiam2013dynamic} and \citet{gaivoronski2005optimal}.


Let the $\sigma$-field $\mathcal{F}_t$ represent the set of market information available at time $t$, i.e., 
$\mathcal{F}_t=\sigma\{I_s, p_{1, s}, p_{2, s}, \ldots, p_{N, s}, X_s; \forall s\leq t\}$, where $X_s$ represents market information other than prices at time $s$. 
The fund manager's trading strategy is adapted to the 
filtration 
$(\mathcal{F}_t)_{t \ge 0}$.  
We assume that at each time $t=0, 1, 2, \ldots$, the fund manager selects portfolio weights $w_{t}\in \mathcal{F}_t$, and then rebalances the portfolio at each rebalancing time $t+\frac{k}{M}$, $k=0, 1, \ldots, M-1$, such that $w_{t+\frac{k}{M}}=w_t$ for all $k=0, 1, \ldots, M-1$.  In other words, we assume that the portfolio weights right after each rebalancing time are fixed to be $w_t$ during the whole period $[t, t+1)$.

For the sole purpose of better comparison with existing methods in the literature, we adopt the assumption that portfolio weights are determined at time $t$ and the portfolio is rebalanced to keep the same weights at each rebalancing time during $[t, t+1)$. This assumption is adopted by many existing index tracking methods which mainly use portfolio weights as controls \citep[see, e.g.,][]{jansen2002optimal, mutunge2018minimizing, benidis2017sparse}. The assumption is necessary for the formulation of the index tracking problem in these literature; however, it is not necessary in our method. Our method can accommodate any other assumptions such as a fixed number of shares for each stock is determined at time $t$ and then no rebalancing is performed during $[t, t+1)$ for each $t$ (see, Section \ref{subsec:fis_form}), which is typically adopted by models that make decisions on portfolio shares rather than weights \citep[see, e.g.,][]{beasley2003evolutionary, canakgoz2009mixed, strub2018optimal}. Another example is that \citet{yao2006tracking} track financial indices to minimize the value-tracking error via directly adjusting the wealth invested in each stock from time to time. In this case, they neither fix weights nor shares after rebalancing. Our method, albeit can similarly set stock values as controls, still makes decisions upon weights when tracking in values, as we shall discuss in Section \ref{subsec:vb_form}.

\subsection{Formulation Based on Return Tracking Error}\label{subsec:formulation_return_based}
In this formulation, we measure index tracking error as the difference between portfolio returns and index returns. Generally, the return-based tracking error over period $(t_1, t_2]$, with $t_1, t_2 \in \left\{t + \frac{k}{M}: t = 0, 1, 2, \dots, \ k = 0, \dots, M - 1 \right\}$ and $t_1 < t_2$, is defined as
\begin{equation}\label{eq:te}
	\textbf{R-TE}^{q}_{(t_1, t_2]} = \left[ \frac{1}{M (t_2 - t_1)} \sum_{k=1}^{M (t_2 - t_1)}  \left| r^{\text{TP}}_{t_1+\frac{k}{M}} - r^{\text{I}}_{t_1+\frac{k}{M}} \right|^q \right]^{\frac{1}{q}},
\end{equation}
where the parameter $q > 0$ is typically chosen to be 1 or 2. 

At time 0, the fund manager holds an initial cash position $H_{0-}>0$ and shares of stocks $x_{0-}$. Then, 
at the beginning of each time period $t=0, 1, \ldots$ (i.e., at each time $t=0,1,\ldots$), the fund manager chooses the after-rebalancing portfolio weights $w_{t} \in \mathcal{F}_t$ so as to minimize the cumulative discounted tracking error. As cash generates zero return, we assume that 
\begin{equation}\label{equ:weights_sum_1}
	\sum_{i=1}^N w_{i,t} = 1,\ \text{for any}\ t=0,1,\ldots, 
\end{equation}
namely, after each rebalancing, the tracking portfolio only holds stocks. This implies that 
\begin{equation}
	H_{\tk}=0,\ \text{and}\ H_{(t+\frac{k+1}{M})-}=0,\ \text{for any}\ t=0,1,\ldots,k=0,1,\ldots,M-1.
\end{equation}
More precisely, the fund manager's problem is formulated as
{\allowdisplaybreaks
	\begin{align}\label{eq:cumulative_discounted_tracking_error}
		\mathop{\min}_{\left\{ w_{t} \in \mathcal{F}_t: \ t = 0, 1, \dots \right\}}\ & E \left[ \sum_{t = 0}^{\infty} \gamma^{t} \textbf{R-TE}^{q}_{(t, t+1]}\right],\\
		\text{s.t.}\quad\quad\ \   & w_{i,t} \geq 0, i=1,\ldots,N, \sum_{i=1}^N w_{i,t} = 1, t=0,1,\ldots,\notag\\
		& V_{(t+\frac{k}{M})-} - c_{t+\frac{k}{M}} = V_{t+\frac{k}{M}}, t=0,1,\ldots, k=0, \ldots, M-1,\label{eq:rebalance_eq}\\
		& x_{i,t+\frac{k}{M}}p_{i,t+\frac{k}{M}} = w_{i,t+\frac{k}{M}}V_{t+\frac{k}{M}}, i=1,\ldots,N, t=0,1,\ldots, k=0, \ldots, M-1,\label{eq:x_equal_wV}\\		
		& w_{i,t+\frac{k}{M}} = w_{i,t}, i=1,\ldots,N, t=0,1,\ldots, k=0, \ldots, M-1, \label{eq:weight_const_constraint}
\end{align}}%
where $c_{t+\frac{k}{M}}$ is the transaction cost at time $t+\frac{k}{M}$ due to rebalancing the portfolio from number of shares $x_{(\tk)-}$ to $x_{\tk }$; $\gamma \in (0, 1)$ is a discount factor. In general, the transaction cost $c_{t+\frac{k}{M}}$ is a nonlinear function of $x_{i,(t+\frac{k}{M})-}, V_{i,(t+\frac{k}{M})-}, x_{i,t+\frac{k}{M}}$ and $V_{i,t+\frac{k}{M}}$, $i=1,\ldots, N$, which is specified by the broker. Equation \eqref{eq:rebalance_eq} is the rebalance equation (see Lemma \ref{lemma:portfolio_value_equaility} in the Appendix \ref{appendix:proof_prop_achieve_target_weight} for a formal proof). Equation \eqref{eq:x_equal_wV} holds as both sides of the equation are equal to $V_{i,\tk}$, the dollar value of the $i$th stock in the portfolio right after the rebalancing. 
Equation \eqref{eq:weight_const_constraint} is the constraint that portfolio weights right after each rebalancing time $t+\frac{k}{M}$ should be equal to $w_t$, which is decided at time $t$. Note that Equation \eqref{eq:rebalance_eq} holds at time 0, so the transaction cost at time 0 is taken into account in the problem. 

The timing of events is as follows. Right before each rebalancing time $t+\frac{k}{M}$, the fund manager holds $x_{i,(t+\frac{k}{M})-}$ shares of stock $i$ in the tracking portfolio, and he or she observes the stock prices $\left(p_{1, t+\frac{k}{M}}, \dots, p_{N, t+\frac{k}{M}} \right)$ at time $t+\frac{k}{M}$, the portfolio value $V_{(t+\frac{k}{M})-}$ before rebalancing, and the portfolio weights $w_{(t+\frac{k}{M})-}$. He or she needs to find the number of shares $x_{i,t+\frac{k}{M}}$ of stock $i$ to be held right after the rebalancing in order to satisfy the rebalancing constraint \eqref{eq:weight_const_constraint}, i.e., portfolio weights right after the rebalancing $w_{t+\frac{k}{M}}$ are equal to $w_t$ that are decided at time $t$. The number of shares 
$x_{i,t+\frac{k}{M}}$, $i=1,\ldots, N$ and the portfolio value right after rebalancing $V_{t+\frac{k}{M}}$ satisfy the nonlinear rebalancing equations \eqref{eq:rebalance_eq}-\eqref{eq:x_equal_wV}, 
where $c_{t+\frac{k}{M}}$ is a nonlinear function of the unknown $x_{i,t+\frac{k}{M}}$, $i=1, \ldots, N$ and $V_{t+\frac{k}{M}}$. We will show in Section \ref{subsec:transaction_cost}
that for any given $w_{i,t+\frac{k}{M}},i=1,\ldots,N$,  
the system of equations has a unique solution for the unknown variables $x_{i,t+\frac{k}{M}}$, $i=1,\ldots, N$ and $V_{t+\frac{k}{M}}$, and can be solved accurately and efficiently by a simple fixed point iteration algorithm based on the Banach Fixed Point Theorem. After the unknown variables $x_{i,t+\frac{k}{M}}$, $i=1,\ldots, N$ and $V_{t+\frac{k}{M}}$ are solved, the exact transaction cost $c_{t+\frac{k}{M}}$ as a function of these variables will be obtained.

\subsection{Formulation Based on Value Tracking Error}\label{subsec:vb_form}
In this formulation, we measure tracking error as the distance between the trajectories of dollar values of the tracking portfolio and index. Unlike the above return-based tracking model that targets for investors who care about daily return, the value-based model in this section can be more attractive to investors with long investment horizon and more focusing on cumulative return. Formally, the value-based tracking error over period $(t_1, t_2]$, with $t_1, t_2 \in \left\{t + \frac{k}{M}: t = 0, 1, 2, \dots, \ k = 0, \dots, M - 1 \right\}$ and $t_1 < t_2$, is
\begin{equation}\label{eq:te_value_based}
	\textbf{V-TE}^{q}_{(t_1, t_2]} = \left[ \frac{1}{M (t_2 - t_1)} \sum_{k=1}^{M (t_2 - t_1)}  \left| \frac{V_{(t_1+\frac{k}{M})-}}{N_{0}} - I_{t_1+\frac{k}{M}} \right|^q \right]^{\frac{1}{q}},
\end{equation}
where $N_{0} = \frac{V_{0-}}{I_{0}}$ represents total number of shares of the index fund (i.e., tracking portfolio) that is sold to investors at time 0, with the value of each share expected to be close to the index value at all time.  
The quantity $\frac{V_{(t_1+\frac{k}{M})-}}{N_{0}}$ is the value of one share of the index fund before rebalancing at time $t_1+\frac{k}{M}$, which should be close to the index value $I_{t_1+\frac{k}{M}}$.

In this formulation, the fund manager may choose to 
withdraw or inject certain amount of cash to the tracking portfolio at each portfolio rebalancing time $t+\frac{k}{M}$. 
The amount of cash is defined to be a proportion of $I_{t+\frac{k}{M}} N_0 - V_{(t+\frac{k}{M})-}$, the amount of deviation of the value of tracking portfolio to the index value multiplied by $N_0$. More precisely, at each time $t$ (i.e., the beginning of the time period $[t, t+1)$), the fund manager chooses portfolio weights $w_t$ and a decision variable $f_t$, and then at each rebalancing time $t + \frac{k}{M}, k=0, 1, \ldots, M-1$, 
the fund manager injects cash 
\begin{equation}\label{eq:cash_fn}
	h_{t+\frac{k}{M}} = \mathop{\max} \left\{\left(I_{t+\frac{k}{M}} N_0 - V_{(t+\frac{k}{M})-}\right) f_t, \ -\xi V_{(t+\frac{k}{M})-} \right\}
\end{equation}
to the tracking portfolio, and rebalances the portfolio so that portfolio weights after the rebalancing become $w_t$. Here $\xi \in (0, 1)$ is a constant. 

By the definition in \eqref{eq:cash_fn}, at the rebalancing time $t+\frac{k}{M}=0$, $h_0=0$. At each other rebalancing time $t+\frac{k}{M}>0$, if $I_{t+\frac{k}{M}} N_0 \ge V_{(t+\frac{k}{M})-}$, i.e., 
the fund value falls short of the index value multiplied by $N_0$, 
the fund manager injects cash $h_{t+\frac{k}{M}}$ to the tracking portfolio. In this case, the fund manager turns the portfolio with value $V_{(t+\frac{k}{M})-}$ and the additional cash $h_{t+\frac{k}{M}}$ into a new portfolio with value $V_{t+\frac{k}{M}}$ through stock transactions. If $I_{t+\frac{k}{M}} N_0 < V_{(t+\frac{k}{M})-}$, i.e., 
the fund value exceeds the index value multiplied by $N_0$, the fund manager injects a negative amount of cash $h_{t+\frac{k}{M}}$ into the tracking portfolio, i.e., he or she withdraws cash 
$-h_{t+\frac{k}{M}}$, which has 
an upper bound $\xi V_{(t+\frac{k}{M})-}$, from the tracking portfolio. In this case, the fund manager turns the portfolio with value $V_{(t+\frac{k}{M})-}$ into a new portfolio with value $V_{t+\frac{k}{M}}$ and cash $-h_{t+\frac{k}{M}}$ by stock transactions. 


The profit of the fund manager over the period $[t, t+1)$ is the management fee collected from the investors subtracting the total cash injected into the tracking portfolio over the period, $\sum_{k=0}^{M-1} h_{t+\frac{k}{M}}$,\footnote{The interest rate adjusted total cash injection is $\sum_{k=0}^{M-1} \left[\prod_{j=k}^{M-1}(1+r_{t+\frac{j}{M}})\right] h_{t+\frac{k}{M}} $, where $r_{t+\frac{j}{M}}$ is the simple interest rate from $t+\frac{j}{M}$ to $t+\frac{j+1}{M}$.\label{fn:interest_rate_adj_cash}} and other costs. 

%

Starting from time 0, the fund manager selects portfolio weights and cash flow decision variables to solve the following problem
{\allowdisplaybreaks
	\begin{align}\label{eq:cumulative_discounted_tracking_error_vb}
		\mathop{\min}_{\left\{w_{t}, f_t \in \mathcal{F}_t: \ t = 0, 1, \dots \right\}}\ & E \left[\sum_{t = 0}^{\infty} \gamma^{t} \textbf{V-TE}^{q}_{(t, t+1]}\right]\\
		\text{s.t.}\quad\quad\ \ \    & w_{i,t} \geq 0, i=1,\ldots,N, \sum_{i=1}^N w_{i,t} = 1, t=0,1,\ldots,\notag\\	
		& 0\leq  f_t \leq b_f,\label{eq:constraint_proportion}	\\
		& 	h_{t+\frac{k}{M}} = \mathop{\max} \left\{\left(I_{t+\frac{k}{M}} N_0 - V_{(t+\frac{k}{M})-}\right) f_t, \ -\xi V_{(t+\frac{k}{M})-} \right\}, \label{eq:cash_flow_formula}\\
		&\quad\quad\quad\quad  t=0,1,\ldots, k=0, \ldots, M-1,\notag\\
		& V_{(t+\frac{k}{M})-} - c_{t+\frac{k}{M}} + h_{t+\frac{k}{M}} = V_{t+\frac{k}{M}}, t=0,1,\ldots, k=0, \ldots, M-1, \label{eq:rebalance_eq_with_cash}\\
		& x_{i,t+\frac{k}{M}}p_{i,t+\frac{k}{M}} = w_{i,t+\frac{k}{M}}V_{t+\frac{k}{M}}, i=1,\ldots,N, t=0,1,\ldots, k=0, \ldots, M-1,\label{eq:x_equal_wV_cash}\\		
		& w_{i,t+\frac{k}{M}} = w_{i,t}, i=1,\ldots,N, t=0,1,\ldots, k=0, \ldots, M-1. \label{eq:weight_const_constraint_with_cash}	
\end{align}}

In this formulation, $b_f$ in Equation \eqref{eq:constraint_proportion} is a nonnegative constant representing an upper bound of the decision variable $f_t$. The rebalancing equation \eqref{eq:rebalance_eq_with_cash} has an additional term $h_{t+\frac{k}{M}}$ compared with \eqref{eq:rebalance_eq}; see Lemma \ref{lemma:portfolio_value_equaility} in the Appendix \ref{appendix:proof_prop_achieve_target_weight} for a formal proof of the equation. In particular, the equation holds at time 0, so the transaction cost at time 0 is taken into account in the problem. 
Equation \eqref{eq:weight_const_constraint_with_cash} imposes the constraint that the portfolio weights right after each rebalancing time $w_{i,\tk}$ should be equal to $w_{i,t}$. 
Similar to problem \eqref{eq:cumulative_discounted_tracking_error}, at each rebalancing time $t+\frac{k}{M}$, the fund manager needs to find the number of shares $x_{i,t+\frac{k}{M}}$ of stock $i$ to be held right after the rebalancing and the portfolio value $V_{t+\frac{k}{M}}$ right after rebalancing by solving the rebalancing equations \eqref{eq:rebalance_eq_with_cash}-\eqref{eq:x_equal_wV_cash}.
We will show in Section \ref{subsec:transaction_cost} that for any given 
$w_{i,t+\frac{k}{M}},i=1,\ldots,N$, 
the rebalancing equations \eqref{eq:rebalance_eq_with_cash}-\eqref{eq:x_equal_wV_cash} has a unique solution for the unknown variables $x_{i,t+\frac{k}{M}}$, $i=1,\ldots, N$ and $V_{t+\frac{k}{M}}$, and can be solved efficiently by a simple fixed point iteration algorithm.

The allowance of cash flow at the rebalancing times generalizes the case of no cash flow. In fact, when the constant $b_f$ in Equation \ref{eq:constraint_proportion} is defined to be 0, the cash flow $h_{t+\frac{k}{M}}$ is always 0, reducing to the case of no cash flow is allowed.

\subsection{Formulation for Other Rebalancing Frequency}\label{subsec:fis_form}

Although the problem formulation in Section \ref{subsec:formulation_return_based} and Section \ref{subsec:vb_form} assumes that the portfolio is rebalanced every trading day to keep the portfolio weights the same as the weights determined at the beginning of each time period, the formulation is flexible to incorporate general  rebalancing frequency. More precisely, suppose we would like to formulate that the portfolio is rebalanced every $n_b$ trading days during each time period $[t, t+1)$. This can be achieved by requiring Equation \eqref{eq:rebalance_eq}-\eqref{eq:weight_const_constraint} and \eqref{eq:cash_flow_formula}-\eqref{eq:weight_const_constraint_with_cash} to hold only for $k=0,n_b,2n_b, \ldots, M-n_b$.
In particular, by specifying $n_b=M$, we obtain the formulation that during each time period $[t, t+1)$, the portfolio is rebalanced only once at time $t$.


\subsection{Extension for Index Enhancement Strategies}\label{subsec:index_enhancement}

We can extend the formulations of index tracking problems in \eqref{eq:cumulative_discounted_tracking_error} and \eqref{eq:cumulative_discounted_tracking_error_vb} for index enhancement problems by replacing $\textbf{V-TE}^{q}_{(t_1, t_2]}$ or $\textbf{R-TE}^{q}_{(t_1, t_2]}$ with other measures that only penalize under-performance of the tracking portfolio with respect to the index. For example, we can replace the absolute value function $|\cdot|$ in \eqref{eq:te} and \eqref{eq:te_value_based} by the negative part function $x^-=\max(0, -x)$, $x\in\mathbb{R}$. The extended formulation for index enhancement can also be solved by the reinforcement learning method proposed in Section \ref{sec:RL_for_IT}.

\section{A Reinforcement Learning Approach for Index Tracking}\label{sec:RL_for_IT}

In this section, we propose a reinforcement learning method for solving the two index tracking problems in \eqref{eq:cumulative_discounted_tracking_error} and \eqref{eq:cumulative_discounted_tracking_error_vb}. First, we propose a Banach fixed point iteration method for solving the portfolio rebalancing equations \eqref{eq:rebalance_eq}-\eqref{eq:x_equal_wV} and \eqref{eq:rebalance_eq_with_cash}-\eqref{eq:x_equal_wV_cash} respectively for the return-based problem and the value-based problem. This method can also be applied in the portfolio rebalancing involving transaction costs for other portfolio selection problems. Then, we propose the reinforcement learning method for solving the index tracking problems. 

\subsection{A Fixed Point Iteration Algorithm for Solving the Rebalancing Equation}\label{subsec:transaction_cost}

In this subsection, we propose a Banach fixed point algorithm for 
solving the portfolio rebalancing equations \eqref{eq:rebalance_eq}-\eqref{eq:x_equal_wV} and \eqref{eq:rebalance_eq_with_cash}-\eqref{eq:x_equal_wV_cash} at each rebalancing time $t+\frac{k}{M}$. Recall that at the beginning of each time period $[t, t+1)$ (i.e., at time $t$), we determine the portfolio weights $w_t$ and $f_t$ (for the value-based tracking problem), then at each trading day $t+\frac{k}{M}$, $k=0,1,\ldots,M-1$, we
inject cash $h_{t+\frac{k}{M}}$ defined by \eqref{eq:cash_fn} into the tracking portfolio, and
rebalance the portfolio so that the portfolio weights after rebalancing become $w_t$, i.e., $w_{t+\frac{k}{M}}=w_t$. 

Right before rebalancing at time $t+\frac{k}{M}$, we observe the stock prices $\left(p_{1, t+\frac{k}{M}}, \dots, p_{N, t+\frac{k}{M}} \right)$, the stock shares $\left( x_{1, (t+\frac{k}{M})-}, \dots, x_{N, (t+\frac{k}{M})-} \right)$, the rebalancing target portfolio weights (i.e., the portfolio weights right after rebalancing) $w_{t+\frac{k}{M}}=w_t$, and the cash injection $h_{t+\frac{k}{M}}$. Hence, in the portfolio rebalancing equations  \eqref{eq:rebalance_eq}-\eqref{eq:x_equal_wV} and \eqref{eq:rebalance_eq_with_cash}-\eqref{eq:x_equal_wV_cash}, the unknown variables are $V_{t+\frac{k}{M}}$ and $x_{i,t+\frac{k}{M}}$, $i=1,\ldots, N$. 

In the rebalancing equations, the transaction cost $c_{t+\frac{k}{M}}$ is a given function of observable variables and the unknown variables, which is specified by stock brokerage firms. Following  \citet[][pp.~17-18]{benidis2018optimization}, we specify the transaction cost incurred at time $t+\frac{k}{M}$, $t = 0, 1, 2, \dots, k = 0, \dots, M-1$, as
\begin{align}\label{eq:tc}
	&c_{t+\frac{k}{M}} = \sum_{i=1}^{N} \mathop{\min} \left \{ \mathop{\max} \left \{ \xi_{1i} \left|x_{i,(t+\frac{k}{M})-} - x_{i,t+\frac{k}{M}}\right|, \; \xi_{3i} \right \},\; \xi_{2i} \left|V_{i,(t+\frac{k}{M})-} - V_{i, t+\frac{k}{M}} \right| \right \},
\end{align}
where $\xi_{1i}$ is transaction fee per share for stock $i$; $\xi_{2i}>0$ is a proportion of the trade value; $\xi_{3i}>0$ is the minimum transaction fee of the trade for stock $i$. In words, the transaction cost is $\xi_{1i}$ per share, with a minimum of $\xi_{3i}$ and a maximum of $\xi_{2i}$ of the trade value. 
The transaction cost in Equation \eqref{eq:tc} is adopted by Interactive Brokers LLC (IB), one of the largest brokerage firm in the US that serves both retail and institutional clients. IB specifies $\xi_{1i} = 0.005$ USD, $\xi_{2i} = 0.5\%$, and $\xi_{3i} = 1.0$ USD for investors of the IBKR Pro-Fixed category.\footnote{ \url{https://www.interactivebrokers.com/en/index.php?f=49637}.} These constants may be specified to be smaller for investors with higher trading volume.

The transaction cost in the US market also includes SEC transaction fee and FINRA trading activity fee, which charge respectively based on the value and shares of stocks that are sold. 
Our method of computing the exact transaction cost and portfolio value after rebalancing will still work if these regulatory  fees are added to the transaction cost in Equation \eqref{eq:tc}; see the transaction cost specified in Equation \eqref{eq:tc_with_reg_fee} and Proposition \ref{prop:regulatory_fees} in Appendix \ref{appendix:regulatory_fees} for details. 
For simplicity of exposition, Equation \eqref{eq:tc} will be our focus in the remainder of the paper.

The transaction fee specified in Equation \eqref{eq:tc} (and in Equation \eqref{eq:tc_with_reg_fee} in the Appendix) is more general and realistic than the simple form 
of transaction cost specified as a fixed cost plus a proportion of the transaction value, which is adopted in most existing studies \citep[see, e.g.,][]{gaivoronski2005optimal, canakgoz2009mixed, strub2018optimal}. Our method also works for this simple form of transaction cost; see Proposition \ref{prop:fixed_proportional_cost} in Appendix \ref{appendix:fixed_proportional_cost}.

By using Equation \eqref{eq:x_equal_wV} to eliminate variables $x_{i, t+\frac{k}{M}}$ in Equation \eqref{eq:tc}, we obtain
\begin{align}\label{eq:tc_as_func_V}
	c_{t+\frac{k}{M}} & = \sum_{i=1}^{N} \mathop{\min} \left \{ \mathop{\max} \left \{ \frac{\xi_{1i}}{p_{i,t+\frac{k}{M}}} \left|V_{i,(t+\frac{k}{M})-} - w_{i,t+\frac{k}{M}} V_{t+\frac{k}{M}} \right|, \; \xi_{3i} \right \}, \xi_{2i} \left|V_{i,(t+\frac{k}{M})-} - w_{i,t+\frac{k}{M}} V_{t+\frac{k}{M}} \right| \right \}.
\end{align}

At the portfolio rebalancing time $t+\frac{k}{M}$,  $V_{i,(t+\frac{k}{M})-}$, $p_{i,t+\frac{k}{M}}$, and $w_{i,t+\frac{k}{M}}$ are all given; hence, $c_{t+\frac{k}{M}}$ is a function of the unknown variable $V_{t+\frac{k}{M}}$, denoted as $c_{t+\frac{k}{M}}(V_{t+\frac{k}{M}})$. Then, the rebalancing equation \eqref{eq:rebalance_eq} becomes 
\begin{equation}\label{equ:V_fixed_point_prob}
	V_{(t+\frac{k}{M})-} - c_{t+\frac{k}{M}}(V_{t+\frac{k}{M}}) = V_{t+\frac{k}{M}},
\end{equation}
and for the value-tracking problem, the rebalancing equation \eqref{eq:rebalance_eq_with_cash} becomes
\begin{equation}\label{equ:V_fixed_point_prob_cash_injection}
	V_{(t+\frac{k}{M})-} - c_{t+\frac{k}{M}}(V_{t+\frac{k}{M}}) + h_{t+\frac{k}{M}} = V_{t+\frac{k}{M}}.
\end{equation}


%

The following proposition shows that under mild conditions that hold in practice, either of the rebalancing equations \eqref{equ:V_fixed_point_prob} and \eqref{equ:V_fixed_point_prob_cash_injection} has a unique solution $V_{t+\frac{k}{M}}$, which can be solved by a Banach fixed point iteration algorithm. Then, the exact transaction cost $c_{t+\frac{k}{M}}(V_{\tk})$ can be calculated, and the stock holding after rebalancing $x_{i,t+\frac{k}{M}}$ is uniquely determined by Equation \eqref{eq:x_equal_wV} and \eqref{eq:x_equal_wV_cash} respectively for the return-based and value-based tracking problems. Therefore, the number of shares of each stock to be traded at the rebalancing time is also uniquely determined. 

\begin{proposition}\label{prop:achieve_target_weight}
	($i$) If $V_{(t+\frac{k}{M})-} > 0$, $\sum_{i=1}^{N} \xi_{2i} | w_{i, (t+\frac{k}{M})-} | < 1$, and $\sum_{i=1}^{N} (\frac{\xi_{1i}}{p_{i, t+\frac{k}{M}}} + \xi_{2i})| w_{i, t+\frac{k}{M}} | < 1$, then 
	the function on the left-hand side of Equation \eqref{equ:V_fixed_point_prob} is a contraction on $\mathbb{R}$ with 
	the contraction coefficient $\sum_{i=1}^{N} ( \frac{\xi_{1i}}{p_{i,\tk}} + \xi_{2i} ) |w_{i,\tk}|$, and 
	the rebalancing equation \eqref{equ:V_fixed_point_prob} has a unique solution $V^*_{t+\frac{k}{M}}$  on $\mathbb{R}$. In addition, the solution satisfies $0<V^*_{t+\frac{k}{M}} \leq V_{(t+\frac{k}{M})-}$.
	
	($ii$) If all the conditions in (i) hold, and $0 < \xi < 1 - \sum_{i=1}^{N} \xi_{2i} | w_{i,(t+\frac{k}{M})-} |$, and  
	$h_{t+\frac{k}{M}} \ge \ -\xi V_{(t+\frac{k}{M})-}$, then the function on the left-hand side of Equation \eqref{equ:V_fixed_point_prob_cash_injection} is a contraction on $\mathbb{R}$ with 
	the contraction coefficient $\sum_{i=1}^{N} ( \frac{\xi_{1i}}{p_{i,\tk}} + \xi_{2i} ) |w_{i,\tk}|$, and the rebalancing equation \eqref{equ:V_fixed_point_prob_cash_injection} has a unique solution $V^*_{t+\frac{k}{M}}$ on $\mathbb{R}$. In addition, the solution satisfies $0<V^*_{t+\frac{k}{M}}\leq V_{(t+\frac{k}{M})-} + h_{t+\frac{k}{M}}$. 
\end{proposition}

\begin{proof}
See Appendix \ref{appendix:proof_prop_achieve_target_weight}.
\end{proof}

\begin{remark}
	The proposition can be applied in the case of short selling as it allows general portfolio weights in which some  $w_{i,(\tk)-}$ or $w_{i,\tk}$ can be negative.
\end{remark}

\begin{remark}\label{rmk:relation_of_two_weights}
	Due to price movement, the weight $w_{i, (t+\frac{k}{M})-}$  is usually not equal to $w_{i, t+\frac{k-1}{M}}$ for $t + \frac{k}{M} > 0$. To calculate $w_{i, (t+\frac{k}{M})-}$ to verify the condition of the proposition, one can use the following simple formula (see Lemma \ref{lemma:relation_of_weights} in Appendix \ref{appendix:proof_prop_achieve_target_weight})
	\begin{align}
		& w_{i, (t+\frac{k}{M})-} = \frac{w_{i, t+\frac{k-1}{M}} p_{i,t+\frac{k}{M}}}{p_{i,t+\frac{k-1}{M}} \sum_{j=1}^{N} w_{j,t+\frac{k-1}{M}} \frac{p_{j,t+\frac{k}{M}}}{p_{j,t+\frac{k-1}{M}}}},\ \text{for any}\ t+\frac{k}{M}>0. \notag
	\end{align}
\end{remark}

\begin{remark}
	The conditions of the proposition depend on the specific functional form of the transaction cost. Proposition \ref{prop:regulatory_fees} in the Appendix \ref{appendix:regulatory_fees} provides the conditions when the transaction cost includes regulatory fees. Proposition \ref{prop:fixed_proportional_cost} in the Appendix \ref{appendix:fixed_proportional_cost} provides the conditions when the transaction cost is specified as proportional cost plus a fixed cost.
\end{remark}

The conditions in Proposition \ref{prop:achieve_target_weight} typically hold in practice, as $\xi_{1i}$ and $\xi_{2i}$ are in the order of 0.001, $\sum_{i=1}^N |w_{i,\tk}|$ and $\sum_{i=1}^N |w_{i,(\tk)-}|$ are in the order of 1, and usually $p_{i,\tk}>1$ for all $i$. Therefore, at any time $\tk$, one can rebalance the portfolio so that the portfolio weights after rebalancing are equal to the target portfolio weights and the transaction cost is exactly taken into account in the rebalancing. 

By the contraction property of the function on the left-hand side of Equation \eqref{equ:V_fixed_point_prob} and \eqref{equ:V_fixed_point_prob_cash_injection}, the solution $V_{t+\frac{k}{M}}$ to the equations can be computed by a Banach fixed point iteration algorithm. More precisely, Algorithm \ref{algo1} and Algorithm \ref{algo2} show the algorithm for solving Equation \eqref{equ:V_fixed_point_prob} and \eqref{equ:V_fixed_point_prob_cash_injection}, respectively. Under the conditions of Proposition \ref{prop:achieve_target_weight}, the algorithms are guaranteed to converge by the Banach Fixed Point Theorem. Furthermore, the contraction coefficient $\sum_{i=1}^{N} \left( \frac{\xi_{1i}}{p_{i,\tk}} + \xi_{2i} \right ) |w_{i,\tk}|$  is in the order of 0.01 if stock price $p_i > 1$ for all $i = 1, \ldots, N$, so the loops in Algorithm \ref{algo1} and \ref{algo2} end after just a few iterations.

After the solution $V_{t+\frac{k}{M}}$ to the rebalancing equation is computed, the number of shares of stock $i$ that should be held after rebalancing is computed as $x_{i, t+\frac{k}{M}} = \frac{w_{i, t+\frac{k}{M}} V_{t+\frac{k}{M}}}{p_{i, t+\frac{k}{M}}}$, $i=1, \dots, N$. 


\begin{algorithm}[htbp]
	\caption{A fixed point iteration algorithm for solving Equation \eqref{equ:V_fixed_point_prob}\label{algo1}} 
	\begin{algorithmic}[1] \label{code:train_alg_rb}
		\STATE Set convergence tolerance $\varepsilon=\text{1E-15}$
		\STATE Initialize $V^{\text{old}}_{t+\frac{k}{M}} \leftarrow 0$ and $V^{\text{new}}_{t+\frac{k}{M}} \leftarrow V_{(t+\frac{k}{M})-} - c_{t+\frac{k}{M}}(V^{\text{old}}_{t+\frac{k}{M}})$ 
		\WHILE {$\left| V^{\text{new}}_{t+\frac{k}{M}} - V^{\text{old}}_{t+\frac{k}{M}} \right| \ge \varepsilon$}
		\STATE $V^{\text{old}}_{t+\frac{k}{M}} \leftarrow V^{\text{new}}_{t+\frac{k}{M}}$
		\STATE $V^{\text{new}}_{t+\frac{k}{M}} \leftarrow V_{(t+\frac{k}{M})-} - c_{t+\frac{k}{M}}(V^{\text{old}}_{t+\frac{k}{M}})$
		\ENDWHILE
		\RETURN $V^{\text{new}}_{t+\frac{k}{M}}$
	\end{algorithmic}
\end{algorithm}

\begin{algorithm}[htbp]
	\caption{A fixed point iteration algorithm for solving Equation \eqref{equ:V_fixed_point_prob_cash_injection}\label{algo2}}  
	\begin{algorithmic}[1] \label{code:train_alg_vb}
		\STATE Set convergence tolerance $\varepsilon=\text{1E-15}$
		\STATE Initialize $V^{\text{old}}_{t+\frac{k}{M}} \leftarrow 0$ and $V^{\text{new}}_{t+\frac{k}{M}} \leftarrow V_{(t+\frac{k}{M})-} + h_{t+\frac{k}{M}} - c_{t+\frac{k}{M}}(V^{\text{old}}_{t+\frac{k}{M}})$
		\WHILE {$\left| V^{\text{new}}_{t+\frac{k}{M}} - V^{\text{old}}_{t+\frac{k}{M}} \right| \ge \varepsilon$}
		\STATE $V^{\text{old}}_{t+\frac{k}{M}} \leftarrow V^{\text{new}}_{t+\frac{k}{M}}$
		\STATE $V^{\text{new}}_{t+\frac{k}{M}} \leftarrow V_{(t+\frac{k}{M})-} + h_{t+\frac{k}{M}} - c_{t+\frac{k}{M}}(V^{\text{old}}_{t+\frac{k}{M}})$
		\ENDWHILE
		\RETURN $V^{\text{new}}_{t+\frac{k}{M}}$
	\end{algorithmic}
\end{algorithm}

\subsection{A Reinforcement Learning (RL) Approach for Solving the Index Tracking Problems}\label{subsec:RL_problem_solution}

In this subsection, we first propose a RL approach for solving the dynamic formulations of the return tracking problem in \eqref{eq:cumulative_discounted_tracking_error} and the value tracking problem in \eqref{eq:cumulative_discounted_tracking_error_vb}. We then discuss the specification of state variables, actions, policies, and neutral networks in the RL method for the two problems.

The major difficulty of the RL approach for solving the index tracking problem is the limitation of available historical data for training. We propose a new training scheme to address this issue in Section \ref{subsec:new_training_scheme}.

\subsubsection{A RL Approach for Solving the Dynamic Formation of Index Tracking Problem}
In the RL framework, at each time period $t$, an agent observes the state of environment $s_t$, and then takes an action $a_t$ that follows a policy (i.e., distribution) $\pi_{\theta}(\cdot |s_t)$, where $\theta$ is the parameter of the policy. Then, the state of environment evolves to $s_{t+1}$ at time $t+1$, where $s_{t+1}$ follows the transition distribution $P(\cdot | s_t, a_t)$ determined by the environment; in addition, the agent receives a feedback from the environment in form of a reward  $r_{t} = R(s_t, a_t, s_{t+1})$. Then the cumulative discounted reward of the agent is 
\begin{equation}\label{equ:cumu_disc_reward}
	\eta(\theta) = E_{s_0, a_0, s_1, a_1, \dots} \left[ \sum_{t=0}^{\infty} \gamma^t r_{t} \right],\ \text{where}\ s_0\sim \rho_0, a_t\sim \pi_{\theta}(\cdot |s_t), s_{t+1}\sim P(\cdot \mid s_t, a_t), 	
\end{equation}
where $\rho_0$ is the distribution of the initial state $s_0$ at time $0$, $\gamma\in (0, 1)$ is the discount factor. The objective of the agent is to find an optimal policy $\pi_{\theta^*}$ that maximizes the cumulative discounted reward.

In the index tracking problem, the fund manager is the agent and the financial market is the environment. The return tracking problem \eqref{eq:cumulative_discounted_tracking_error} is equivalent to the objective of the agent who 
chooses an action $a_t$ to determine $w_t$ at time $t$ and then receives the reward 
\begin{align}\label{eq:reward_rb}
	r_{t} = R(s_t, a_{t}, s_{t+1}) = - \beta \textbf{R-TE}^{q}_{(t, t+1]},
\end{align}
where $\beta$ is a constant that rescales the tracking error to facilitate model training. The value tracking problem  \eqref{eq:cumulative_discounted_tracking_error_vb} is equivalent to the objective of the agent who chooses an action $a_t$ to determine $(w_t, f_t)$ at time period $t$ and then receives the reward
\begin{align}\label{eq:reward_vb}
	r_{t} = R(s_t, a_{t}, s_{t+1}) = - \beta \textbf{V-TE}^{q}_{(t, t+1]}.
\end{align}

The index tracking problems have high-dimensional continuous states and continuous actions. We propose to solve the index tracking problems by an extension of the proximal policy optimization (PPO) algorithm \citep{schulman2017proximal} due to its empirical stability and high performance. 
PPO is based on the Trust Region Policy Optimization (TRPO) algorithm for RL \citep{schulman2015trust}. \citet{Kakade-Langford-2002} consider the following function $L_{\theta}(\cdot)$ as an approximation to the RL objective function $\eta(\cdot)$: 
\begin{equation}
	L_{\theta}(\tilde{\theta}) = \eta(\theta) + \sum_{s} \rho_{\theta}(s) \sum_{a} \pi_{\tilde{\theta}} (a|s) A_{\theta}(s, a),\notag
\end{equation}
where $\rho_{\tilde \theta}(s) = P(s_0 = s | \pi_{\tilde \theta}) + \gamma P(s_1 = s | \pi_{\tilde \theta}) + \gamma^2 P(s_2 = s | \pi_{\tilde \theta}) + \cdots $ is the unnormalized discounted visitation frequency following $\pi_{\tilde \theta}$; $A_{\theta}(s, a) = Q_{\theta}(s, a) - V_{\theta}(s)$ is the  advantage function; $Q_{\theta}(s_t, a_t) = E_{s_{t+1}, a_{t+1}, \dots \sim \pi_{\theta}} \left[ \sum_{l = 0}^{\infty} \gamma^l r_{t+l} \mid s_t, a_t\right]$ is state-action value function; $V_{\theta}(s_t) = E_{a_{t}, s_{t+1}, \dots \sim \pi_{\theta}} \left[ \sum_{l = 0}^{\infty} \gamma^l r_{t+l} \mid s_t\right]$ is the value function. $L_{\theta}(\tilde \theta)$ matches $\eta(\tilde \theta)$ at $\tilde \theta=\theta$ to the first order. 

Based on the function $L_{\theta}(\cdot)$, \citet{schulman2015trust} propose to solve the following problem to generate a policy update at $\theta_{\text{old}}$:
\begin{align}\label{eq:eq_12}
	\mathop{\max}_{\theta}\ & L_{\theta_{\text{old}}}(\theta) \\ 
	\text{s.t.}\ & \bar{D}^{\rho_{\theta_{\text{old}}}}_{\text{KL}} \left( \theta_{\text{old}}, \theta \right) \le \delta,\notag 
\end{align}
where $\bar{D}^{\rho_{\theta_{\text{old}}}}_{\text{KL}} \left( \theta_{\text{old}}, \theta \right) = E_{s \sim \rho_{\theta_{\text{old}}}} \left[D_{\text{KL}}(\pi_{\theta_{\text{old}}}(\cdot\mid s) \parallel \pi_{\theta}(\cdot\mid s))\right]$ is the average KL divergence, $D_{\text{KL}}( \pi_{\theta_{\text{old}}}(\cdot\mid s) \parallel \pi_{\theta}(\cdot\mid s))$ is the KL divergence between $\pi_{\theta_{\text{old}}}(\cdot \mid s)$ and $\pi_{\theta}(\cdot \mid s)$, and $\delta$ is a parameter specifying an upper bound of the divergence. 

To solve the problem \eqref{eq:eq_12} using Monte Carlo simulation, \citet{schulman2015trust} reformulate the problem as:
\begin{align}\label{eq:eq_14}
	\mathop{\max}_{\theta}\ & E_{s \sim \rho_{\theta_{\text{old}}}, a \sim \pi_{\theta_{\text{old}}}} \left[\frac{\pi_{\theta}(a|s)}{\pi_{\theta_{\text{old}}}(a|s)} Q_{\theta_{\text{old}}}(s,a) \right] \notag \\ 
	\text{s.t.}\ & \bar{D}^{\rho_{\theta_{\text{old}}}}_{\text{KL}} \left( \theta_{\text{old}}, \theta \right) \le \delta,
\end{align}
which leads to the TRPO algorithm for RL. 

\citet{schulman2017proximal} propose the PPO algorithm that solves a variant of the problem \eqref{eq:eq_14}. The PPO algorithm replaces the trust-region constraint in \eqref{eq:eq_14} by bounding the probability ratio $\frac{\pi_{\theta}(a|s)}{\pi_{\theta_{\text{old}}}(a|s)}$ in an interval. Formally, when updating parameter $\theta_{\text{old}}$, the PPO algorithm solves the problem 
\begin{align}\label{eq:ppo_obj}
	\mathop{\max}_{\theta}\	&\mathcal{L}^{\text{CLIP}}(\theta) := \notag \\
	&E_{s \sim \rho_{\theta_{\text{old}}}, a \sim \pi_{\theta_{\text{old}}}} \left[\mathop{\min} \left \{ \frac{\pi_{\theta}(a|s)}{\pi_{\theta_{\text{old}}}(a|s)} A_{\theta_{\text{old}}}(s, a), \ \operatorname{clip} \left( \frac{\pi_{\theta}(a|s)}{\pi_{\theta_{\text{old}}}(a|s)}, 1-\epsilon, 1+\epsilon \right ) A_{\theta_{\text{old}}}(s, a) \right \} \right],
\end{align}
where the probability ratio is bounded in an interval $\left[ 1 - \epsilon, 1 + \epsilon \right]$ by a clip function with $\epsilon > 0$, and the clip function $\operatorname{clip}(y, b_1, b_2)$ with $b_1 \le b_2$ is defined as $\operatorname{clip}(y, b_1, b_2) = \max(\min(y, b_2), b_1)$. 
Through clipping the probability ratio, 
PPO can take a larger 
improvement step size when updating a policy, but not too far away to undermine performance. 

The value function can be parameterized by another parameter $\phi$ and denoted as $V_{\phi}$. The parameter $\phi$ is updated through the following problem
\begin{align}\label{eq:value_optimize}
	\mathop{\min}_{\phi}\ &\mathcal{V}(\phi) := E_{s \sim \rho_{\theta_{\text{old}}}} \left[\left(V_{\phi}(s) - V^{\text{target}}_{\phi_{\text{old}}}(s)\right)^2\right],
\end{align}
where $V^{\text{target}}_{\phi_{\text{old}}}(s)$ is the target value estimated based on observed sample path of rewards and parameter $\phi_{\text{old}}$.
In addition, for sufficient exploration, it is desirable to minimize the entropy loss $\mathcal{U}(\theta)=-E_{s \sim \rho_{\theta_{\text{old}}}}[\text{Entropy}(\pi_{\theta}(\cdot|s))]$. 

Finally, the PPO algorithm updates the parameters from $(\theta_{\text{old}}, \phi_{\text{old}})$ to new parameters $(\theta, \phi)$  by applying multiple steps of stochastic gradient descent for solving the following problem
\begin{align}\label{eq:total_loss}
	\mathop{\min}_{\theta, \phi}\ &\mathcal{L}(\theta, \phi) := -\mathcal{L}^{\text{CLIP}}(\theta) + e_1\mathcal{V}(\phi) + e_2 \mathcal{U}(\theta),
\end{align}
for some coefficients $e_1 > 0$ and $e_2 \ge 0$.

The PPO algorithm is implemented as an actor-critic on-policy algorithm in which the policy $\pi_{\theta}$ and the value function $V_{\phi}$ are modeled by neural networks, and the expectation in \eqref{eq:total_loss} is replaced by the average of samples obtained by running the policy $\pi_{\theta_{\text{old}}}$ and using $V_{\phi_{\text{old}}}$.


\subsubsection{The States and Networks in the RL Method for the Return Tracking Problems}\label{subsubsec:states}

One of the big advantages of the RL method is that it can effectively utilize market information other than prices or returns, such as market volatility, sentiment, liquidity, etc. Although such information has been shown to provide  effective prediction of future stock and index returns, it cannot be used in existing methods for index tracking, which can only use historical data of prices or returns. In order to learn the time series correlation and cross-sectional comovements of the $N$ stocks and the index through RL, we specify the state $s_t$ at time $t$ to include the market information during $(t-n_s, t]$, where $n_s$ is the number of look-back time periods. Since each time period has $M$ trading days, $(t-n_s, t]$ include $n_sM$ trading days. 


More precisely, we define the state $s_t$ as $s_t = \left(I_{1, t}, I_{2, t}, I_{3, t}, I_{4, t}, I_{5, t} \right) \in \mathbb{R}^{n_sM \times (2N+3)}$, where $I_{1, t}\in \mathbb{R}^{n_sM}$ is the column vector of daily index returns from $t-n_s+\frac{1}{M}$ to time $t$;
$I_{2, t}\in \mathbb{R}^{n_sM}$ is the column vector of daily VIX index values from $t-n_s+\frac{1}{M}$ to time $t$; 
$I_{3, t}\in \mathbb{R}^{n_sM}$ is the column vector of daily US one month Treasury bill rates from $t-n_s+\frac{1}{M}$ to time $t$;
$I_{4, t}\in \mathbb{R}^{n_sM\times N}$ is the matrix of daily returns of the $N$ stocks from $t-n_s+\frac{1}{M}$ to time $t$; 
$I_{5, t}\in \mathbb{R}^{n_sM\times N}$ is the matrix of daily trading volumes of the $N$ stocks from $t-n_s+\frac{1}{M}$ to time $t$.

We specify the stochastic policy $\pi_{\theta}(\cdot|s_t)$  to be a truncated diagonal Gaussian distribution with a mean vector 
$\mu_{\theta_1}(s_t)$ and a standard deviation vector $\sigma_{\theta_2}(s_t)$, where 
$\mu_{\theta_1}(\cdot)$ and $\sigma_{\theta_2}(\cdot)$ are respectively represented by two feed-forward neural networks (FNNs) with parameters $\theta_1$ and $\theta_2$. So the two FNNs constitute the policy network and $\theta = (\theta_1, \theta_2)$. Then, given $\mu_{\theta_1}(s_t)$, $ \sigma_{\theta_2}(s_t)$, $a_t$ can be generated by 
\begin{align}\label{eq:sto_action}
	a_t = \operatorname{clip}\left( \mu_{\theta_1}(s_t) + \sigma_{\theta_2}(s_t) \odot z, -b, b\right), \ b > 0,
\end{align}
where $z$ is a vector of independent standard normal random variables, $\odot$ denotes the element-wise product of two vectors, $b>0$ is a specified parameter. We truncate all elements of $a_t$ to be within $[-b, b]$ for the sake of numerical stability. 

For the return tracking problem \eqref{eq:cumulative_discounted_tracking_error}, $\mu_{\theta_1}(s_t)$ and $\sigma_{\theta_2}(s_t)$ are of dimension $N$.
The $n_sM \times (2N+3)$ matrix $s_t$ is first flattened into a one dimensional row vector by sequentially concatenating the rows of $s_t$. Then the one dimensional row vector is inputted into the two FNNs. Both FNNs consist of an input layer with $n_sM(2N+3)$ nodes, eight hidden layers, each of which contains 128 nodes, and an output layer with $N$ nodes. The activation function of the hidden layers is $\tanh$ function. 
The FNN for the mean vector has a linear output layer; the FNN for the standard deviation vector has an output layer with an exponential activation function so that the output values are positive. Besides, for stability and efficiency of training, a batch normalization layer in front of each hidden layer is active during training and frozen thereafter.

For the return tracking problem \eqref{eq:cumulative_discounted_tracking_error}, the action $a_t$ is of dimension $N$. $a_t$ determines $w_t$ through the softmax function $ g: \mathbb{R}^{N} \rightarrow \mathbb{R}^{N} $:
\begin{align}
	& w_{t} = g(a_t) = \left(\frac{e^{a_{1, t}}}{\sum_{j=1}^{N} e^{a_{j,t}}}, \dots, \frac{e^{a_{N, t}}}{\sum_{j=1}^{N} e^{ a_{j,t}}}\right). \notag
\end{align}


%

Similar to the policy network, we specify the value network representing $V_{\phi}(s_t)$ to be a FNN with one input layer consisting of $n_sM(2N+3)$ nodes, six hidden layers with 128 nodes, and one linear output layer with one node. The FNN also has a batch normalization layer in front of each hidden layer. The activation function of the hidden layers is a $\tanh$ function.

\subsubsection{The States and Networks in the RL Method for the Value Tracking Problem}\label{subsubsec:RL_problem_solution_vb}


For the value tracking problem \eqref{eq:cumulative_discounted_tracking_error_vb}, the state $s_t$ is specified the same as that for the return tracking problem. The stochastic policy $\pi_{\theta}(\cdot|s_t)$ and $a_t$ are specified the same as those for the return tracking problem, except that $a_t$,
$\mu_{\theta_1}(s_t)$, and $\sigma_{\theta_2}(s_t)$ are all of dimension $N+1$ instead of $N$. The policy network shares the same structure as that for the return tracking problem, except that the two FNNs both have four hidden layers with 64 nodes and an output layer with $N+1$ nodes. The action 
$a_t = (a_{1,t}, \dots, a_{N,t}, a_{N+1, t})$ determines $w_t$ and $f_t$ through 
\begin{align}
	w_{t} &= g(a_{1,t}, \dots, a_{N,t}) = \left(\frac{e^{a_{1, t}}}{\sum_{j=1}^{N} e^{a_{j,t}}}, \dots, \frac{e^{a_{N, t}}}{\sum_{j=1}^{N} e^{ a_{j,t}}}\right),\notag\\
	f_t & = b_f \frac{\operatorname{sig}\left(a_{N+1, t}\right)}{\operatorname{sig}\left(b\right)},\label{eq:f_a_definition}
\end{align}
where $\operatorname{sig}(\cdot)$ is the sigmoid function $\operatorname{sig}\left(y\right) = 1/(1 + e^{-y}), \ y \in \mathbb{R}$. As $|a_{i,t}|\leq b$ (by Equation \eqref{eq:sto_action}), $f_t$ defined in \eqref{eq:f_a_definition} satisfies the bound constraint in Equation \eqref{eq:constraint_proportion}.

The value network representing $V_{\phi}(s_t)$ is the same as that for the return tracking problem, except that the FNN has two hidden layers with 64 nodes.

\subsection{Data Limitation and A New Training Scheme of the RL Approach}\label{subsec:new_training_scheme}


Data limitation is one of the major difficulties of the proposed RL method for financial index tracking. In the application of RL for computer games and robotics \citep[see, e.g.,][]{mnih2015human, silver2017mastering, levine2016end},
unlimited data can be generated for training by playing games or operating robots as many times as needed, but there is only one sample path of the time series of financial data. The time period of index tracking problems is typically a quarter or half a year, so basically we need a large number of quarterly or semiyearly training trajectories of the interaction between the agent and financial market from the single sample path. However, we only have a sample path of at most a few decades of training data. In addition, 
the financial index tracking problem involves a large number (in the order of a hundred) of assets 
and market information variables that are correlated cross-sectionally and across time, so it is difficult to build a good parametric model for the joint dynamics of the large number of assets and market information variables. Hence, no simulated data is available for training.

We propose a new training scheme for the RL method to address the issue of data limitation. The basic idea is to obtain enough quarterly or semiyearly data from daily data by randomly selecting the starting date of the first time period of each episode during training.

\subsubsection{Training Scheme}\label{subsubsec:training_strategy}

We split a given data set of daily market information into a training set and a testing set. 
The training set includes daily data at time $-n_s, -n_s + \frac{1}{M}, \dots, -n_s + \frac{M-1}{M}, -n_s + 1, -n_s + 1 + \frac{1}{M}, \dots, T_{\text{train}} - \frac{1}{M}, T_{\text{train}}$;
the testing set includes daily data at time  $T_{\text{train}}, T_{\text{train}}+\frac{1}{M}, \dots,  T_{\text{train}} + \frac{M-1}{M}, T_{\text{train}} + 1, T_{\text{train}} + 1 + \frac{1}{M}, \dots, T_{\text{test}} - 1 + \frac{M-1}{M}, T_{\text{test}}$,
where $T_{\text{test}} > T_{\text{train}} > 0$. The training set starts at time  $-n_s$ because state $s_0$ contains market information from time $-n_s$ to time $0$.


During training, the PPO algorithm updates the parameters $\theta$ and $\phi$ in $H$ epochs of iterations. In each epoch, based on the current parameters $\theta_{\text{old}}$ and $\phi_{\text{old}}$, the algorithm runs the policy $\pi_{\theta_{\text{old}}}$ on the training set
to collect 
$K$ episodes of training trajectories 
\begin{equation}\label{equ:tuple_for_train} 
	\left\{\left.\left(s_t, a_t, s_{t+1}, r_{t}, 
	\pi_{\theta_{\text{old}}}(s_t|a_t), \sigma_{\theta_{\text{old}}}(s_{t}), \hat{A}_t, \hat{\eta}_t\right)\right | t=t_0, t_0+1, \ldots, t_0 + \min(n-1, n') \right\},
\end{equation} 
where $t_0$ is the starting time of the episode,
$\hat{A}_t$ is the advantage estimator, $\hat{\eta}_t$ is the cumulative discounted reward,
$n\in \mathbb{N}$ is the pre-specified episode length, and 
\begin{equation}
	n' = \mathop{\max} \left\{ k \in \mathbb{Z}_{+}: t_0 + k < T_{\text{train}}\right\}.
\end{equation} 
If $t_0$ is close to $T_{\text{train}}$, then $n'$ may be smaller than $n - 1$, and then the actual episode length is $n'+1$.  
The $K$ episodes of experiences in \eqref{equ:tuple_for_train} are put into a buffer, and then mini-batches of experiences are randomly selected from the buffer and used to estimate the expectation in the objective function \eqref{eq:total_loss} by sample average, and then the parameters are updated by minimizing the estimated objective function based on stochastic gradient descent.



To obtain as many different episodes as possible from the single sample path of training data, we 
randomly sample the initial time $t_0$ from the following  discrete distribution
\begin{equation}\label{equ:t_0_dist}
	t_0 \overset{d}{\sim} 
	\begin{cases}
		\operatorname{uniform}\left\{0, \frac{1}{M}, \frac{2}{M},\ldots, T_0 - \frac{1}{M} \right\}, & \text{with probability} \ \zeta, \\
		\operatorname{uniform}\left\{T_0, T_0+\frac{1}{M}, T_0+\frac{2}{M},\ldots, T_{\text{train}} - \frac{1}{M} \right\}, & \text{with probability} \ 1 - \zeta, 
	\end{cases}
\end{equation}
where $\zeta \in [0, 1]$ and $T_0 \in \left\{\frac{1}{M}, \frac{2}{M}, \dots, T_{\text{train}}- \frac{1}{M} \right\}$ are two parameters. If $\zeta = \frac{T_0}{T_{\text{train}}}$, then $t_0$ is uniformly distributed on the dates of training set; if $\zeta < \frac{T_0}{T_{\text{train}}}$, then $t_0$ is more likely sampled from the dates in the time range $[T_0, T_{\text{train}}-\frac{1}{M}]$, so the market information during time periods near the time range of the testing set is given more weight in training, but the market information further earlier than the testing set is still effectively utilized in training. 
For example, in the empirical studies, we specify $\zeta = \frac{1}{4}$ and $T_0 = \frac{1}{M}\lceil \frac{3}{4} (M T_{\text{train}}-1) \rceil$, where $\lceil y \rceil$ denotes the ceiling function. Then $t_0$ leads to the first state $s_{t_0}$. For consistency of notation, we use $s_{t_0} \sim \rho_0(\cdot)$ to denote the sampling of $s_{t_0}$.



To compute the advantage estimator $\hat{A}_t$ in \eqref{equ:tuple_for_train}, we first compute temporal difference (TD) residuals of the approximate value function $V_{\phi_{\text{old}}}(\cdot)$. 
We define the TD residual at time $t_0$ as 
\begin{equation}\label{equ:delta_t_0}
	\delta_{t_0} = 
	\begin{cases}
		r_{t_0} + \gamma V_{\phi_{\text{old}}}(s_{t_0+1}) - V_{\phi_{\text{old}}}(s_{t_0}), & \operatorname{if} t_0 + 1 \le T_{\text{train}}, \\
		\tilde{r}_{t_0} + \gamma V_{\phi_{\text{old}}}(s_{T_{\text{train}}}) - V_{\phi_{\text{old}}}(s_{t_0}), & \operatorname{otherwise}, 
	\end{cases}
\end{equation}
where $\tilde{r}_{t_0} = - \beta \textbf{R-TE}^{q}_{(t_0, T_{\text{train}}]}$ for the return-based tracking problem and $- \beta \textbf{V-TE}^{q}_{(t_0, T_{\text{train}}]}$ for the value-based tracking problem. 
Note that if $t_0 + 1 > T_{\text{train}}$, namely, if the time interval $[t_0, T_{\text{train}})$ has less than $M$ trading days, we use $\tilde{r}_{t_0}$ as the single period reward and use $s_{T_{\text{train}}}$ as the next (and terminal) state. In this case, $n'=0$, and the actual episode length is 1. 

In general, for $l = 0, \dots, \min(n-1, n')$, 
we define the TD residual at time $t_0+l$ as
\begin{equation}\label{equ:delta_boundary_def}
	\delta_{t_0+l} = 
	\begin{cases}
		r_{t_0+l} + \gamma V_{\phi_{\text{old}}}(s_{t_0+l+1}) - V_{\phi_{\text{old}}}(s_{t_0+l}), & \operatorname{if} t_0 + l + 1 \le T_{\text{train}}, \\
		\tilde{r}_{t_0+l} + \gamma V_{\phi_{\text{old}}}(s_{T_{\text{train}}}) - V_{\phi_{\text{old}}}(s_{t_0+l}), & \operatorname{otherwise},
	\end{cases}
\end{equation}
where in the expression of the second case, $\tilde{r}_{t_0+l} = - \beta \textbf{R-TE}^{q}_{(t_0+l, T_{\text{train}}]}$ for the return-based tracking and $- \beta \textbf{V-TE}^{q}_{(t_0+l, T_{\text{train}}]}$ for the value-based tracking. In the second case, it holds that $t_0+l<T_{\text{train}}<t_0+l+1$, so the time range $[t_0+l, T_{\text{train}})$ contains less than $M$ trading days. 

We then estimate the advantage $\hat{A}_{t_0}$ by the generalized advantage estimator \citep{schulman2015high}
\begin{equation}\label{eq:advantage_estimator}
	\hat{A}_{t_0}  = 
	\begin{cases}
		\delta_{t_0} + (\gamma \lambda) \delta_{t_0+1} + \cdots + (\gamma \lambda)^{n-1} \delta_{t_0+n-1}, & \operatorname{if} t_0 + n \le T_{\text{train}}, \\ 
		\delta_{t_0} + (\gamma \lambda) \delta_{t_0+1} + \cdots + (\gamma \lambda)^{n'} \delta_{t_0+n'}, & \operatorname{otherwise},
	\end{cases}
\end{equation}
where 
$\lambda$ is the parameter that makes a compromise between bias and variance. Figure \ref{fig:advantages_1} and \ref{fig:advantages_2} demonstrate the computation of the advantage estimator $\hat{A}_{t_0}$ by using TD residuals for the two cases in Equation \eqref{eq:advantage_estimator}, respectively.

\begin{figure}[htbp]
	{\includegraphics[width=1.0\textwidth]{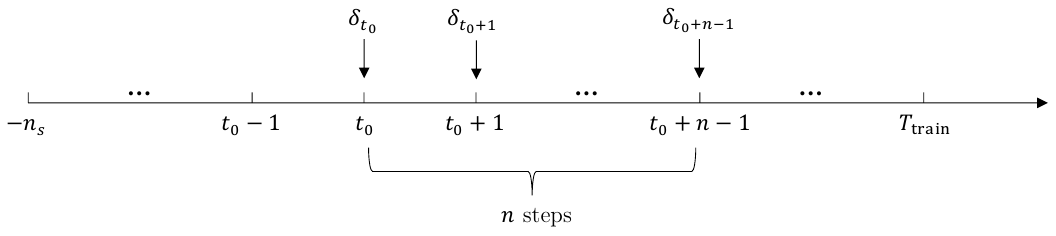}}
	\caption{\textbf{Estimation of the advantage $\hat{A}_{t_0}$ in the case of $t_0 + n \le T_{\text{train}}$}. \label{fig:advantages_1}}
\end{figure}

\begin{figure}[htbp]
	{\includegraphics[width=1.0\textwidth]{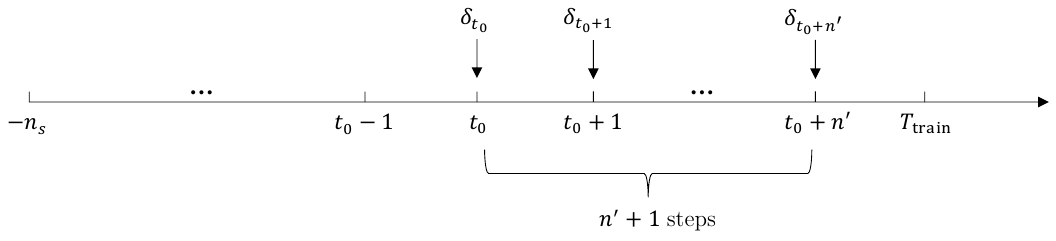}}
	\caption{\textbf{Estimation of the advantage $\hat{A}_{t_0}$ in the case of $t_0 + n > T_{\text{train}}$}. \label{fig:advantages_2}}
\end{figure}

In general, for $l = 0, 1, \dots, \min(n-1, n')$, 
the advantage estimator $\hat{A}_{t_0+l}$ is given by
\begin{equation}
	\hat{A}_{t_0+l} = 
	\begin{cases}
		\delta_{t_0+l} + (\gamma \lambda) \delta_{t_0+l+1} + \cdots + (\gamma \lambda)^{n-l-1} \delta_{t_0+n-1}, & \operatorname{if} t_0+n \le T_{\text{train}}, \notag \\ 
		\delta_{t_0+l} + (\gamma \lambda) \delta_{t_0+l+1} + \cdots + (\gamma \lambda)^{n'-l} \delta_{t_0+n'}, & \operatorname{otherwise}. \notag
	\end{cases}
\end{equation}


Similarly, for $l = 0, 1, \dots,  \min(n-1,n')$, 
the cumulative discounted reward $\hat{\eta}_{t_0+l}$ is estimated by 
\begin{equation}
	\hat{\eta}_{t_0+l} = 
	\begin{cases}
		r_{t_0+l} + \gamma r_{t_0+l+1} + \cdots + \gamma^{n-l-1} r_{t_0+n-1} + \gamma^{n-l} V_{\phi_{\text{old}}}(s_{t_0+n}), & \operatorname{if} t_0 + n \le T_{\text{train}}, \notag \\
		r_{t_0+l} + \gamma r_{t_0+l+1} + \cdots + \gamma^{n'-l} r_{t_0+n'} + \gamma^{n'-l+1} V_{\phi_{\text{old}}}(s_{T_{\text{train}}}), & \operatorname{otherwise}. \notag
	\end{cases}
\end{equation}
$\hat{\eta}_{t_0+l}$ is used as an estimate of the target value function $V^{\text{target}}_{\phi_{\text{old}}}(s_{t_0+l})$ in \eqref{eq:value_optimize}. 


For detailed training procedure of the RL method, we refer the readers to the pseudo code in Appendix \ref{appendix:ppo_training}.

\subsubsection{Value of Terminal State}

The terminal state of the episode in \eqref{equ:tuple_for_train} is $s_{t_0+n}$ if $t_0+n\leq T_{\text{train}}$ and $s_{T_{\text{train}}}$ otherwise. 
An important distinction between our RL method and most existing application of RL in computer games and robotics lies in the value function at the terminal state of each episode. Our RL method solves an infinite-horizon problem; hence, at the terminal state $s_{t_0+n}$ or $s_{T_{\text{train}}}$, the value function is never zero and should be estimated by the current value network $V_{\phi_{\text{old}}}(\cdot)$. In contrast, the training trajectory of computer game problems reaches a terminal state after a finite number of steps, and the value function at the terminal state is simply equal to zero.


\subsubsection{Evaluation of the Trading Strategy in the Testing Set}\label{subsubsec:evaluate_trading_strategy}

After training, we evaluate the performance of the RL method by the testing set. We can use either the stochastic action in Equation \eqref{eq:sto_action} or the deterministic action $a^d_t = \operatorname{clip}(\mu_{\theta_1}(s_t), -b, b)$ to decide the portfolio weights $w_t$ (and $f_t$ for the value-based tracking problem). 
Then, the performance of the RL method is evaluated by 
the return-based tracking error in \eqref{eq:cumulative_discounted_tracking_error} and the value-based tracking error in \eqref{eq:cumulative_discounted_tracking_error_vb} together with the cash flows in \eqref{eq:cash_fn} for the testing set. 

\section{Empirical Results}\label{sec:experiments}
In this section we present empirical results of the proposed RL method and compare with the majorization-minimization method (henceforth, the MM method) proposed in \citet{benidis2017sparse}. We compare the two methods in terms of return-based tracking error, value-based tracking error, cash flows, and transaction costs. According to \citet{benidis2017sparse}, the MM method outperforms some other methods such as the hybrid half-thresholding algorithm \citep{xu2015sparse}, the diversity method \citep{jansen2002optimal}, and the mixed-integer programming method \citep{fuller2008evolution}. Therefore, we choose the MM method as a benchmark to show the effectiveness of the RL method.

We will evaluate the performance of the RL method for tracking the S\&P 500 index, the S\&P 500 EWI, and the DJIA index, which are value-weighted, equal weighted, and price weighted, respectively. 

\subsection{Data and Setting of Out-of-sample Evaluation}\label{subsec:data}
The data set includes daily data of the closing prices and trading volumes of constituent stocks in the indices, the closing prices of the S\&P 500 index (from 01/02/1975 to 12/31/2021), the S\&P 500 EWI (from 02/02/2004 to 12/31/2021), and the DJIA index (from 01/02/1990 to 12/31/2021), and daily closing value of VIX index and T-Bill rates as described in Section \ref{subsubsec:states}. The data of prices, trading volumes, and T-Bill rates are obtained from the Center for Research in Security Prices, LLC (CRSP).\footnote{If the closing price on a day was not available, CRSP records the average of bid and ask prices with a leading dash so as to distinguish it from an actual closing price. We delete the dash and use the bid/ask average price as the missing closing price.} The VIX data are downloaded from the website of Chicago Board Options Exchange (CBOE). 

We use a rolling window approach illustrated in Figure \ref{fig: roll_windows} to evaluate the performance of the RL method. 
Taking the S\&P 500 index as an example, we first 
train the RL model using the data in the training window from 01/02/1990 to 01/04/2010; the tracking portfolio includes all those stocks that were in the S\&P 500 index during the full time period of the training window.  
We then evaluate the trained RL model in the testing window from 01/04/2010 to 01/03/2011. We then roll the training and testing window one year forward, retrain the model in the new training window and evaluate its performance in the new testing window, and so on.


Table \ref{table:data_info} and Table \ref{table:data_info_sp500_rb} show the details of date information and the number of stocks in the tracking portfolio for each training window and testing window for the three indices. The training window length is specified to be 20 years in all cases (Table \ref{table:data_info}) except that 30 years of training data is used for the return-based tracking problem for the S\&P 500 index (Table \ref{table:data_info_sp500_rb});\footnote{When using the 30-year training window, VIX data $I_{2, t}$ are excluded due to data unavailability, so the state $s_t = \left(I_{1,t}, I_{3, t}, I_{4, t}, I_{5, t} \right) \in \mathbb{R}^{n_sM \times (2N+2)}$.} if 20-year training data is not available, the training window contains all available data from the earliest possible starting date. For example, all the first five training windows for the S\&P 500 index start from 01/02/1990, and all the training windows for the S\&P 500 EWI start from 02/02/2004.

For the S\&P 500 index and the DJIA index, we test yearly from 01/03/2005 to 12/31/2021; for the S\&P 500 EWI, we skip the year 2005 and test yearly from 01/03/2006 to 12/31/2021 since the index price data are only available from 02/02/2004. 



\begin{table}[htbp]
	\centering
	\caption
	{\textbf{Time periods and number of stocks in the tracking portfolio for each training and testing windows for the S\&P 500 index, the S\&P 500 EWI, and the DJIA index}. 
		$N$ refers to the number of stocks in the tracking portfolio, which consists of all the stocks that were in the index during the full time period of the training window. }
	\resizebox{0.64\linewidth}{!}{ 
	\begin{tabular}{crrr}
			\toprule
			\multirow{1}{*}{Index} & \multicolumn{1}{c}{Training window} & \multicolumn{1}{c}{Testing window} & \multicolumn{1}{c}{$N$} \\
			\midrule
			\multirow{17}{*}{S\&P 500} & 01/02/1990 - 01/03/2005 & 01/03/2005 - 01/03/2006 & 261 \\
			& 01/02/1990 - 01/03/2006 & 01/03/2006 - 01/03/2007 & 246 \\
			& 01/02/1990 - 01/03/2007 & 01/03/2007 - 01/02/2008 & 228 \\
			& 01/02/1990 - 01/02/2008 & 01/02/2008 - 01/02/2009 & 223 \\
			& 01/02/1990 - 01/02/2009 & 01/02/2009 - 01/04/2010 & 218 \\
			& 01/02/1990 - 01/04/2010 & 01/04/2010 - 01/03/2011 & 213 \\
			& 01/02/1991 - 01/03/2011 & 01/03/2011 - 01/03/2012 & 219 \\
			& 01/02/1992 - 01/03/2012 & 01/03/2012 - 01/02/2013 & 228 \\
			& 01/04/1993 - 01/02/2013 & 01/02/2013 - 01/02/2014 & 239 \\
			& 01/03/1994 - 01/02/2014 & 01/02/2014 - 01/02/2015 & 249 \\
			& 01/03/1995 - 01/02/2015 & 01/02/2015 - 01/04/2016 & 244 \\
			& 01/02/1996 - 01/04/2016 & 01/04/2016 - 01/03/2017 & 257 \\
			& 01/02/1997 - 01/03/2017 & 01/03/2017 - 01/02/2018 & 262 \\
			& 01/02/1998 - 01/02/2018 & 01/02/2018 - 01/02/2019 & 286 \\
			& 01/04/1999 - 01/02/2019 & 01/02/2019 - 01/02/2020 & 295 \\
			& 01/03/2000 - 01/02/2020 & 01/02/2020 - 01/04/2021 & 308 \\
			& 01/03/2000 - 01/04/2021 & 01/04/2021 - 12/31/2021 & 318 \\
			\midrule
			\multirow{16}{*}{S\&P 500 EWI} & 02/02/2004 - 01/03/2006 & 01/03/2006 - 01/03/2007 & 246 \\
			& 02/02/2004 - 01/03/2007 & 01/03/2007 - 01/02/2008 & 228 \\
			& 02/02/2004 - 01/02/2008 & 01/02/2008 - 01/02/2009 & 223 \\
			& 02/02/2004 - 01/02/2009 & 01/02/2009 - 01/04/2010 & 218 \\
			& 02/02/2004 - 01/04/2010 & 01/04/2010 - 01/03/2011 & 213 \\
			& 02/02/2004 - 01/03/2011 & 01/03/2011 - 01/03/2012 & 219 \\
			& 02/02/2004 - 01/03/2012 & 01/03/2012 - 01/02/2013 & 228 \\
			& 02/02/2004 - 01/02/2013 & 01/02/2013 - 01/02/2014 & 239 \\
			& 02/02/2004 - 01/02/2014 & 01/02/2014 - 01/02/2015 & 249 \\
			& 02/02/2004 - 01/02/2015 & 01/02/2015 - 01/04/2016 & 244 \\
			& 02/02/2004 - 01/04/2016 & 01/04/2016 - 01/03/2017 & 257 \\
			& 02/02/2004 - 01/03/2017 & 01/03/2017 - 01/02/2018 & 262 \\
			& 02/02/2004 - 01/02/2018 & 01/02/2018 - 01/02/2019 & 286 \\
			& 02/02/2004 - 01/02/2019 & 01/02/2019 - 01/02/2020 & 295 \\
			& 02/02/2004 - 01/02/2020 & 01/02/2020 - 01/04/2021 & 308 \\
			& 02/02/2004 - 01/04/2021 & 01/04/2021 - 12/31/2021 & 318 \\
			\midrule
			\multirow{17}{*}{DJIA} & 01/02/1990 - 01/03/2005 & 01/03/2005 - 01/03/2006 & 23 \\
			& 01/02/1990 - 01/03/2006 & 01/03/2006 - 01/03/2007 & 23 \\
			& 01/02/1990 - 01/03/2007 & 01/03/2007 - 01/02/2008 & 23 \\
			& 01/02/1990 - 01/02/2008 & 01/02/2008 - 01/02/2009 & 23 \\
			& 01/02/1990 - 01/02/2009 & 01/02/2009 - 01/04/2010 & 21 \\
			& 01/02/1990 - 01/04/2010 & 01/04/2010 - 01/03/2011 & 21 \\
			& 01/02/1991 - 01/03/2011 & 01/03/2011 - 01/03/2012 & 22 \\
			& 01/02/1992 - 01/03/2012 & 01/03/2012 - 01/02/2013 & 22 \\
			& 01/04/1993 - 01/02/2013 & 01/02/2013 - 01/02/2014 & 22 \\
			& 01/03/1994 - 01/02/2014 & 01/02/2014 - 01/02/2015 & 21 \\
			& 01/03/1995 - 01/02/2015 & 01/02/2015 - 01/04/2016 & 21 \\
			& 01/02/1996 - 01/04/2016 & 01/04/2016 - 01/03/2017 & 22 \\
			& 01/02/1997 - 01/03/2017 & 01/03/2017 - 01/02/2018 & 21 \\
			& 01/02/1998 - 01/02/2018 & 01/02/2018 - 01/02/2019 & 22 \\
			& 01/04/1999 - 01/02/2019 & 01/02/2019 - 01/02/2020 & 22 \\
			& 01/03/2000 - 01/02/2020 & 01/02/2020 - 01/04/2021 & 23 \\
			& 01/03/2000 - 01/04/2021 & 01/04/2021 - 12/31/2021 & 24 \\
			\bottomrule
	\end{tabular}}
    \label{table:data_info}
	
\end{table}

\begin{table}[htbp]
	\centering 
	\caption
	{\textbf{Time periods and number of stocks in the tracking portfolio for each training and testing windows for the return-based tracking of S\&P 500 index}. $N$ refers to the number of stocks in the tracking portfolio, which consists of all the stocks that were in the index during the full time period of the training window.}
	\begin{tabular}{crrr}
			\toprule
			\multirow{1}{*}{Index} & \multicolumn{1}{c}{Training window} & \multicolumn{1}{c}{Testing window} & \multicolumn{1}{c}{$N$} \\
			\midrule
			\multirow{17}{*}{S\&P 500} & 01/02/1975 - 01/03/2005 & 01/03/2005 - 01/03/2006 & 90 \\
			& 01/02/1976 - 01/03/2006 & 01/03/2006 - 01/03/2007 & 92 \\
			& 01/03/1977 - 01/03/2007 & 01/03/2007 - 01/02/2008 & 85 \\
			& 01/03/1978 - 01/02/2008 & 01/02/2008 - 01/02/2009 & 88 \\
			& 01/02/1979 - 01/02/2009 & 01/02/2009 - 01/04/2010 & 100 \\
			& 01/02/1980 - 01/04/2010 & 01/04/2010 - 01/03/2011 & 105 \\
			& 01/02/1981 - 01/03/2011 & 01/03/2011 - 01/03/2012 & 110 \\
			& 01/04/1982 - 01/03/2012 & 01/03/2012 - 01/02/2013 & 114 \\
			& 01/03/1983 - 01/02/2013 & 01/02/2013 - 01/02/2014 & 133 \\
			& 01/03/1984 - 01/02/2014 & 01/02/2014 - 01/02/2015 & 136 \\
			& 01/02/1985 - 01/02/2015 & 01/02/2015 - 01/04/2016 & 140 \\
			& 01/02/1986 - 01/04/2016 & 01/04/2016 - 01/03/2017 & 142 \\
			& 01/02/1987 - 01/03/2017 & 01/03/2017 - 01/02/2018 & 152 \\
			& 01/04/1988 - 01/02/2018 & 01/02/2018 - 01/02/2019 & 157 \\
			& 01/03/1989 - 01/02/2019 & 01/02/2019 - 01/02/2020 & 159 \\
			& 01/02/1990 - 01/02/2020 & 01/02/2020 - 01/04/2021 & 155 \\
			& 01/02/1991 - 01/04/2021 & 01/04/2021 - 12/31/2021 & 161 \\
			\bottomrule
	\end{tabular}
	\label{table:data_info_sp500_rb}
\end{table}

\begin{figure}[htbp]
	{\includegraphics[width=1.0\textwidth]{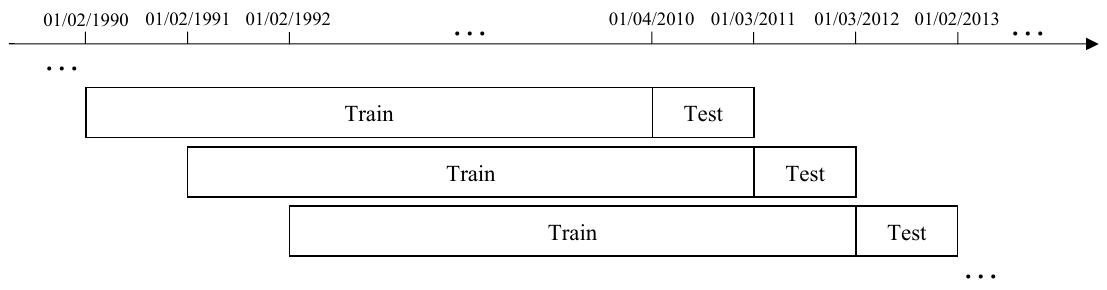}}
	\caption{\textbf{Rolling windows for training and testing}. \label{fig: roll_windows}}
\end{figure}

In all cases, the MM model is trained on a one-year training window right before the testing window. Empirically, the longer the training window is, the worse performance the MM model achieves. Hence, we train the MM model on the training window with a length equal to that of the testing window. For a fair comparison with the RL method, the tracking portfolio under the MM method is also rebalanced daily in order to keep the portfolio weights be identical  to the weights determined at the beginning of the testing window, and the rebalancing is carried out in the same way as that in the RL method. We use the original code provided by \citet{benidis2017sparse}\footnote{The code can be downloaded at \url{https://github.com/dppalomar/sparseIndexTracking}.} to obtain portfolio weights under the MM method.

On each testing window, we use the deterministic actions predicted by the RL method as defined in Section \ref{subsubsec:evaluate_trading_strategy} to obtain the out-of-sample tracking error. All the tests were performed on a high performance computing clusters.

\subsection{Parameters}

We specify the time unit for the  experiments as follows.  
$(i)$ The time unit for the return-based tracking of the S\&P 500 index is half a year, so $M = 126$. $(ii)$ The time unit for the value-based tracking of the DJIA index is one month, so $M = 21$. 
$(iii)$ The time unit for all other cases is a quarter, so $M=63$. Then, we specify parameter $n_s$, the number of time periods included in each state variable, and the episode length $n$ to be $n_s = n = \frac{252}{M}$. Then, the state $s_t$ contains daily information of the past $n_s\times M=252$ trading days.


The MM method uses a parameter $\lambda_1$ in the objective function to control the sparsity of the portfolio, i.e., to limit the number of stocks in the portfolio. Intuitively, with higher sparsity one could achieve lower transaction costs but reduce tracking accuracy. 
To ensure a fair comparison between the two methods, we use an extremely small value 1.000E-20 of the $\lambda_1$ in the MM method so that it achieves the smallest tracking error. Besides, for practical purpose, we adopt a threshold parameter with value 1.000E-09 (the default value of the MM method); stock weights below or equal to the threshold are set to be zero.

Table \ref{table:parameters} lists the main parameters used in the empirical evaluation. 

\begin{table}[htbp]
	\centering 
	\caption
	{\textbf{Parameters for the proposed RL method and the MM method}. \label{table:parameters}}
	\resizebox{\linewidth}{!}{ 	
	\begin{tabular}{clll}
			\toprule
			\multicolumn{1}{c}{Method} & \multicolumn{1}{c}{Parameter} & \multicolumn{1}{c}{Value} & \multicolumn{1}{c}{Usage} \\
			\midrule
			\multirow{15}{*}{RL} & $V_{0-}$  & 2.000E+10 & Initial portfolio wealth \\
			& $x_{0-}$  & $\mathbf{0}$ & Initial shares of stocks \\			
			& $\epsilon$ & 0.2 & Clipping parameter in Equation \eqref{eq:ppo_obj} \\
			& $q$ & 2 & Power in Equation \eqref{eq:te} and \eqref{eq:te_value_based} \\
			& $e_1$ & 0.5 & Weight of value loss in Equation \eqref{eq:total_loss}\\
			& $e_2$ & 0 & Weight of entropy loss in Equation \eqref{eq:total_loss}\\			
			& $b_f$ & 0.5 & Constant in Equation \eqref{eq:cash_fn} \\
			& $b$ & 1.0 & Upper bound of actions in Equation \eqref{eq:sto_action} \\
			& $\zeta$ & 0.25 & Probability parameter in Equation \eqref{equ:t_0_dist}\\
			& $T_0$ & $\frac{1}{M}\lceil \frac{3}{4} (M\times T_{\text{train}}-1) \rceil$ & Parameter in Equation \eqref{equ:t_0_dist}\\            
			& $\tilde M$ & 8 & Number of paralell agents in RL training Algorithm \ref{code:ppo_ts}\\
			& $H$ & 200,000 & Number of epochs in RL training Algorithm \ref{code:ppo_ts}\\
			& $K$ & 50 & Number of episodes per epoch in RL training Algorithm \ref{code:ppo_ts} \\
			& $\tilde n$ & 64 & Minibatch size in RL training Algorithm \ref{code:ppo_ts}\\
			& \multirow{2}{*}{$\beta$} & 1000 & Parameter to rescale rewards in Equation \eqref{eq:reward_rb} \\ & & 0.001 & Parameter to rescale rewards in Equation \eqref{eq:reward_vb} \\
			& learning rate & 1.000E-05 & Learning rate for training \\
			& $\gamma$ & 0.99 & Parameter in Equation \eqref{eq:advantage_estimator} \\
			& $\lambda$ & 0.95 & Parameter in Equation \eqref{eq:advantage_estimator} \\
			& threshold & 1.000E-09 & Stock weights below or equal to the threshold are set to 0 \\
			\midrule
			\multirow{4}{*}{MM} & $\lambda_{1}$ & 1.000E-20 & Sparsity control parameter \\
			& $u$ & 0.5 & Upper bound of the weights \\
			& measure & mean square error & Tracking error measure used in training\\
			& threshold & 1.000E-09 & Stock weights below or equal to the threshold are set to 0 \\
			\bottomrule
	\end{tabular}}
\end{table}

\subsection{Out-of-Sample Performance Metrics}
For each testing window $(t_1, t_2]$ in Table \ref{table:data_info}, 
the out-of-sample return-based tracking error and value-based tracking error are respectively defined in \eqref{eq:te} and \eqref{eq:te_value_based}.

%

For the value-based tracking problem, we are also interested in the total cash injection over the testing window $(t_1, t_2]$:
\begin{equation}
	\text{CF} = \sum_{k=1}^{M (t_2 - t_1)-1} h_{t_1+\frac{k}{M}}, \notag
\end{equation}
where $h_{t_1+\frac{k}{M}}$ is the cash inflow defined in Equation \eqref{eq:cash_fn}. The cash injection with interest rate adjustment is $\text{CF-adj}=\sum_{k=1}^{M (t_2 - t_1)-1} \left[\prod_{j=k}^{M (t_2 - t_1)-1}(1+r_{t_1+\frac{j}{M}})\right] h_{t_1+\frac{k}{M}}$, where $r_{t_1+\frac{j}{M}}$ is the simple interest rate from $t_1+\frac{j}{M}$ to $t_1+\frac{j+1}{M}$. Negative value of $\text{CF}$ or $\text{CF-adj}$ means the fund manager can make extra profit by withdrawing cash from the tracking portfolio over the testing window. 

\subsection{Empirical Performance of the RL Approach for the Return-based Problem}

First, the training process of the RL method for the return-based problem is very stable. For example, 
Figure \ref{fig:rb_sp500_loss_reward} shows the learning curves with respect to training losses and cumulative rewards on  the training window 01/02/1998-01/02/2018 for the return-based tracking of S\&P 500 index. It is evident that training losses and cumulative rewards converge stably. Similar results for the S\&P 500 EWI and the DJIA index can be found in Appendix \ref{appendix:learning_curves}.

\begin{figure}[htbp]
	\centering 
	{\includegraphics[width=0.5\linewidth]{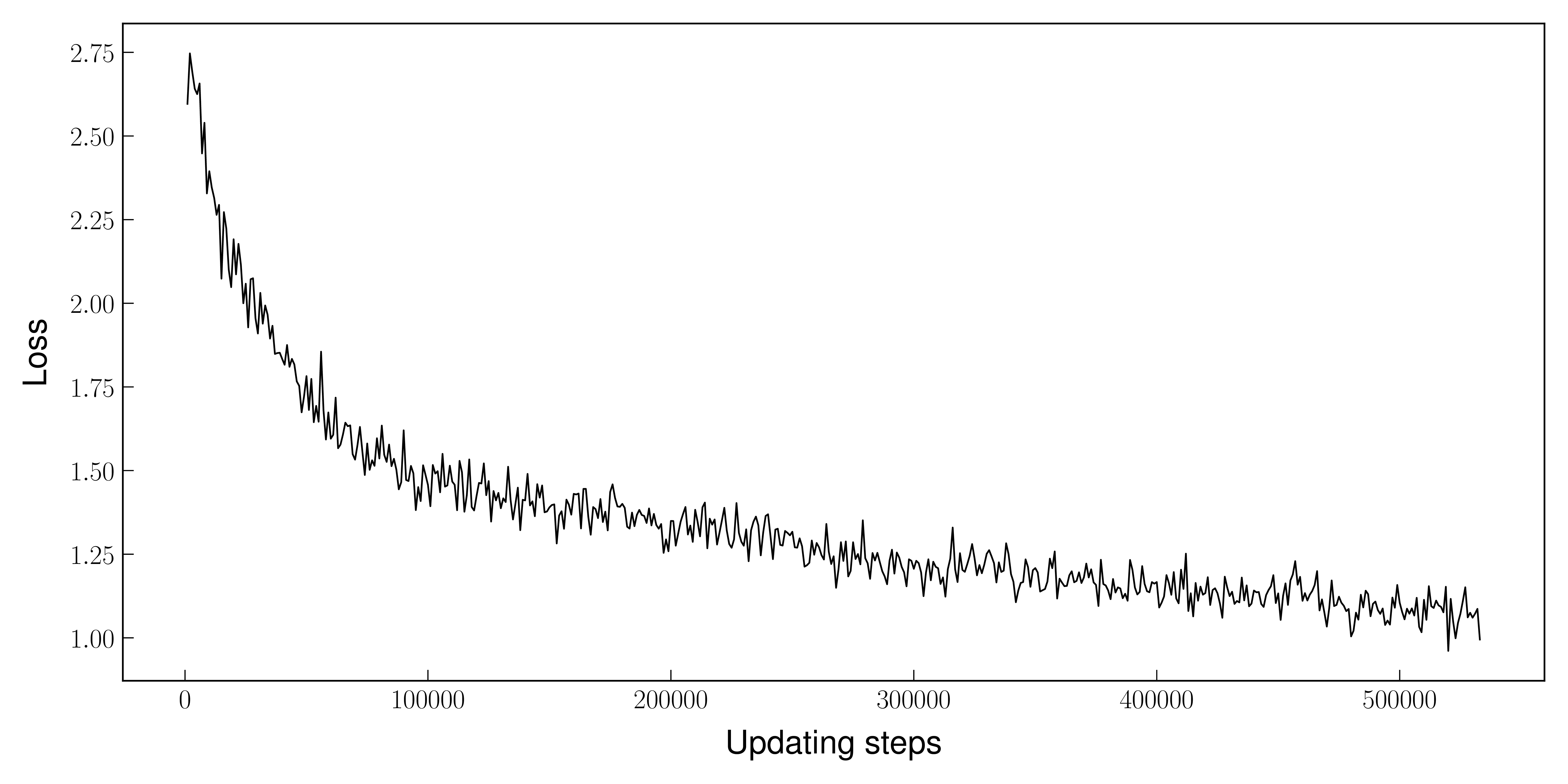}\includegraphics[width=0.5\linewidth]{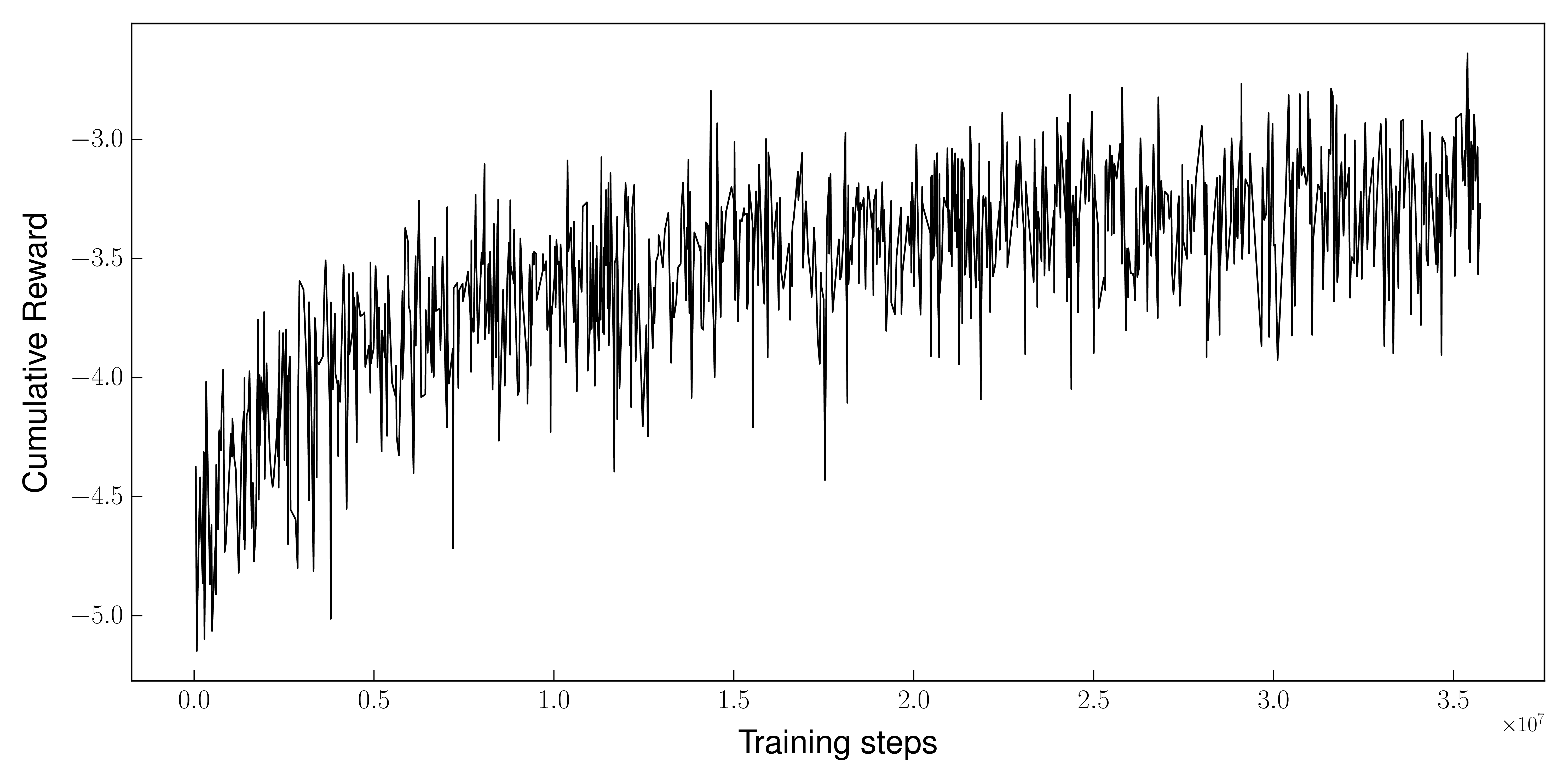}}
	\caption{\textbf{Learning curves for the return-based tracking of S\&P 500 index on the training window 01/02/1998-01/02/2018}. The $x$-axis and the $y$-axis respectively represent the updating step and the loss defined in Equation \eqref{eq:total_loss} in the left subfigure, and the training step and the cumulative reward in the right subfigure.}
	\label{fig:rb_sp500_loss_reward}
\end{figure}

Table \ref{tab:rb_te} shows the out-of-sample return-based tracking error (R-TE) of the proposed RL method and the MM method for each testing window and 
for each of the S\&P 500 index, S\&P 500 EWI, and DJIA index. For all indices, the RL method achieves a lower mean R-TE across the testing years than the MM method. The RL method also achieves a lower standard error of R-TE across the testing years than the MM method for the S\&P 500 EWI and the DJIA index. 

Figure \ref{fig:rb_sp500}, \ref{fig:rb_ewi}, and \ref{fig:rb_djia} show the out-of-sample return-based tracking errors in each testing year of the proposed RL method and the MM method respectively for the three indices. 
The three figures show that the proposed RL method achieves lower out-of-sample return-based tracking error than the MM method in 2008 and 2009, the financial crisis time period, for all three indices.
In particular, the market condition in 2008 was dramatically worse than that in 2007, with VIX index reached its peak value 80.86 on 20 November 2008. The reason that the MM method did not perform well in 2008 may be because the portfolio weights were obtained only based on the data in 2007. In contrast, the RL model for the testing year 2008 was trained based on the data during the training window from 1990 to 2007, which include two other recessions in 1990 and 2002. The RL model may better learn the market dynamics from a much longer time period of data so that it can perform better than the MM method for the testing year of 2008. 


\begin{table}[htbp]
	\centering 
	\caption
	{\textbf{Out-of-sample return-based tracking error (R-TE) of the RL method and the MM method for the S\&P 500 index, S\&P 500 EWI, and DJIA index}. In the last two rows, the mean and stderr refer to the mean and standard error of the R-TE across the testing years from 2005 to 2021 in each column, respectively.	
	}
	{\begin{tabular}{crrrrrr}
			\toprule
			\multirow{2}{*}{Testing year} & \multicolumn{2}{c}{S\&P 500} & \multicolumn{2}{c}{S\&P 500 EWI} & \multicolumn{2}{c}{DJIA} \\
			\cmidrule(lr){2-3}\cmidrule(lr){4-5}\cmidrule(lr){6-7}
			& \multicolumn{1}{c}{RL} & \multicolumn{1}{c}{MM} & \multicolumn{1}{c}{RL} & \multicolumn{1}{c}{MM} & \multicolumn{1}{c}{RL} & \multicolumn{1}{c}{MM}\\
			\midrule
			2005  & 1.483E-03 & 1.885E-03 &   N.A.   &   N.A.   & 2.834E-03 & 3.081E-03 \\
			2006  & 1.672E-03 & 2.894E-03 & 1.800E-03 & 3.409E-03 & 1.356E-03 & 1.385E-03 \\
			2007  & 2.060E-03 & 1.891E-03 & 1.363E-03 & 2.262E-03 & 1.020E-03 & 1.038E-03 \\
			2008  & 4.175E-03 & 5.452E-03 & 2.404E-03 & 2.999E-03 & 4.398E-03 & 4.573E-03 \\
			2009  & 3.800E-03 & 4.053E-03 & 5.119E-03 & 5.753E-03 & 3.220E-03 & 3.986E-03 \\
			2010  & 1.607E-03 & 1.865E-03 & 1.153E-03 & 1.434E-03 & 1.041E-03 & 1.113E-03 \\
			2011  & 1.693E-03 & 1.896E-03 & 1.106E-03 & 1.374E-03 & 1.310E-03 & 1.315E-03 \\
			2012  & 1.895E-03 & 2.185E-03 & 1.033E-03 & 1.123E-03 & 2.397E-03 & 2.556E-03 \\
			2013  & 1.240E-03 & 1.479E-03 & 8.700E-04 & 1.249E-03 & 1.139E-03 & 1.144E-03 \\
			2014  & 1.447E-03 & 2.172E-03 & 8.989E-04 & 1.053E-03 & 1.644E-03 & 1.579E-03 \\
			2015  & 1.876E-03 & 1.785E-03 & 8.992E-04 & 1.078E-03 & 1.900E-03 & 2.106E-03 \\
			2016  & 1.659E-03 & 1.483E-03 & 8.584E-04 & 1.127E-03 & 1.432E-03 & 1.319E-03 \\
			2017  & 1.569E-03 & 1.403E-03 & 9.164E-04 & 1.076E-03 & 1.433E-03 & 1.137E-03 \\
			2018  & 2.410E-03 & 1.769E-03 & 8.799E-04 & 9.845E-04 & 1.624E-03 & 1.304E-03 \\
			2019  & 2.597E-03 & 1.502E-03 & 1.344E-03 & 9.148E-04 & 1.248E-03 & 1.148E-03 \\
			2020  & 5.975E-03 & 5.001E-03 & 2.426E-03 & 2.060E-03 & 3.551E-03 & 4.309E-03 \\
			2021  & 3.260E-03 & 2.283E-03 & 1.318E-03 & 1.399E-03 & 1.320E-03 & 1.346E-03 \\
			\midrule
			mean  & 2.378E-03 & 2.412E-03 & 1.524E-03 & 1.831E-03 & 1.933E-03 & 2.026E-03 \\
			stderr & 3.050E-04 & 3.001E-04 & 2.714E-04 & 3.206E-04 & 2.417E-04 & 2.951E-04 \\
			\bottomrule
	\end{tabular}}
	\label{tab:rb_te}
\end{table}

\begin{figure}[htbp]
	\centering 
	{\includegraphics[width=1\linewidth]{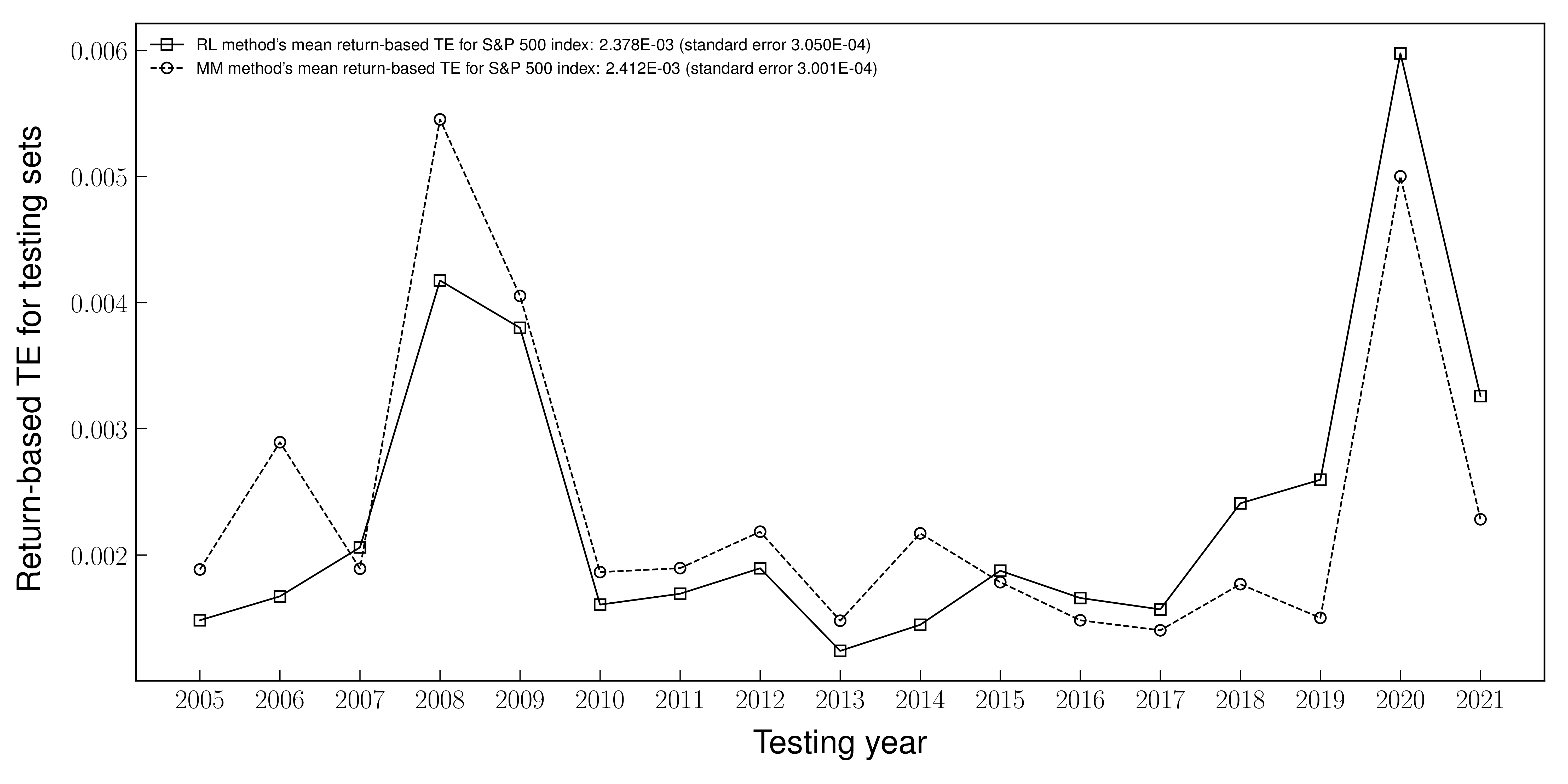}}
	\caption{\textbf{Out-of-sample return-based tracking error (R-TE) of the proposed RL method and the MM method from 2005 to 2021 for the S\&P 500 index}.
	The mean and standard error of the proposed RL method's R-TE across the testing years are 2.378E-03 and 3.050E-04, respectively. The mean and standard error of the MM method's R-TE across the testing years are 2.412E-03 and 3.001E-04, respectively.}
	\label{fig:rb_sp500}
\end{figure}

\begin{figure}[htbp]
	{\includegraphics[width=1\linewidth]{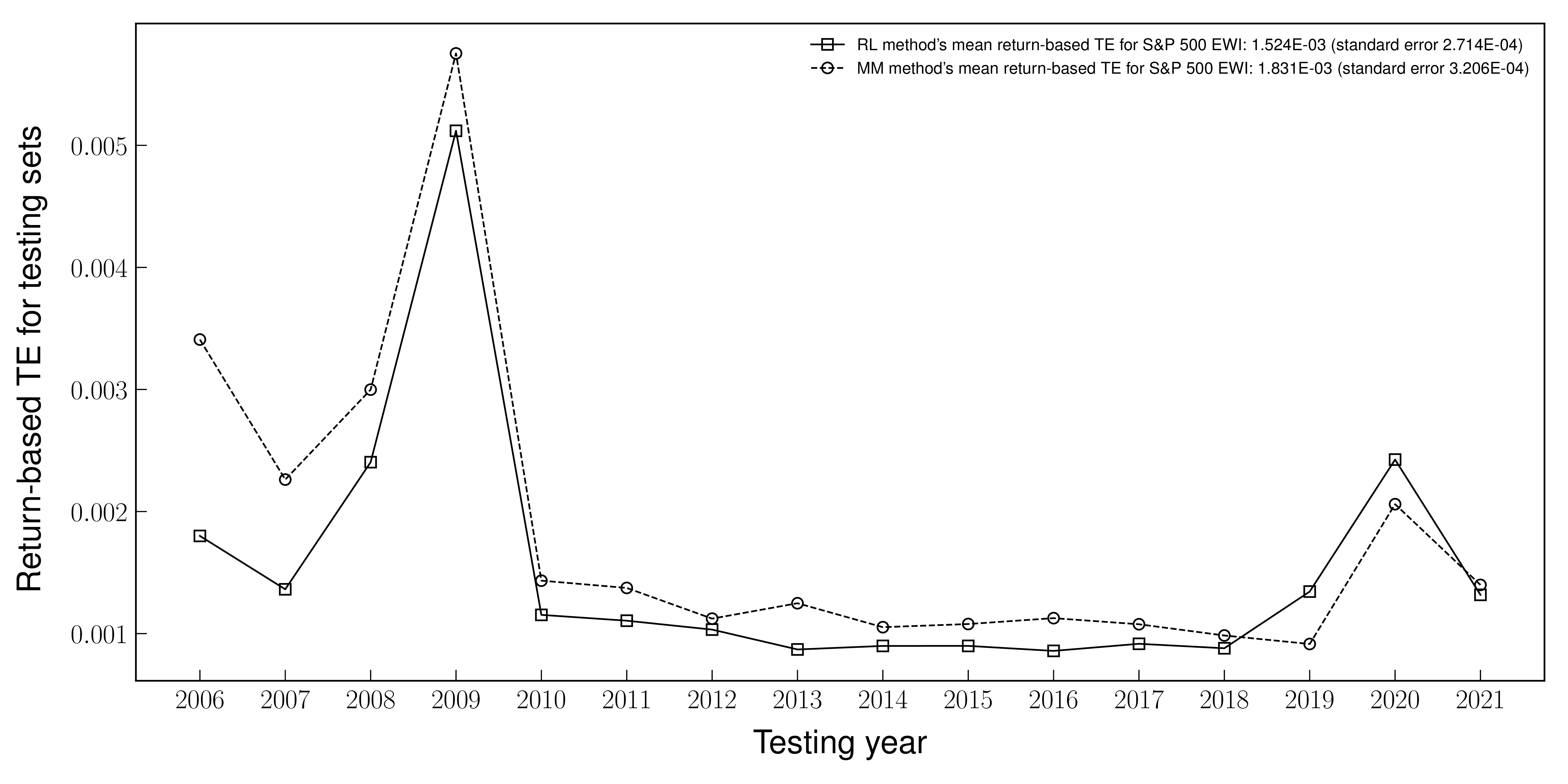}}
	\caption{\textbf{Out-of-sample return-based tracking error (R-TE) of the proposed RL method and the MM method from 2006 to 2021 for the S\&P 500 EWI}. 
	The mean and standard error of the proposed RL method's R-TE across the testing years are 1.524E-03 and 2.714E-04, respectively. The mean and standard error of the MM method's R-TE across the testing years are 1.831E-03 and 3.206E-04, respectively.}
	\label{fig:rb_ewi}
\end{figure}

\begin{figure}[htbp]
	{\includegraphics[width=1\linewidth]{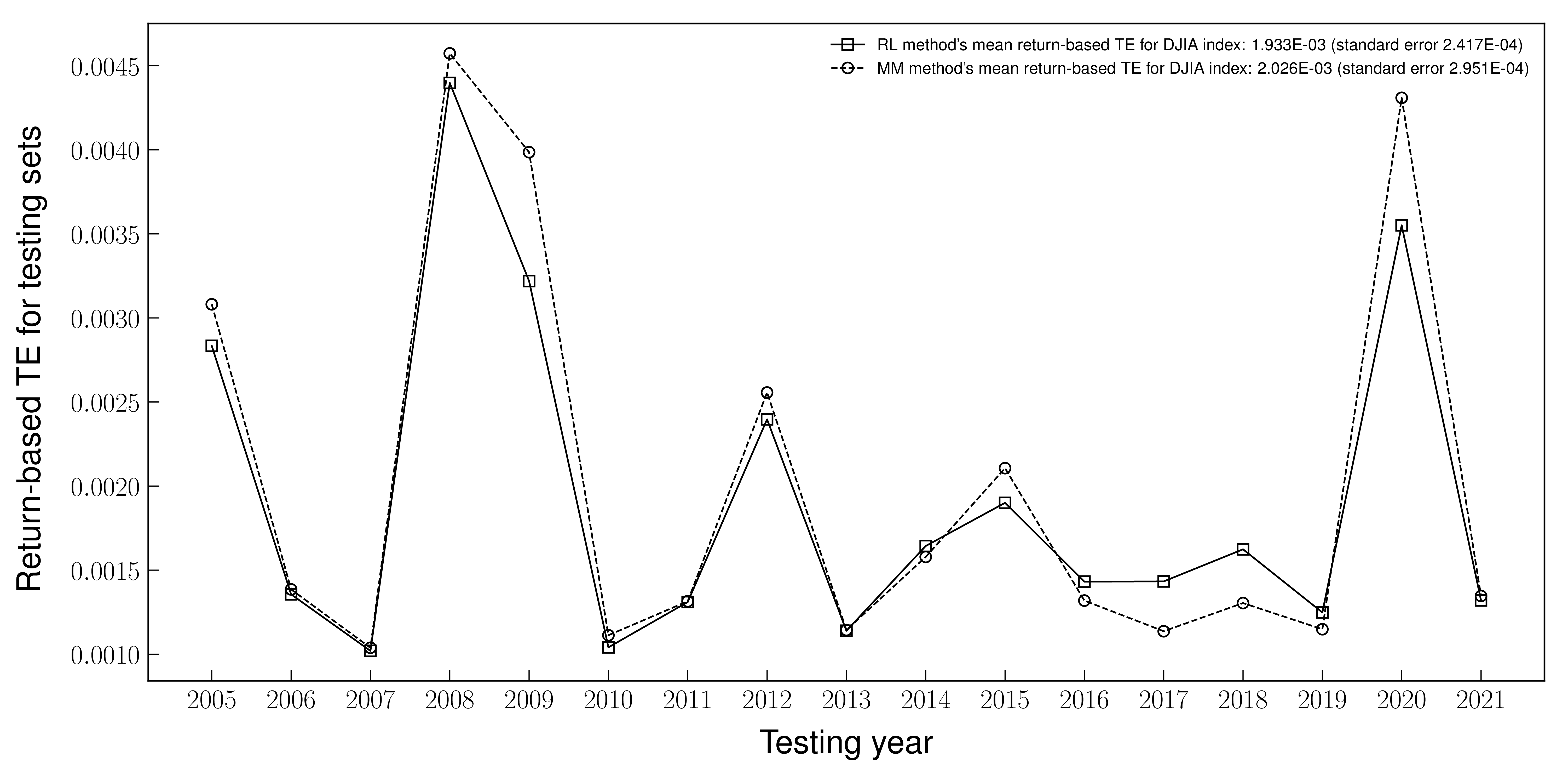}}
	\caption{\textbf{Out-of-sample return-based tracking error (R-TE) of the proposed RL method and the MM method from 2005 to 2021 for the DJIA Index}. 
	The mean and standard error of the proposed RL method's R-TE across the testing years are 1.933E-03 and 2.417E-04, respectively. The mean and standard error of the MM method's R-TE across the testing years are 2.026E-03 and 2.951E-04, respectively.}
	\label{fig:rb_djia}
\end{figure}

\subsection{Empirical Performance of the RL Approach for the Value-based Problem}

First, similar to the return-based problem, the training process of the RL method for the value-based problem is very stable. For example, Figure \ref{fig:vb_sp500_loss_reward} shows the learning curves with respect to training losses and cumulative rewards on the training window 01/02/1998-01/02/2018 for the value-based tracking of S\&P 500 index. It is evident that training losses and cumulative rewards converge stably. Similar results for the S\&P 500 EWI and DJIA index can be found in Appendix \ref{appendix:learning_curves}.

\begin{figure}[htbp]
	{\includegraphics[width=0.5\linewidth]{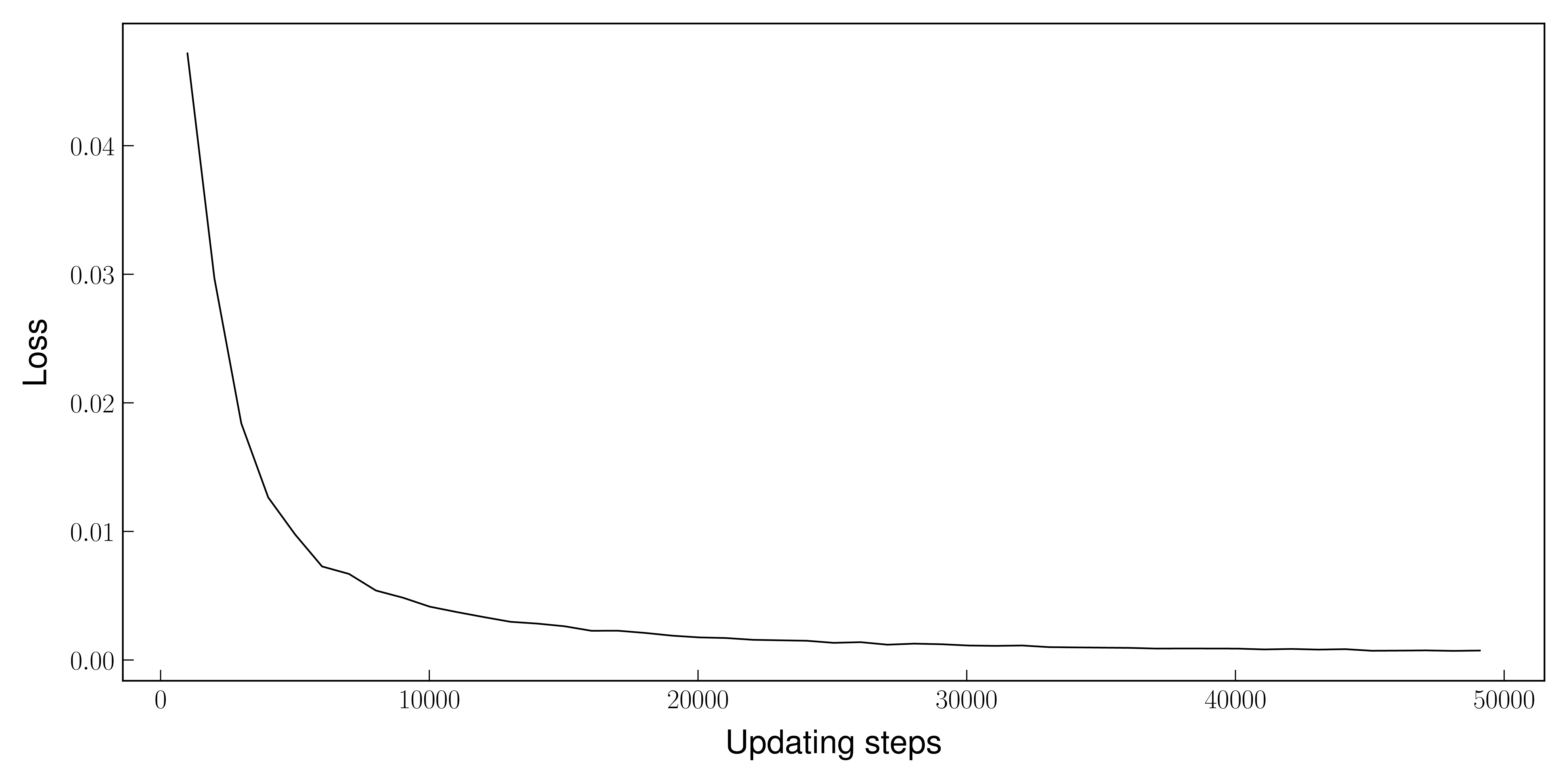}\includegraphics[width=0.5\linewidth]{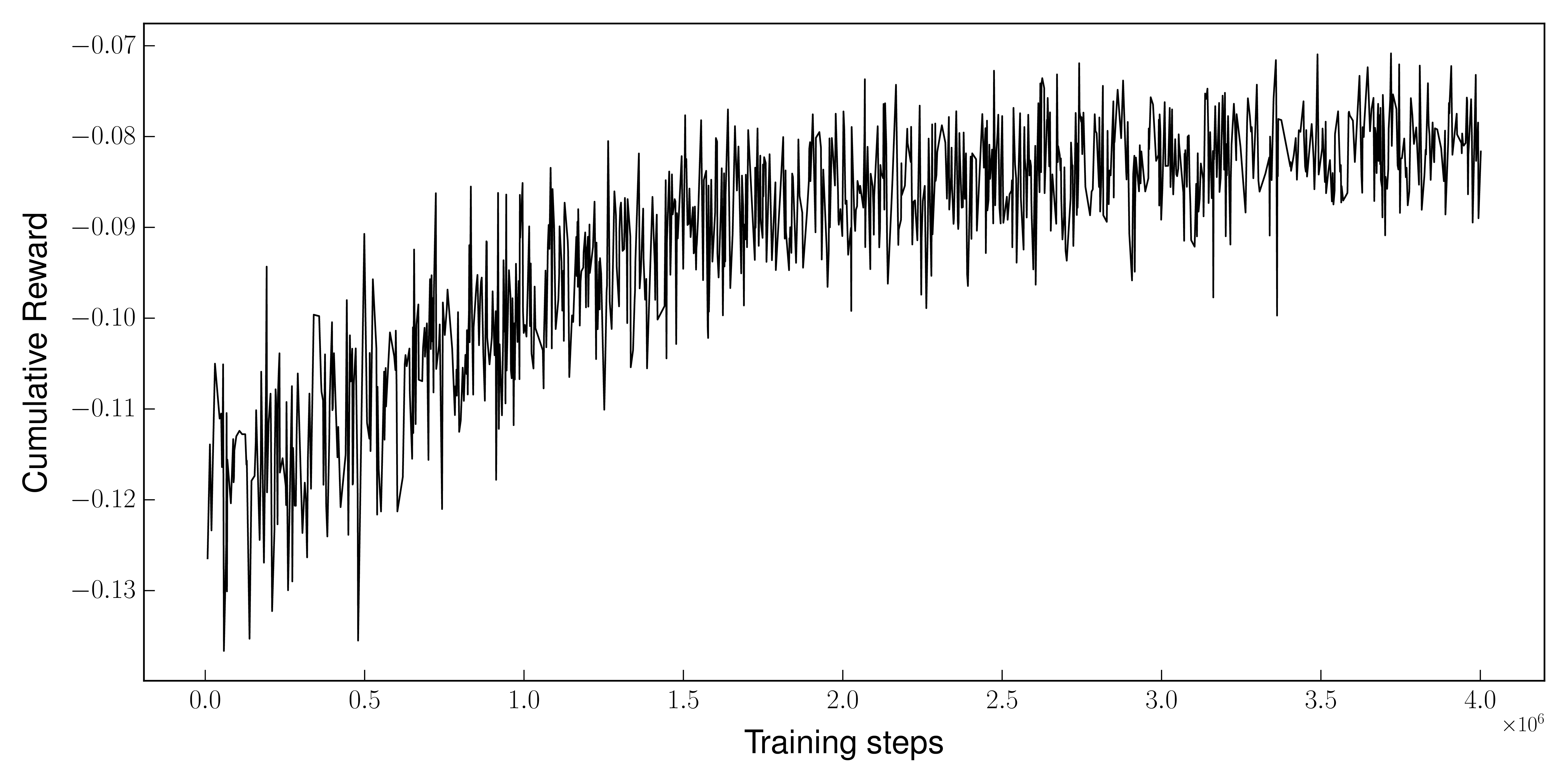}}
	\caption{\textbf{Learning curves of the proposed RL method on the training window 01/02/1998-01/02/2018 for the value-based tracking of S\&P 500 index}.
	The $x$-axis and the $y$-axis respectively represent the updating step and the loss defined in Equation \eqref{eq:total_loss} in the left subfigure, and the training step and the cumulative reward in the right subfigure.}
	\label{fig:vb_sp500_loss_reward}
\end{figure}

Table \ref{tab:vb_sp500_te} and Figure \ref{fig:vb_sp500} 
show the out-of-sample value-based tracking error (V-TE) of the proposed RL method and the MM method for the S\&P 500 index. 
Table \ref{tab:vb_sp500_te} also shows the cash flow (CF), its ratio to the initial wealth (CF/$V_{0-}$), the cash flow with interest rate adjustment (CF-adj), and its ratio to the initial wealth (CF-adj/$V_{0-}$) of the RL method for the S\&P 500 index. 
It is evident that: $(i)$ The RL method achieves about 82\% lower mean value-based tracking error across the testing years than the MM method for the S\&P 500 index. $(ii)$ The RL method achieves lower value-based tracking error than the MM method for each testing year. $(iii)$ The RL method performs stably across the testing years including the 2007 financial crisis. It obtains about 82\% lower standard error of value-based tracking error across the testing years than the MM method. $(iv)$ The mean cash flow without and with interest rate adjustment across the testing years is negative, indicating that the proposed RL method may earn extra profit by withdrawing cash from the tracking portfolio. 


\begin{table}[htbp]
	\caption
	{\textbf{Out-of-sample value-based tracking error (V-TE) and cash flow (CF) of the proposed RL method and the MM method for the S\&P 500 index}. The initial wealth $V_{0-}$ of each testing year is USD 20 billion. 
	CF-adj is the cash flow with interest rate adjustment. The cash flow for the MM method is zero. 		
	In the last two rows, mean and stderr respectively refer to the mean and standard error of the values across the testing years from 2005 to 2021 in each column.	
	 }
	\resizebox{\linewidth}{!}{
	\begin{tabular}{crrrrrr}
			\toprule
			\multirow{2}{*}{Testing year} & \multicolumn{2}{c}{V-TE} & \multicolumn{1}{c}{CF} & \multicolumn{1}{c}{CF-adj} & \multicolumn{1}{c}{CF/$V_{0-}$} & \multicolumn{1}{c}{CF-adj/$V_{0-}$} \\
			\cmidrule(lr){2-3}\cmidrule(lr){4-4}\cmidrule(lr){5-5}\cmidrule(lr){6-6}\cmidrule(lr){7-7}
			& \multicolumn{1}{c}{RL} & \multicolumn{1}{c}{MM} & \multicolumn{1}{c}{RL} & \multicolumn{1}{c}{RL} & \multicolumn{1}{c}{RL} & \multicolumn{1}{c}{RL}\\
			\midrule
			2005  & 2.855E+00 & 2.629E+01 & -4.605E+08 & -4.663E+08 & -2.303E-02 & -2.332E-02 \\
			2006  & 4.647E+00 & 7.175E+01 & 1.517E+08 & 1.612E+08 & 7.587E-03 & 8.061E-03 \\
			2007  & 4.009E+00 & 4.934E+01 & 1.286E+09 & 1.289E+09 & 6.428E-02 & 6.447E-02 \\
			2008  & 5.875E+00 & 1.964E+01 & -1.197E+09 & -1.209E+09 & -5.984E-02 & -6.045E-02 \\
			2009  & 9.996E+00 & 1.806E+02 & -4.273E+09 & -4.276E+09 & -2.136E-01 & -2.138E-01 \\
			2010  & 2.494E+00 & 2.357E+01 & -7.825E+08 & -7.837E+08 & -3.913E-02 & -3.918E-02 \\
			2011  & 2.944E+00 & 1.460E+01 & 3.314E+08 & 3.317E+08 & 1.657E-02 & 1.659E-02 \\
			2012  & 3.611E+00 & 9.055E+00 & -1.242E+07 & -1.211E+07 & -6.209E-04 & -6.055E-04 \\
			2013  & 3.052E+00 & 1.269E+01 & -4.921E+08 & -4.926E+08 & -2.461E-02 & -2.463E-02 \\
			2014  & 3.746E+00 & 3.168E+01 & 4.301E+08 & 4.302E+08 & 2.150E-02 & 2.151E-02 \\
			2015  & 6.225E+00 & 3.524E+01 & 1.374E+09 & 1.374E+09 & 6.868E-02 & 6.871E-02 \\
			2016  & 7.365E+00 & 1.067E+01 & -6.939E+08 & -6.953E+08 & -3.469E-02 & -3.476E-02 \\
			2017  & 6.470E+00 & 7.372E+01 & 7.118E+08 & 7.155E+08 & 3.559E-02 & 3.578E-02 \\
			2018  & 8.676E+00 & 1.816E+01 & 1.267E+09 & 1.277E+09 & 6.337E-02 & 6.385E-02 \\
			2019  & 1.158E+01 & 4.369E+01 & -5.858E+08 & -5.943E+08 & -2.929E-02 & -2.971E-02 \\
			2020  & 3.039E+01 & 5.703E+01 & 4.522E+08 & 4.531E+08 & 2.261E-02 & 2.266E-02 \\
			2021  & 2.589E+01 & 3.130E+01 & -1.356E+09 & -1.357E+09 & -6.781E-02 & -6.785E-02 \\
			\midrule
			mean  & 8.225E+00 & 4.171E+01 & -2.264E+08 & -2.267E+08 & -1.132E-02 & -1.133E-02 \\
			stderr & 1.936E+00 & 9.960E+00 & 3.250E+08 & 3.256E+08 & 1.625E-02 & 1.628E-02 \\
			\bottomrule
	\end{tabular}}
	\label{tab:vb_sp500_te}
\end{table}

\begin{figure}[htbp]
	{\includegraphics[width=1\linewidth]{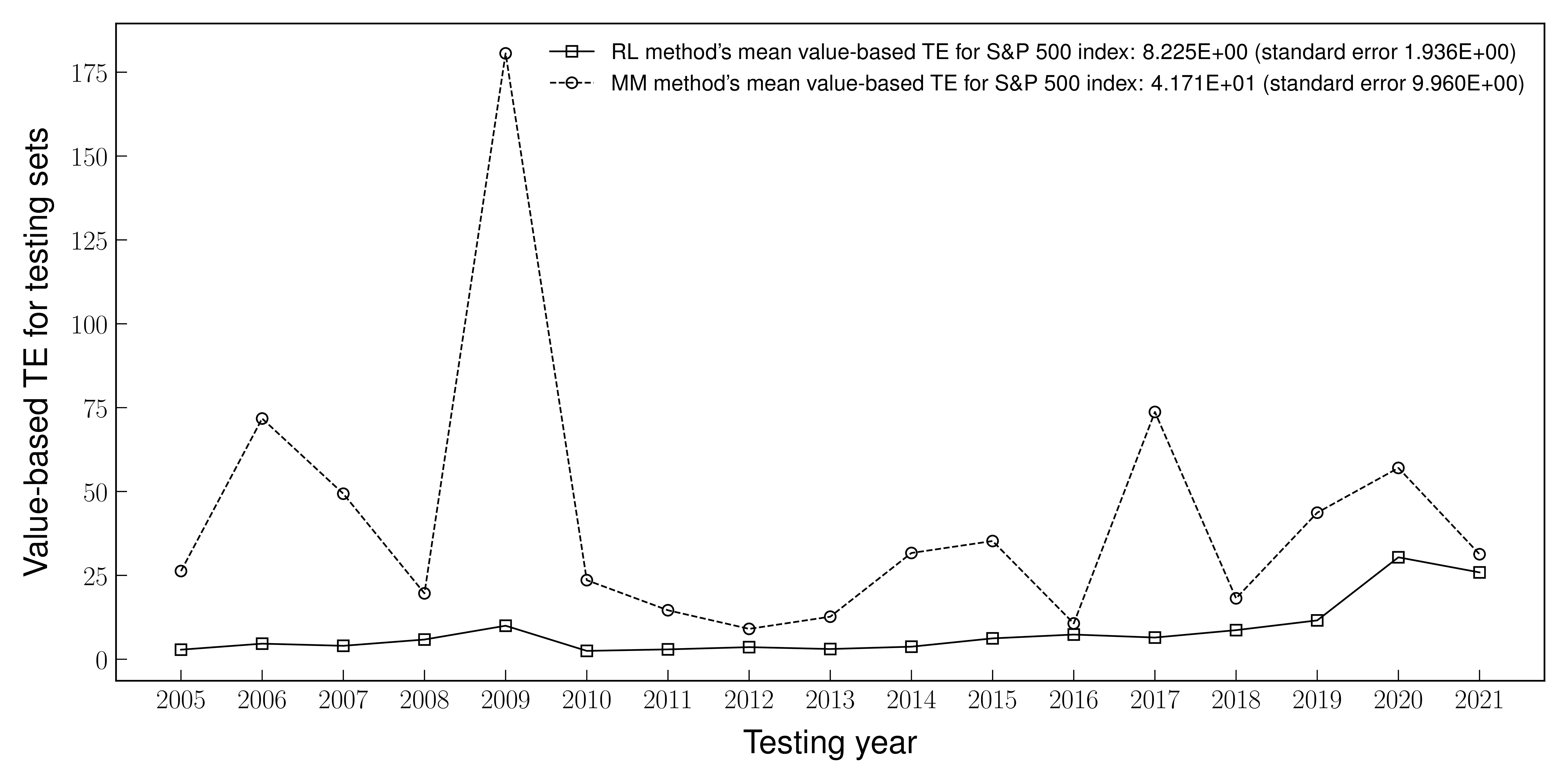}}
	\caption{\textbf{Out-of-sample value-based tracking error (V-TE) of the proposed RL method and the MM method from 2005 to 2021 for the S\&P 500 index}. 
	The mean and standard error of the proposed RL method's R-TE  across the testing years  are 7.311E+00 and 1.756E+00, respectively. The mean and standard error of the MM method's R-TE across the testing years  are 4.171E+01 and 9.960E+00, respectively.}
	\label{fig:vb_sp500}
\end{figure}

Table \ref{tab:vb_ewi_te} and Figure \ref{fig:vb_ewi} show the out-of-sample value-base tracking error (V-TE)
of the proposed RL method and the MM method for the S\&P 500 EWI. Table \ref{tab:vb_ewi_te} also shows the cash flow (CF), its ratio to the initial wealth (CF/$V_{0-}$), the cash flow with interest rate adjustment (CF-adj), and its ratio to the initial wealth (CF-adj/$V_{0-}$). The table and figure show similar results as those for the S\&P 500 index; in particular, the RL method achieves more than 85\% reduction in the mean and standard error of value-based tracking error across the testing years for the S\&P 500 EWI.

\begin{table}[htbp]
	\caption
	{\textbf{Out-of-sample value-based tracking error (V-TE) and cash flow (CF) of the proposed RL method and the MM method for the S\&P 500 EWI}. The initial wealth $V_{0-}$ of each testing year is USD 20 billion. 
		CF-adj is the cash flow with interest rate adjustment. The cash flow for the MM method is zero. 		
		In the last two rows, mean and stderr respectively refer to the mean and standard error of the values across the testing years from 2006 to 2021 in each column.	}
	\resizebox{\linewidth}{!}{
	\begin{tabular}{crrrrrr}
			\toprule
			\multirow{2}{*}{Testing year} & \multicolumn{2}{c}{V-TE} & \multicolumn{1}{c}{CF} & \multicolumn{1}{c}{CF-adj} & \multicolumn{1}{c}{CF/$V_{0-}$} & \multicolumn{1}{c}{CF-adj/$V_{0-}$} \\
			\cmidrule(lr){2-3}\cmidrule(lr){4-4}\cmidrule(lr){5-5}\cmidrule(lr){6-6}\cmidrule(lr){7-7}
			& \multicolumn{1}{c}{RL} & \multicolumn{1}{c}{MM} & \multicolumn{1}{c}{RL} & \multicolumn{1}{c}{RL} & \multicolumn{1}{c}{RL} & \multicolumn{1}{c}{RL}\\
			\midrule
			2006  & 4.976E+00 & 7.175E+01 & 1.107E+07 & 1.955E+07 & 5.534E-04 & 9.776E-04 \\
			2007  & 3.198E+00 & 2.013E+01 & 4.757E+08 & 4.806E+08 & 2.378E-02 & 2.403E-02 \\
			2008  & 4.136E+00 & 4.507E+01 & -1.538E+09 & -1.545E+09 & -7.688E-02 & -7.724E-02 \\
			2009  & 1.575E+01 & 1.195E+02 & -1.655E+09 & -1.656E+09 & -8.275E-02 & -8.280E-02 \\
			2010  & 2.155E+00 & 3.253E+01 & 4.799E+08 & 4.800E+08 & 2.400E-02 & 2.400E-02 \\
			2011  & 3.397E+00 & 3.489E+01 & 2.054E+07 & 2.093E+07 & 1.027E-03 & 1.047E-03 \\
			2012  & 2.893E+00 & 3.365E+01 & 3.789E+08 & 3.791E+08 & 1.894E-02 & 1.896E-02 \\
			2013  & 2.841E+00 & 1.938E+01 & 2.006E+08 & 2.006E+08 & 1.003E-02 & 1.003E-02 \\
			2014  & 3.636E+00 & 5.086E+01 & 5.943E+08 & 5.945E+08 & 2.971E-02 & 2.973E-02 \\
			2015  & 4.297E+00 & 1.500E+01 & 6.347E+08 & 6.351E+08 & 3.174E-02 & 3.176E-02 \\
			2016  & 4.670E+00 & 1.573E+01 & -2.360E+08 & -2.365E+08 & -1.180E-02 & -1.182E-02 \\
			2017  & 4.245E+00 & 9.234E+01 & 2.273E+08 & 2.285E+08 & 1.137E-02 & 1.143E-02 \\
			2018  & 5.044E+00 & 3.838E+01 & 5.611E+08 & 5.658E+08 & 2.805E-02 & 2.829E-02 \\
			2019  & 1.492E+01 & 3.224E+01 & -1.153E+09 & -1.162E+09 & -5.766E-02 & -5.811E-02 \\
			2020  & 8.885E+00 & 4.198E+01 & 2.577E+07 & 2.567E+07 & 1.288E-03 & 1.283E-03 \\
			2021  & 2.085E+01 & 5.205E+01 & -1.466E+09 & -1.467E+09 & -7.331E-02 & -7.333E-02 \\
			\midrule
			mean  & 6.618E+00 & 4.471E+01 & -1.524E+08 & -1.522E+08 & -7.619E-03 & -7.611E-03 \\
			stderr & 1.392E+00 & 7.135E+00 & 2.044E+08 & 2.051E+08 & 1.022E-02 & 1.025E-02 \\
			\bottomrule
	\end{tabular}}
	\label{tab:vb_ewi_te}
\end{table}

\begin{figure}[htbp]
	{\includegraphics[width=1\linewidth]{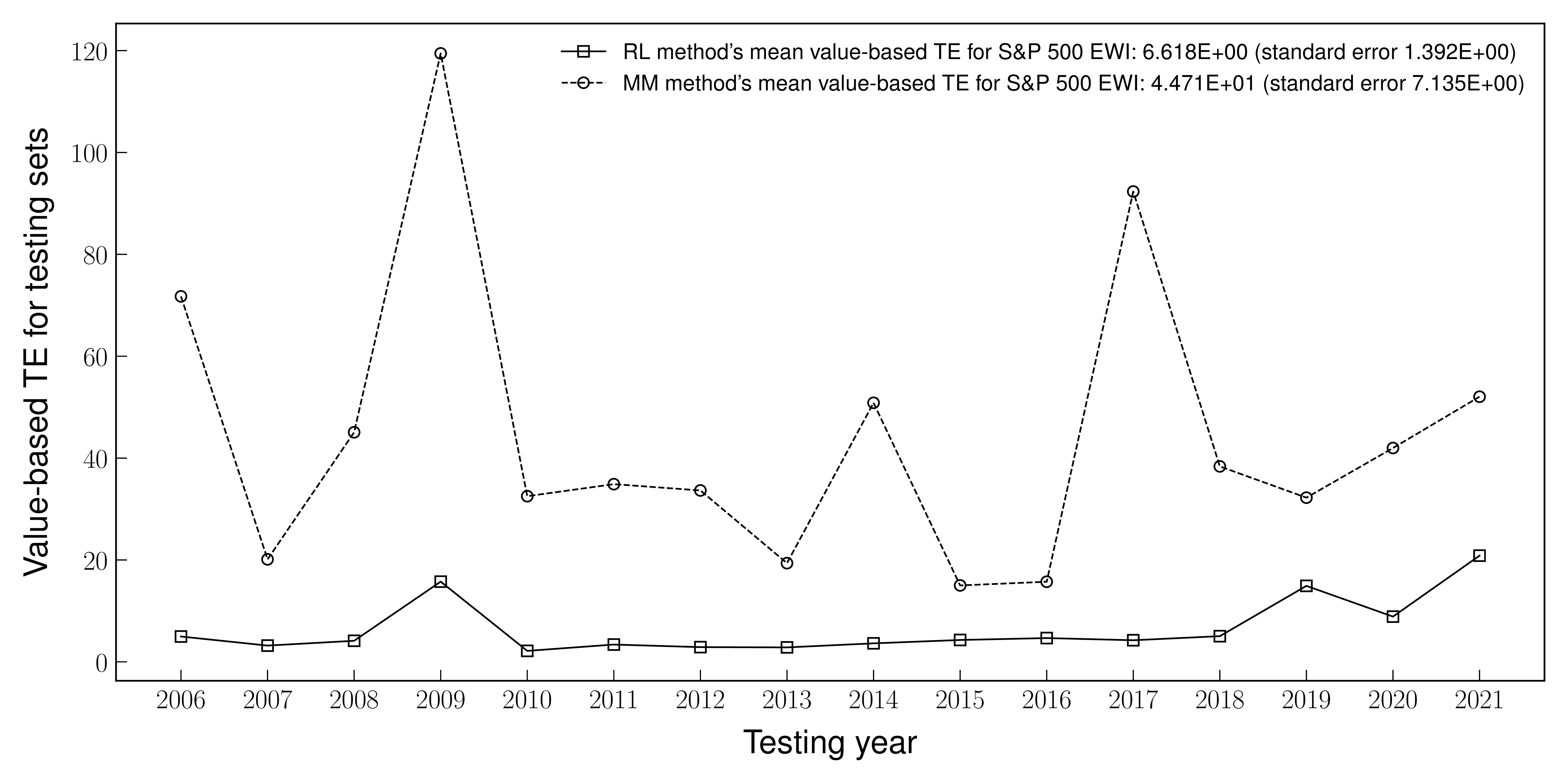}}
	\caption{\textbf{Out-of-sample value-based tracking error (V-TE) of the proposed RL method and the MM method from 2006 to 2021 for the S\&P 500 EWI}. The mean and standard error of the proposed RL method's R-TE across the testing years are 5.524E+00 and 1.051E+00, respectively. The mean and standard error of the MM method's R-TE  across the testing years are 4.471E+01 and 7.135E+00, respectively.}
	\label{fig:vb_ewi}
\end{figure}

Table \ref{tab:vb_djia_te} and Figure \ref{fig:vb_djia} show the out-of-sample value-base tracking error (V-TE) of the proposed RL method and the MM method for the DJIA index. Table \ref{tab:vb_djia_te} also shows the cash flow (CF), its ratio to the initial wealth (CF/$V_{0-}$), the cash flow with interest rate adjustment (CF-adj), and its ratio to the initial wealth (CF-adj/$V_{0-}$) of the RL method for the DJIA index. The table and figure show similar results as those for the S\&P 500 index; in particular, the RL method achieves more than 87\% reduction in the mean and standard error of value-based tracking error across the testing years for the DJIA index. 


\begin{table}[htbp]
	\caption
	{\textbf{Out-of-sample value-based tracking error (V-TE) and cash flow (CF) of the proposed RL method and the MM method for the DJIA Index}. The initial wealth $V_{0-}$ of each testing year is USD 20 billion. 
		CF-adj is the cash flow with interest rate adjustment. The cash flow for the MM method is zero. 		
		In the last two rows, mean and stderr respectively refer to the mean and standard error of the values across testing years from 2005 to 2021 in each column.	}
	\resizebox{\linewidth}{!}{		
	\begin{tabular}{crrrrrr}
			\toprule
			\multirow{2}{*}{Testing year} & \multicolumn{2}{c}{V-TE} & \multicolumn{1}{c}{CF} & \multicolumn{1}{c}{CF-adj} & \multicolumn{1}{c}{CF/$V_{0-}$} & \multicolumn{1}{c}{CF-adj/$V_{0-}$} \\
			\cmidrule(lr){2-3}\cmidrule(lr){4-4}\cmidrule(lr){5-5}\cmidrule(lr){6-6}\cmidrule(lr){7-7}
			& \multicolumn{1}{c}{RL} & \multicolumn{1}{c}{MM} & \multicolumn{1}{c}{RL} & \multicolumn{1}{c}{RL} & \multicolumn{1}{c}{RL} & \multicolumn{1}{c}{RL}\\
			\midrule
			2005  & 3.745E+01 & 6.506E+02 & 1.609E+09 & 1.630E+09 & 8.047E-02 & 8.148E-02 \\
			2006  & 2.104E+01 & 2.084E+02 & -1.280E+08 & -1.321E+08 & -6.399E-03 & -6.604E-03 \\
			2007  & 2.458E+01 & 5.647E+01 & -1.752E+08 & -1.785E+08 & -8.761E-03 & -8.927E-03 \\
			2008  & 8.220E+01 & 3.451E+02 & 1.001E+09 & 1.001E+09 & 5.004E-02 & 5.003E-02 \\
			2009  & 5.290E+01 & 7.099E+02 & -2.286E+09 & -2.287E+09 & -1.143E-01 & -1.144E-01 \\
			2010  & 1.913E+01 & 1.991E+02 & -5.960E+07 & -5.977E+07 & -2.980E-03 & -2.989E-03 \\
			2011  & 3.198E+01 & 1.235E+02 & 8.006E+08 & 8.012E+08 & 4.003E-02 & 4.006E-02 \\
			2012  & 3.685E+01 & 2.631E+02 & -5.326E+08 & -5.332E+08 & -2.663E-02 & -2.666E-02 \\
			2013  & 3.240E+01 & 1.431E+02 & -2.287E+08 & -2.287E+08 & -1.143E-02 & -1.144E-02 \\
			2014  & 3.384E+01 & 2.653E+02 & -3.250E+08 & -3.252E+08 & -1.625E-02 & -1.626E-02 \\
			2015  & 4.186E+01 & 2.397E+02 & -6.255E+07 & -6.242E+07 & -3.127E-03 & -3.121E-03 \\
			2016  & 3.740E+01 & 1.182E+02 & 4.702E+08 & 4.700E+08 & 2.351E-02 & 2.350E-02 \\
			2017  & 3.988E+01 & 3.959E+02 & 6.719E+08 & 6.726E+08 & 3.359E-02 & 3.363E-02 \\
			2018  & 7.724E+01 & 1.269E+02 & -5.753E+08 & -5.773E+08 & -2.876E-02 & -2.887E-02 \\
			2019  & 5.438E+01 & 4.661E+02 & -7.726E+08 & -7.765E+08 & -3.863E-02 & -3.882E-02 \\
			2020  & 1.240E+02 & 7.189E+02 & 4.854E+08 & 4.859E+08 & 2.427E-02 & 2.429E-02 \\
			2021  & 7.041E+01 & 6.728E+02 & 3.180E+08 & 3.181E+08 & 1.590E-02 & 1.591E-02 \\
			\midrule
			mean  & 4.809E+01 & 3.355E+02 & 1.243E+07 & 1.276E+07 & 6.214E-04 & 6.380E-04 \\
			stderr & 6.540E+00 & 5.510E+01 & 2.091E+08 & 2.098E+08 & 1.045E-02 & 1.049E-02 \\
			\bottomrule
	\end{tabular}}
	\label{tab:vb_djia_te}
\end{table}

\begin{figure}[htbp]
	{\includegraphics[width=1\linewidth]{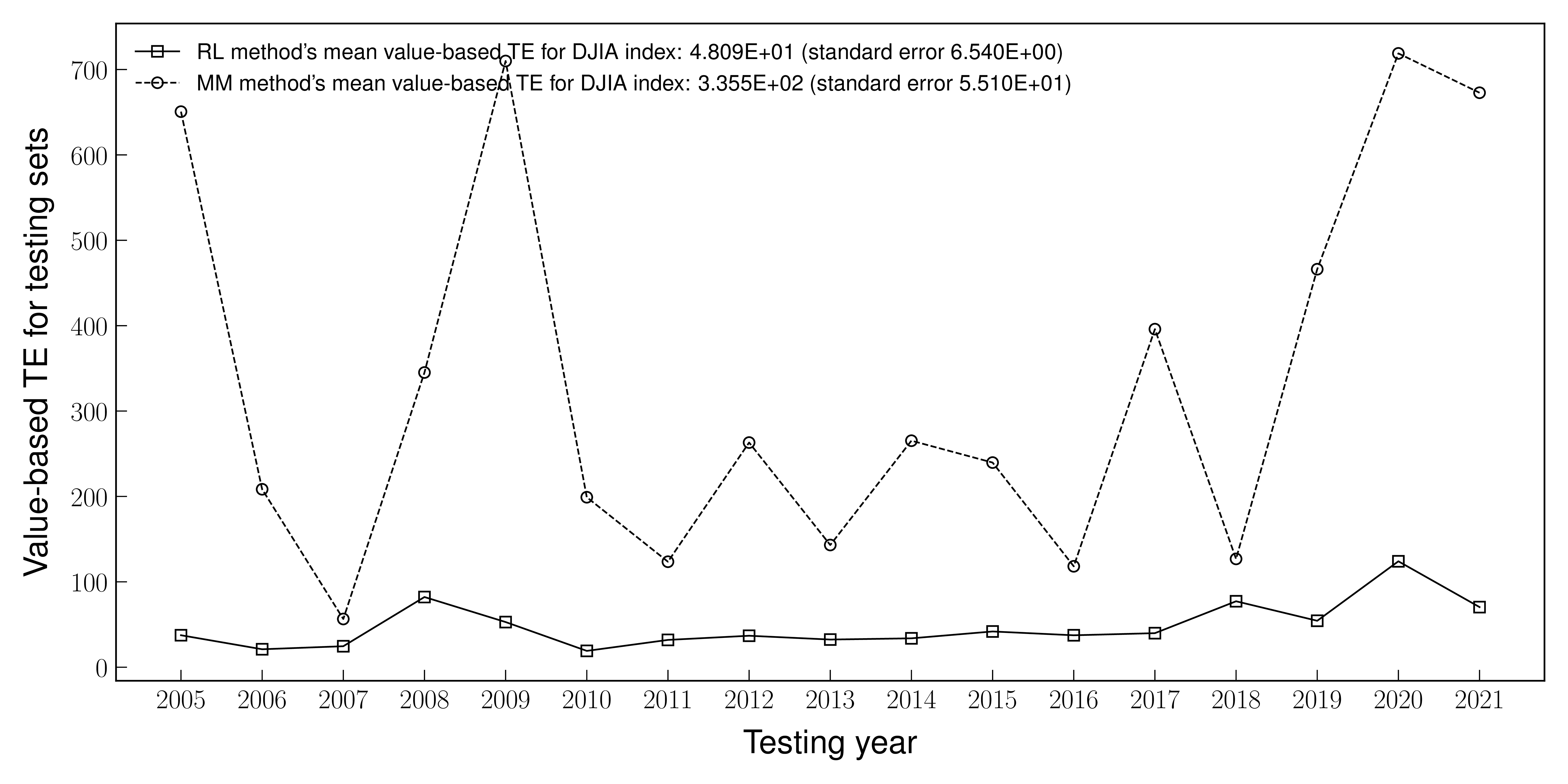}}
	\caption{\textbf{Out-of-sample value-based tracking error (V-TE) of the proposed RL method and the MM method from 2005 to 2021  for the DJIA index}. The mean and standard error of the proposed RL method's R-TE across the testing years  are 4.129E+01 and 5.692E+00, respectively. The mean and standard error of the MM method's R-TE across the testing years  are 3.355E+02 and 5.510E+01, respectively.}
	\label{fig:vb_djia}
\end{figure}

\subsection{Transaction Costs}

Transaction costs have already been deducted from the portfolio values when we rebalance the portfolio, so there is no need to directly compare transaction costs of the RL method with those of other methods. Nonetheless, the transaction costs are informative for better understanding of the proposed RL method. The total transaction costs over the testing window $(t_1, t_2]$ is
\begin{equation}
	\text{TC} = \sum_{k=1}^{M (t_2 - t_1)} c_{t_1+\frac{k}{M}}, \notag
\end{equation}
where $c_{t_1+\frac{k}{M}}$ is the transaction cost defined in Equation \eqref{eq:tc}. Table \ref{tab:rb_sp500_tc} and \ref{tab:vb_sp500_tc} present the out-of-sample trading volumes and transaction costs of the RL method and the MM method respectively for the return-based tracking and the value-based tracking of the S\&P 500 index. The tables show that the trading volumes and transaction costs of the RL method are about 11.28\% (resp., 43.92\%) and 10.34\% (resp., 42.73\%) higher than those of the MM method for the return-based (resp., value-based) tracking of S\&P 500 index, respectively.
The tables also show that the transaction costs incurred in the RL method account for about 4.87 (resp., 7.06) basis point of the initial portfolio value for the return-based (resp., value-based) tracking of S\&P 500 index. 

The results for the S\&P 500 EWI and DJIA index are similar; see Table \ref{tab:rb_ewi_tc} to \ref{tab:vb_djia_tc} in Appendix \ref{appendix:tc_ewi_djia} for details.


\begin{table}[htbp]
	\caption
	{\textbf{Out-of-sample trading volume (Vol) and transaction costs (TC) of the proposed RL method and the MM method  for the return-based tracking of S\&P 500 index}. The trading volume is in terms of number of shares. The initial wealth $V_{0-}$ of each testing year is USD 20 billion. In the last two rows, mean and stderr respectively refer to the mean and standard error of the values across the testing years from 2005 to 2021 in each column.	
	}
	{\begin{tabular}{crrrrrr}
			\toprule
			\multirow{2}{*}{Testing year} & \multicolumn{2}{c}{Vol} & \multicolumn{2}{c}{TC} & \multicolumn{2}{c}{TC/$V_{0-}$} \\
			\cmidrule(lr){2-3}\cmidrule(lr){4-5}\cmidrule(lr){6-7}
			& \multicolumn{1}{c}{RL} & \multicolumn{1}{c}{MM} & \multicolumn{1}{c}{RL} & \multicolumn{1}{c}{MM} & \multicolumn{1}{c}{RL} & \multicolumn{1}{c}{MM} \\
			\midrule
			2005  & 1.726E+09 & 1.493E+09 & 8.630E+06 & 7.466E+06 & 4.315E-04 & 3.733E-04 \\
			2006  & 1.740E+09 & 1.777E+09 & 8.700E+06 & 8.883E+06 & 4.350E-04 & 4.441E-04 \\
			2007  & 1.676E+09 & 1.433E+09 & 8.381E+06 & 7.165E+06 & 4.191E-04 & 3.583E-04 \\
			2008  & 3.487E+09 & 3.469E+09 & 1.743E+07 & 1.735E+07 & 8.717E-04 & 8.674E-04 \\
			2009  & 4.725E+09 & 4.948E+09 & 2.362E+07 & 2.474E+07 & 1.181E-03 & 1.237E-03 \\
			2010  & 2.202E+09 & 1.867E+09 & 1.101E+07 & 9.337E+06 & 5.506E-04 & 4.669E-04 \\
			2011  & 2.093E+09 & 1.844E+09 & 1.047E+07 & 9.221E+06 & 5.233E-04 & 4.611E-04 \\
			2012  & 2.457E+09 & 2.211E+09 & 1.229E+07 & 1.105E+07 & 6.143E-04 & 5.527E-04 \\
			2013  & 1.703E+09 & 1.265E+09 & 8.513E+06 & 6.324E+06 & 4.257E-04 & 3.162E-04 \\
			2014  & 1.233E+09 & 1.090E+09 & 6.163E+06 & 5.453E+06 & 3.081E-04 & 2.727E-04 \\
			2015  & 1.296E+09 & 1.206E+09 & 6.482E+06 & 6.029E+06 & 3.241E-04 & 3.014E-04 \\
			2016  & 1.431E+09 & 1.231E+09 & 7.157E+06 & 6.155E+06 & 3.578E-04 & 3.078E-04 \\
			2017  & 1.149E+09 & 1.073E+09 & 5.744E+06 & 5.364E+06 & 2.872E-04 & 2.682E-04 \\
			2018  & 1.420E+09 & 1.115E+09 & 6.707E+06 & 5.575E+06 & 3.354E-04 & 2.787E-04 \\
			2019  & 2.224E+09 & 1.307E+09 & 1.008E+07 & 6.536E+06 & 5.039E-04 & 3.268E-04 \\
			2020  & 1.872E+09 & 1.807E+09 & 9.362E+06 & 9.038E+06 & 4.681E-04 & 4.519E-04 \\
			2021  & 9.367E+08 & 8.476E+08 & 4.684E+06 & 4.239E+06 & 2.342E-04 & 2.120E-04 \\
			\midrule
			mean  & 1.963E+09 & 1.764E+09 & 9.731E+06 & 8.819E+06 & 4.865E-04 & 4.410E-04 \\
			stderr & 2.258E+08 & 2.477E+08 & 1.130E+06 & 1.239E+06 & 5.650E-05 & 6.194E-05 \\
			\bottomrule
	\end{tabular}}
	\label{tab:rb_sp500_tc}
\end{table}

\begin{table}[htbp]
	\caption
	{\textbf{Out-of-sample trading volume (Vol) and transaction costs (TC) of the proposed RL method and the MM method for the value-based tracking of S\&P 500 index}.
	The trading volume is in terms of number of shares. The initial wealth $V_{0-}$ of each testing year is USD 20 billion. In the last two rows, mean and stderr respectively refer to the mean and standard error of the values across the testing years from 2005 to 2021 in each column.	
	  }
	{\begin{tabular}{crrrrrr}
			\toprule
			\multirow{2}{*}{Testing year} & \multicolumn{2}{c}{Vol} & \multicolumn{2}{c}{TC} & \multicolumn{2}{c}{TC/$V_{0-}$} \\
			\cmidrule(lr){2-3}\cmidrule(lr){4-5}\cmidrule(lr){6-7}
			& \multicolumn{1}{c}{RL} & \multicolumn{1}{c}{MM} & \multicolumn{1}{c}{RL} & \multicolumn{1}{c}{MM} & \multicolumn{1}{c}{RL} & \multicolumn{1}{c}{MM} \\
			\midrule
			2005  & 2.633E+09 & 1.794E+09 & 1.315E+07 & 8.972E+06 & 6.574E-04 & 4.486E-04 \\
			2006  & 2.678E+09 & 2.117E+09 & 1.337E+07 & 1.058E+07 & 6.687E-04 & 5.292E-04 \\
			2007  & 2.329E+09 & 1.734E+09 & 1.166E+07 & 8.671E+06 & 5.829E-04 & 4.336E-04 \\
			2008  & 5.538E+09 & 4.772E+09 & 2.599E+07 & 2.356E+07 & 1.299E-03 & 1.178E-03 \\
			2009  & 7.714E+09 & 7.046E+09 & 3.768E+07 & 3.382E+07 & 1.884E-03 & 1.691E-03 \\
			2010  & 3.057E+09 & 1.931E+09 & 1.531E+07 & 9.656E+06 & 7.653E-04 & 4.828E-04 \\
			2011  & 2.825E+09 & 2.117E+09 & 1.417E+07 & 1.058E+07 & 7.086E-04 & 5.292E-04 \\
			2012  & 2.986E+09 & 1.787E+09 & 1.496E+07 & 8.936E+06 & 7.482E-04 & 4.468E-04 \\
			2013  & 2.323E+09 & 1.442E+09 & 1.162E+07 & 7.211E+06 & 5.812E-04 & 3.605E-04 \\
			2014  & 1.687E+09 & 1.124E+09 & 8.441E+06 & 5.621E+06 & 4.220E-04 & 2.811E-04 \\
			2015  & 1.787E+09 & 1.421E+09 & 8.957E+06 & 7.107E+06 & 4.479E-04 & 3.554E-04 \\
			2016  & 2.423E+09 & 1.224E+09 & 1.216E+07 & 6.120E+06 & 6.078E-04 & 3.060E-04 \\
			2017  & 2.120E+09 & 1.022E+09 & 1.062E+07 & 5.113E+06 & 5.309E-04 & 2.557E-04 \\
			2018  & 2.086E+09 & 1.044E+09 & 9.997E+06 & 5.154E+06 & 4.998E-04 & 2.577E-04 \\
			2019  & 2.692E+09 & 1.135E+09 & 1.223E+07 & 5.675E+06 & 6.115E-04 & 2.838E-04 \\
			2020  & 2.840E+09 & 1.499E+09 & 1.433E+07 & 7.438E+06 & 7.166E-04 & 3.719E-04 \\
			2021  & 1.505E+09 & 7.689E+08 & 7.586E+06 & 3.847E+06 & 3.793E-04 & 1.923E-04 \\
			\midrule
			mean  & 2.895E+09 & 1.999E+09 & 1.425E+07 & 9.886E+06 & 7.125E-04 & 4.943E-04 \\
			stderr & 3.696E+08 & 3.822E+08 & 1.768E+06 & 1.836E+06 & 8.839E-05 & 9.178E-05 \\
			\bottomrule
	\end{tabular}}
	\label{tab:vb_sp500_tc}
\end{table}

\section{Conclusion}

In this paper, we propose the first discrete-time infinite-horizon dynamic formulation of the index tracking problem under both return-based tracking error and value-based tracking error. The dynamic formulation overcomes several limitations of the models in the existing literature, all of which are based on static formulations expect one continuous time stochastic control model. $(i)$ The dynamic formulation relies on the intertemporal dynamics of the daily returns and other market variables and hence utilizes market data in a much longer historical time period. $(ii)$ The dynamic formulation learns the optimal functional relation between the market information and the portfolio weights at the rebalancing time, so the portfolio weights are responsive to the current market situation.  
$(iii)$ The formulation directly incorporates and accurately computes transaction costs under various specifications, particularly those practically specified as nonlinear functions of trading volume. 
$(iv)$ The formulation exactly incorporates the overall tradeoff between tracking error and transaction costs over an infinite time horizon. $(v)$ The formulation avoids cardinality constraints or penalty terms resulting in NP-hardness because it directly incorporates the exact transaction costs. $(vi)$ The formulation effectively uses rich market information other than prices or returns by including them into state variables. 

For the value-based tracking problem, we allows the novel inclusion of cash injection or withdraw decision variables in the formulation, which leads to flexible index tracking strategies. 

To accurately incorporate the transaction cost at each rebalancing time, we propose a simple but efficient Banach fixed point iteration algorithm to solve the rebalancing equation, leading to accurate calculation of the transaction cost. The algorithm can be applied in general portfolio management problems to incorporate rebalancing cost. 

The dynamic formulation of index tracking problem is difficult to solve due to high dimensional state variables and decision variables (i.e., the portfolio weights). We propose an extension of deep reinforcement learning (RL) method that allows stochastic control policies specifying the distribution of portfolio weights conditional on state variables.
Our RL method extends the existing RL methods used in computer games or robotics in two aspects. $(i)$ To overcome the issue of data limitation resulting from a single sample path of daily financial data, we propose a new training scheme that generates enough quarterly or semiyearly data from daily data by randomly drawing the starting date of the training episode from the date range of training data set. $(ii)$ For computer game or robotic problems, the value function 
at the terminal state of an training episode is typically zero. In contrast, 
our RL method needs to estimate the value function at the terminal state as the index tracking problem is an infinite-horizon problem.

We carry out comprehensive empirical study of the out-of-sample performance of our RL method by using 
a 17-year-long testing set. The index tracking results for the value-weighted S\&P 500 index, the equally-weighted S\&P 500 Equal Weight Index, and the price-weighted Dow Jones Industrial Average index show that the proposed RL method not only achieves smaller mean return tracking error and one magnitude smaller mean value tracking error compared with a benchmark method, but also produces a positive mean cash withdraw, which implies that the strategy of cash injection and withdraw may earn extra profit for the fund manager.

%
%
%
%
%
%
%
%
%
%
\begin{appendices}
	
\section{Proofs}

\subsection{Proof of Proposition \ref{prop:achieve_target_weight}}\label{appendix:proof_prop_achieve_target_weight}

To prepare, we first give the following four lemmas.

\begin{lemma}\label{lemma:min_inequality}
	For any $a_1, a_2, b_1, b_2 \in \mathbb{R}$, $| \mathop{\min} \left\{a_1, a_2\right\} - \mathop{\min} \left\{b_1, b_2\right\} | \le \left|a_1 - b_1 \right| + \left|a_2 - b_2 \right|$.
\end{lemma}

\begin{proof}
For any $a, b \in \mathbb{R}$, $\mathop{\min} \left\{ a, b \right\} = \frac{1}{2} \left(a + b - \left| a - b \right|\right)$. Then, for any $a_1, a_2, b_1, b_2 \in \mathbb{R}$, we have
\begin{align}
	&\ \ \ \ \left| \mathop{\min} \left\{a_1, a_2\right\} - \mathop{\min} \left\{b_1, b_2\right\} \right| \notag \\
	&= \left| \frac{1}{2} \left(a_1 + a_2 - \left| a_1 - a_2 \right|\right) -  \frac{1}{2} \left(b_1 + b_2 - \left| b_1 - b_2 \right|\right) \right| \notag \\
	&= \left| \frac{1}{2} \left(a_1 - b_1\right) + \frac{1}{2} \left(a_2 - b_2\right)  + \frac{1}{2} \left(\left| b_1 - b_2 \right| - \left| a_1 - a_2 \right|\right) \right| \notag \\
	&\le  \frac{1}{2} \left| a_1 - b_1 \right| + \frac{1}{2} \left| a_2 - b_2 \right|  + \frac{1}{2} \left| \left| b_1 - b_2 \right| - \left| a_1 - a_2 \right| \right| \notag \\
	&\le \frac{1}{2} \left| a_1 - b_1 \right| + \frac{1}{2} \left| a_2 - b_2 \right|  + \frac{1}{2} \left| \left(b_1 - b_2 \right) - \left( a_1 - a_2 \right) \right| \notag \\
	&= \frac{1}{2} \left| a_1 - b_1 \right| + \frac{1}{2} \left| a_2 - b_2\right|  + \frac{1}{2} \left| \left(b_1 - a_1 \right) + \left( a_2 - b_2 \right) \right| \notag \\
	&\le \frac{1}{2} \left| a_1 - b_1 \right| + \frac{1}{2} \left| a_2 - b_2\right|  + \frac{1}{2} \left| b_1 - a_1 \right| + \frac{1}{2} \left| a_2 - b_2 \right| \notag \\
	&= \left|a_1 - b_1 \right| + \left|a_2 - b_2 \right|. \notag
\end{align}
\end{proof}

\begin{lemma}\label{lemma:max_inequality}
	For any $a_1, a_2, b_1, b_2 \in \mathbb{R}$, $\left| \mathop{\max} \left\{ \left| a_1 \right|, \left| a_2 \right| \right\} - \mathop{\max} \left\{\left| b_1 \right|, \left| b_2 \right| \right\} \right| \le \mathop{\max} \left\{ \left|a_1 - b_1 \right|, \left|a_2 - b_2 \right| \right\}$.
\end{lemma}
\begin{proof}
For $i = 1, 2$, $\left| a_i \right| = \left| a_i - b_i + b_i \right| \le \left| a_i - b_i \right| + \left| b_i \right|$. Hence, 
$$\mathop{\max} \left\{ \left| a_1 \right|, \left| a_2 \right| \right\} \le \mathop{\max} \left\{ \left| a_1 - b_1 \right|, \left| a_2 - b_2 \right| \right\} + \mathop{\max} \left\{ \left| b_1 \right|, \left| b_2 \right| \right\},$$ 
which implies 
\begin{align}
	&\mathop{\max} \left\{ \left| a_1 \right|, \left| a_2 \right| \right\} - \mathop{\max} \left\{ \left| b_1 \right|, \left| b_2 \right| \right\} \le \mathop{\max} \left\{ \left| a_1 - b_1 \right|, \left| a_2 - b_2 \right| \right\}. \notag
\end{align}
By symmetry,
\begin{align}
	&\mathop{\max} \left\{ \left| b_1 \right|, \left| b_2 \right| \right\} - \mathop{\max} \left\{ \left| a_1 \right|, \left| a_2 \right| \right\} \le \mathop{\max} \left\{ \left| b_1 - a_1 \right|, \left| b_2 - a_2 \right| \right\}. \notag
\end{align}
Therefore, the conclusion follows. 
\end{proof}

\begin{lemma}\label{lemma:contraction}
	Suppose $\xi_{1i} > 0$, $\xi_{2i} > 0$, $\xi_{3i} > 0$, $p_i > 0$, $V_i \in \mathbb{R}$, $w_i \in \mathbb{R}$, $i=1,\ldots, N$. Define function
	\begin{align}
		\Psi(x) :=  \sum_{i=1}^{N} \mathop{\min} \left \{ \mathop{\max} \left \{ \frac{\xi_{1i}}{p_i} \left | V_i - w_i x \right |, \; \xi_{3i} \right \}, \; \xi_{2i} \left | V_i - w_i x \right | \right \}, x\in \mathbb{R}. \notag
	\end{align}
	If $\sum_{i=1}^{N} \left( \frac{\xi_{1i}}{p_i} + \xi_{2i} \right)\left| w_i \right|  < 1$, then $\Psi$ is a contraction on $\mathbb{R}$ with a contraction coefficient $\sum_{i=1}^{N} \left( \frac{\xi_{1i}}{p_i} + \xi_{2i} \right)\left| w_i \right|$.
\end{lemma}
\begin{proof}
For all $x, y \in \mathbb{R}$,
\allowdisplaybreaks
\begin{align}\label{eq:tr_cost_k}
	& | \Psi(x) - \Psi(y) | \notag \\
	\le& \sum_{i=1}^{N} \bigg| \mathop{\min} \left \{ \mathop{\max} \left \{ \frac{\xi_{1i}}{p_i} \left | V_i - w_i y \right |, \; \xi_{3i} \right \}, \; \xi_{2i} \left | V_i - w_i y \right | \right \}\notag\\
	& \quad\quad\quad\quad - \mathop{\min} \left \{ \mathop{\max} \left \{ \frac{\xi_{1i}}{p_i} \left | V_i - w_i x \right |, \; \xi_{3i} \right \}, \; \xi_{2i} \left | V_i - w_i x \right | \right \} \bigg| \notag \\
	\le& \sum_{i=1}^{N} \left| \mathop{\max} \left \{ \frac{\xi_{1i}}{p_i} \left | V_i - w_i y \right |, \; \xi_{3i} \right \} - \mathop{\max} \left \{ \frac{\xi_{1i}}{p_i} \left | V_i - w_i x \right |, \; \xi_{3i} \right \} \right| + | \xi_{2i} \left | V_i - w_i y \right | - \xi_{2i} \left | V_i - w_i x \right | |\notag\\
	& \quad\quad\quad\quad (\text{by Lemma} \ \ref{lemma:min_inequality}) \notag \\
	\le& \sum_{i=1}^{N} \left| \mathop{\max} \left \{ \frac{\xi_{1i}}{p_i} \left | V_i - w_i y \right |, \; \xi_{3i} \right \} - \mathop{\max} \left \{ \frac{\xi_{1i}}{p_i} \left | V_i - w_i x \right |, \; \xi_{3i} \right \} \right|  + \sum_{i=1}^{N} \xi_{2i} \left| w_i \right| \left| x - y \right|  \notag \\
	\le& \sum_{i=1}^{N} \mathop{\max} \left \{ \left| \frac{\xi_{1i}}{p_i} \left| V_i - w_i y \right| - \frac{\xi_{1i}}{p_i} \left| V_i - w_i x \right| \right|, \; \xi_{3i} - \xi_{3i} \right \} + \sum_{i=1}^{N} \xi_{2i} \left| w_i \right| \left| x - y \right| \ (\text{by Lemma} \ \ref{lemma:max_inequality}) \notag \\
	\le& \left| x - y \right| \sum_{i=1}^{N} \left( \frac{\xi_{1i}}{p_i} + \xi_{2i} \right ) \left| w_i \right|, 
\end{align}
which leads to the conclusion.
\end{proof}

\begin{lemma}\label{lemma:portfolio_value_equaility}
	At any rebalancing time $t + \frac{k}{M}$, $k = 0, \dots, M-1$, $t = 0, 1, 2, \dots$, the portfolio value right before and right after the rebalancing satisfy the rebalancing equation
	\begin{equation}\label{eq:lemma_portfolio_value_equaility_eq1}
		V_{t+\frac{k}{M}} = V_{(t+\frac{k}{M})-} + h_{t+\frac{k}{M}} - c_{t+\frac{k}{M}} .
	\end{equation}
\end{lemma}
\begin{proof}
Right before the rebalancing at time $t + \frac{k}{M}$, the fund manager holds stock shares $\left(x_{1, (t+\frac{k}{M})-}, \dots, x_{N, (t+\frac{k}{M})-}\right)$ and cash $H_{(\tk)-}$; right after the rebalancing, the fund manager holds stock shares $\left(x_{1, t+\frac{k}{M}}, \dots, x_{N, t+\frac{k}{M}}\right)$ and cash $H_{\tk}$. 
We denote the negative part and positive part functions as $x^-=\max(0, -x)$ and $x^+=\max(0, x)$ for $x\in\mathbb{R}$, respectively. 
Counting cash inflow and outflow, we have 
%
%
\begin{align}\label{equ:cash_balance_equation}
	H_{(\tk)-}+\sum_{i=1}^{N} p_{i, t+\frac{k}{M}} \left(x_{i, t+\frac{k}{M}} - x_{i, (t+\frac{k}{M})-}\right)^- + h_{t+\frac{k}{M}} - & \sum_{i=1}^{N} p_{i, t+\frac{k}{M}} \left(x_{i, t+\frac{k}{M}} - x_{i, (t+\frac{k}{M})-}\right)^+ \notag \\
	 - & c_{t+\frac{k}{M}} = H_{\tk}, 
\end{align}
where, on the left-hand side, the second term and the forth term are the cash inflow and outflow resulting from selling and buying stocks, respectively. Hence, 
\begin{align}
	& h_{t+\frac{k}{M}} - c_{t+\frac{k}{M}}\notag\\
	= & \sum_{i=1}^{N} p_{i, t+\frac{k}{M}} \left(x_{i, t+\frac{k}{M}} - x_{i, (t+\frac{k}{M})-}\right)^+ - \sum_{i=1}^{N} p_{i, t+\frac{k}{M}} \left(x_{i, t+\frac{k}{M}} - x_{i, (t+\frac{k}{M})-}\right)^- + H_{\tk} - H_{(\tk)-}\notag \\
	= & \sum_{i=1}^{N} p_{i, t+\frac{k}{M}} x_{i, t+\frac{k}{M}} + H_{\tk} - \sum_{i=1}^{N} p_{i, t+\frac{k}{M}} x_{i, (t+\frac{k}{M})-}  - H_{(\tk)-}\notag \\
	= & V_{t+\frac{k}{M}} - V_{(t+\frac{k}{M})-}. \notag
\end{align}
\end{proof}

We next prove Proposition \ref{prop:achieve_target_weight}.

\begin{proof}[Proof of Proposition \ref{prop:achieve_target_weight}]

We first prove part $(ii)$. Since $\sum_{i=1}^{N} ( \frac{\xi_{1i}}{p_{i, t+\frac{k}{M}}} + \xi_{2i} ) | w_{i, t+\frac{k}{M}} | < 1$ and $V_{(t+\frac{k}{M})-} + h_{t+\frac{k}{M}}$ is a constant, $\Psi(V_{t+\frac{k}{M}}) = V_{(t+\frac{k}{M})-} + h_{t+\frac{k}{M}} - c_{t+\frac{k}{M}}(V_{t+\frac{k}{M}})$ is a contraction by Lemma \ref{lemma:contraction}. By the Banach Fixed Point Theorem, there exists a unique $V^*_{t+\frac{k}{M}} \in \mathbb{R}$ such that
\begin{equation}
	V^*_{t+\frac{k}{M}} = V_{(t+\frac{k}{M})-}  + h_{t+\frac{k}{M}} -  c_{t+\frac{k}{M}}(V^*_{t+\frac{k}{M}}), \notag
\end{equation}
which implies by Lemma \ref{lemma:portfolio_value_equaility} that $V^*_{t+\frac{k}{M}}$ is the unique solution to the rebalancing equation \eqref{equ:V_fixed_point_prob_cash_injection}. The next step is to show that $0 < V^*_{t+\frac{k}{M}} \le V_{(t+\frac{k}{M})-} + h_{t+\frac{k}{M}}$. By definition, $c_{t+\frac{k}{M}} \ge 0$, so $V^*_{t+\frac{k}{M}} \le V_{(t+\frac{k}{M})-} + h_{t+\frac{k}{M}}$. To prove the first inequality, without loss of generality, suppose there exists an integer $\tilde{N} \le N$ such that $i \le \tilde{N}$ implies
\begin{align}
	V_{i,(t+\frac{k}{M})-} \ge w_{i,t+\frac{k}{M}} V^*_{t+\frac{k}{M}}, \notag
\end{align}
and $\tilde{N} < i \le N$ implies
\begin{align}
	V_{i,(t+\frac{k}{M})-} < w_{i,t+\frac{k}{M}} V^*_{t+\frac{k}{M}}. \notag
\end{align}
Then
\begin{align}
	& c_{t+\frac{k}{M}}(V^*_{t+\frac{k}{M}}) \notag\\
	= & \sum_{i=1}^{N} \mathop{\min} \left \{ \mathop{\max} \left \{ \frac{\xi_{1i}}{p_{i,t+\frac{k}{M}}} \left|V_{i,(t+\frac{k}{M})-} - w_{i,t+\frac{k}{M}} V^*_{t+\frac{k}{M}} \right|, \; \xi_{3i} \right \}, \xi_{2i} \left|V_{i,(t+\frac{k}{M})-} - w_{i,t+\frac{k}{M}} V^*_{t+\frac{k}{M}} \right| \right \} \notag \\
	\le & \sum_{i=1}^{N} \xi_{2i} \left|V_{i,(t+\frac{k}{M})-} - w_{i,t+\frac{k}{M}} V^*_{t+\frac{k}{M}} \right| \notag \\
	= & \sum_{i=1}^{\tilde{N}} \xi_{2i} \left(V_{i,(t+\frac{k}{M})-} - w_{i,t+\frac{k}{M}} V^*_{t+\frac{k}{M}} \right) + \sum_{i=\tilde{N}+1}^{N} \xi_{2i} \left(w_{i,t+\frac{k}{M}} V^*_{t+\frac{k}{M}}  - V_{i,(t+\frac{k}{M})-}\right), \notag
\end{align}
and
\begin{align}
	V^*_{t+\frac{k}{M}} \ge V_{(t+\frac{k}{M})-} + h_{t+\frac{k}{M}} - \sum_{i=1}^{\tilde{N}} \xi_{2i} \left(V_{i,(t+\frac{k}{M})-} - w_{i,t+\frac{k}{M}} V^*_{t+\frac{k}{M}} \right) - \sum_{i=\tilde{N}+1}^{N} \xi_{2i} \left(w_{i,t+\frac{k}{M}} V^*_{t+\frac{k}{M}}  - V_{i,(t+\frac{k}{M})-}\right), \notag
\end{align}
which implies
\begin{align}\label{eq:proposition1_eq1}
	& \left(1 - \sum_{i=1}^{\tilde{N}} \xi_{2i} w_{i,t+\frac{k}{M}} + \sum_{i=\tilde{N}+1}^{N} \xi_{2i} w_{i,t+\frac{k}{M}} \right) V^*_{t+\frac{k}{M}}\notag\\
	 \ge & \left(1 - \sum_{i=1}^{\tilde{N}} \xi_{2i} w_{i,(t+\frac{k}{M})-} + \sum_{i=\tilde{N}+1}^{N} \xi_{2i} w_{i,(t+\frac{k}{M})-} \right) V_{(t+\frac{k}{M})-} + h_{t+\frac{k}{M}}.
\end{align}

Notice that by the conditions of the proposition,
\begin{align}\label{eq:proposition1_eq2}
	\sum_{i=1}^{\tilde{N}} \xi_{2i} w_{i,t+\frac{k}{M}} - \sum_{i=\tilde{N}+1}^{N} \xi_{2i} w_{i,t+\frac{k}{M}} \le \sum_{i=1}^{N} \xi_{2i} \left| w_{i,t+\frac{k}{M}} \right| < 1,
\end{align}
and
\begin{align}\label{eq:proposition1_eq3}
	\sum_{i=1}^{\tilde{N}} \xi_{2i} w_{i,(t+\frac{k}{M})-}  - \sum_{i=\tilde{N}+1}^{N} \xi_{2i} w_{i,(t+\frac{k}{M})-}  \le \sum_{i=1}^{N} \xi_{2i} \left| w_{i,(t+\frac{k}{M})-} \right| < 1.
\end{align}

By the conditions of the proposition, 
\begin{align}
	h_{t+\frac{k}{M}} &\ge \ -\xi V_{(t+\frac{k}{M})-} > -\left(1 - \sum_{i=1}^{N} \xi_{2i} \left| w_{i,(t+\frac{k}{M})-} \right| \right) V_{(t+\frac{k}{M})-}, \notag
\end{align}
which, together with \eqref{eq:proposition1_eq1} to \eqref{eq:proposition1_eq3} and $V_{(t+\frac{k}{M})-} > 0$, implies that $V^*_{t+\frac{k}{M}} > 0$.

We then prove part $(i)$. Suppose the conditions in part $(i)$ hold. Let $h_{t+\frac{k}{M}} = 0$ and define $\xi = \frac{1}{2}(1 - \sum_{i=1}^{N} \xi_{2i} |w_{i,(t+\frac{k}{M})-}|)$. Then the conditions in part $(ii)$ hold. Applying part $(ii)$ leads to the conclusion of part $(i)$. 
\end{proof}

Finally, we provide a lemma that is useful for calculating $w_{i, (t+\frac{k}{M})-}$ from $w_{i, t+\frac{k-1}{M}}$. The weight $w_{i, (t+\frac{k}{M})-}$ is needed for verifying the conditions of Proposition \ref{prop:achieve_target_weight}.

\begin{lemma}\label{lemma:relation_of_weights}
	For any stock $i = 1, \dots, N$ and any time $t + \frac{k}{M} > 0$, the weight $w_{i, (t+\frac{k}{M})-}$ and $w_{i, t+\frac{k-1}{M}}$ have the following relationship
	\allowdisplaybreaks
	\begin{align}
		& w_{i, (t+\frac{k}{M})-} = \frac{w_{i, t+\frac{k-1}{M}} p_{i,t+\frac{k}{M}}}{p_{i,t+\frac{k-1}{M}} \sum_{j=1}^{N} w_{j,t+\frac{k-1}{M}} \frac{p_{j,t+\frac{k}{M}}}{p_{j,t+\frac{k-1}{M}}}}. \notag
	\end{align}
\end{lemma}

\begin{proof}
For any stock $i$ and any rebalancing time $t + \frac{k}{M} > 0$, the number of shares right before the rebalancing is
\begin{equation}
	x_{i,(t+\frac{k}{M})-} = x_{i, t+\frac{k-1}{M}} = \frac{w_{i,t+\frac{k-1}{M}} V_{t+\frac{k-1}{M}}}{p_{i,t+\frac{k-1}{M}}}, \notag
\end{equation}
then the value of stock $i$ right before the rebalancing can be represented as
\begin{equation}
	V_{i,(t+\frac{k}{M})-} = p_{i,t+\frac{k}{M}} x_{i,(t+\frac{k}{M})-} = V_{t+\frac{k-1}{M}} \frac{w_{i,t+\frac{k-1}{M}} p_{i,t+\frac{k}{M}}}{p_{i,t+\frac{k-1}{M}}}, \notag
\end{equation}
and then the portfolio value right before the rebalancing is
\begin{equation}
	V_{(t+\frac{k}{M})-} = \sum_{i=1}^{N} V_{i,(t+\frac{k}{M})-} = V_{t+\frac{k-1}{M}}\sum_{i=1}^{N} w_{i,t+\frac{k-1}{M}} \frac{p_{i,t+\frac{k}{M}}}{p_{i,t+\frac{k-1}{M}}}. \notag
\end{equation}
Hence, the weight of stock $i$ right before the rebalancing is
\begin{align}
	w_{i,(t+\frac{k}{M})-} &= \frac{V_{i,(t+\frac{k}{M})-}}{V_{(t+\frac{k}{M})-}} = \frac{V_{t+\frac{k-1}{M}} \frac{w_{i,t+\frac{k-1}{M}} p_{i,t+\frac{k}{M}}}{p_{i,t+\frac{k-1}{M}}}}{V_{t+\frac{k-1}{M}}\sum_{j=1}^{N} w_{j,t+\frac{k-1}{M}} \frac{p_{j,t+\frac{k}{M}}}{p_{j,t+\frac{k-1}{M}}}} =  \frac{w_{i,t+\frac{k-1}{M}} p_{i,t+\frac{k}{M}}}{p_{i,t+\frac{k-1}{M}} \sum_{j=1}^{N} w_{j,t+\frac{k-1}{M}} \frac{p_{j,t+\frac{k}{M}}}{p_{j,t+\frac{k-1}{M}}}}. \notag
\end{align}
\end{proof}

\subsection{Transaction Costs with Regulatory Fees}\label{appendix:regulatory_fees}

In this section, we take into account regulatory fees into the transaction cost. More specifically, we consider two regulatory fees: the SEC transaction fee and the FINRA trading activity fee, which are 
respectively proportional to the value of aggregate sales and 
to the quantity sold per order. More precisely, the two fees are respectively defined as
\begin{align}\label{eq:SEC_fee}
	& c^{\text{S}}_{t + \frac{k}{M}}(V_{t + \frac{k}{M}}) = \nu_1 \sum_{i = 1}^{N} \left(w_{i, t + \frac{k}{M}} V_{t + \frac{k}{M}} - V_{i, (t + \frac{k}{M})-}  \right)^-,\\
	& c^{\text{F}}_{t + \frac{k}{M}}(V_{t + \frac{k}{M}}) = \nu_2 \sum_{i = 1}^{N} \left( \frac{w_{i, t + \frac{k}{M}} V_{t + \frac{k}{M}} - V_{i, (t + \frac{k}{M})-}}{p_{i, t + \frac{k}{M}}} \right)^-,\label{eq:FINRA_fee}
\end{align}
where $\nu_1, \nu_2 > 0$ are two constants. 

The following proposition shows that the conclusion in Proposition 
\ref{prop:achieve_target_weight} still holds when the transaction cost include regulatory fees.

\begin{proposition}\label{prop:regulatory_fees}
	Let the total transaction cost be
	\begin{equation}\label{eq:tc_with_reg_fee}
		\tilde c_{t+\frac{k}{M}}(V_{t+\frac{k}{M}}) = c_{t+\frac{k}{M}}(V_{t+\frac{k}{M}}) + c^{\text{S}}_{t+\frac{k}{M}}(V_{t+\frac{k}{M}}) + c^{\text{F}}_{t+\frac{k}{M}}(V_{t+\frac{k}{M}}),
	\end{equation}
	where $c_{t+\frac{k}{M}}(V_{t+\frac{k}{M}})$, $c^{\text{S}}_{t+\frac{k}{M}}(V_{t+\frac{k}{M}})$, and $c^{\text{F}}_{t+\frac{k}{M}}(V_{t+\frac{k}{M}})$ are respectively defined in \eqref{eq:tc}, \eqref{eq:SEC_fee}, and \eqref{eq:FINRA_fee}. 
	
	($i$) If $V_{(t+\frac{k}{M})-}>0$, $\sum_{i=1}^{N}( \xi_{2i} + \nu_1 + \frac{\nu_2}{p_{i, t + \frac{k}{M}}} ) | w_{i,(t+\frac{k}{M})-} |<1$, and
	$\sum_{i=1}^{N} ( \frac{\xi_{1i}}{p_{i, t + \frac{k}{M}}} + \xi_{2i} + \nu_1 + \frac{\nu_2}{p_{i, t + \frac{k}{M}}} ) |w_{i, t + \frac{k}{M}}| < 1$, then 
	the function $\Psi_1(V_{t+\frac{k}{M}}) = V_{(t+\frac{k}{M})-} - \tilde c_{t+\frac{k}{M}}(V_{t+\frac{k}{M}})$ is a contraction on $\mathbb{R}$ with 
	the contraction coefficient $\sum_{i=1}^{N} ( \frac{\xi_{1i}}{p_{i, t + \frac{k}{M}}} + \xi_{2i} + \nu_1 + \frac{\nu_2}{p_{i, t + \frac{k}{M}}} ) |w_{i, t + \frac{k}{M}}|$ and a unique fixed point $V^*_{t+\frac{k}{M}}$, which is the unique solution to the rebalancing equation $\Psi_1(V_{t+\frac{k}{M}})=V_{t+\frac{k}{M}}$. In addition, the solution satisfies $0 < V^*_{t+\frac{k}{M}} \le V_{(t+\frac{k}{M})-}$.
	
	($ii$) If all the conditions in (i) hold, and $0 < \xi < 1 - \sum_{i=1}^{N} ( \xi_{2i} + \nu_1 + \frac{\nu_2}{p_{i, t + \frac{k}{M}}} ) | w_{i,(t+\frac{k}{M})-} |$, and
	$h_{t+\frac{k}{M}} \ge -\xi V_{(t+\frac{k}{M})-}$, then 
	the function $\Psi_2(V_{t+\frac{k}{M}}) = V_{(t+\frac{k}{M})-} + h_{t+\frac{k}{M}} - \tilde c_{t+\frac{k}{M}}(V_{t+\frac{k}{M}})$ is a contraction on $\mathbb{R}$ with 
	the contraction coefficient $\sum_{i=1}^{N} ( \frac{\xi_{1i}}{p_{i, t + \frac{k}{M}}} + \xi_{2i} + \nu_1 + \frac{\nu_2}{p_{i, t + \frac{k}{M}}} ) |w_{i, t + \frac{k}{M}}|$ and a unique fixed point $V^*_{t+\frac{k}{M}}$, which is the unique solution to the rebalancing equation $\Psi_2(V_{t+\frac{k}{M}})=V_{t+\frac{k}{M}}$. In addition, the solution satisfies $0 < V^*_{t+\frac{k}{M}} \le V_{(t+\frac{k}{M})-} + h_{t+\frac{k}{M}}$.
\end{proposition}

The conditions in Proposition \ref{prop:regulatory_fees} almost always hold as the constants $\nu_1$ and $\nu_2$ are usually sufficiently small, although they may vary from time to time. For example, $\nu_1 = 0.0000229$ and $\nu_2 = 0.00013$ are currently specified.\footnote{\url{https://www.interactivebrokers.com/en/index.php?f=49637}.}

To prove Proposition \ref{prop:regulatory_fees}, we first prove the following simple lemma.

\begin{lemma}\label{lemma:neg_part_ineq}
	For any $a, b \in \mathbb{R}$, $|a^- - b^-| \le |a - b|$. 
\end{lemma}

\begin{proof} We have
\begin{align}\label{regulatory_fees_lemma_le}
	a^- - b^- &=  \frac{1}{2} \left[ |a| - a - (|b|- b)\right] =  \frac{1}{2} \left[ (b - a) + |a| - |b| \right]\notag\\
	&\le  \frac{1}{2} \left[ (b - a) + |b - a| \right]= (b-a)^+ \le |a - b|.
\end{align}
By symmetry, we have $b^- - a^- \le |b - a|$, so the conclusion follows.
\end{proof}

Let us then prove Proposition \ref{prop:regulatory_fees}.

\begin{proof}[Proof of Proposition \ref{prop:regulatory_fees}]
We first prove part ($ii$). We will first show that each term on the right-hand side of \eqref{eq:tc_with_reg_fee} is a contraction on $\mathbb{R}$, and then show that $\Psi_2(x)$ is also a contraction on $\mathbb{R}$.

For any $x, y \in \mathbb{R}$,
\begin{align}\label{eq:c1_k}
	\left|  c^{\text{S}}_{t + \frac{k}{M}}(x) - c^{\text{S}}_{t + \frac{k}{M}}(y) \right| &= \nu_1 \left| \sum_{i=1}^{N} \left[ \left(w_{i, t + \frac{k}{M}} x - V_{i, (t + \frac{k}{M})-}  \right)^- -  \left(w_{i, t + \frac{k}{M}} y - V_{i, (t + \frac{k}{M})-}  \right)^- \right]\right| \notag \\
	&\le \sum_{i=1}^{N} \nu_1 \left| \left(w_{i, t + \frac{k}{M}} x - V_{i, (t + \frac{k}{M})-}  \right)^- -  \left(w_{i, t + \frac{k}{M}} y - V_{i, (t + \frac{k}{M})-}  \right)^- \right|  \notag \\
	&\le |x - y| \sum_{i=1}^{N} \nu_1 \left|w_{i, t + \frac{k}{M}}\right| \ (\text{by Lemma} \ \ref{lemma:neg_part_ineq}).
\end{align}

Thus, $c^{\text{S}}_{t + \frac{k}{M}}$ is a contraction as $\nu_1 \sum_{i=1}^{N}|w_{i, t + \frac{k}{M}}| < 1$.
Similarly, we have 
\begin{align}\label{eq:c2_k}
	\left|  c^{\text{F}}_{t + \frac{k}{M}}(x) - c^{\text{F}}_{t + \frac{k}{M}}(y) \right|
	&\le |x - y| \sum_{i=1}^{N} \frac{\nu_2}{p_{i, t + \frac{k}{M}}} \left|w_{i, t + \frac{k}{M}}\right|. 
\end{align}

By the conditions of the proposition, $\sum_{i=1}^{N} \frac{\nu_2}{p_{i, t + \frac{k}{M}}} |w_{i, t + \frac{k}{M}}| < 1$. Hence, $c^{\text{F}}_{t + \frac{k}{M}}$ is a contraction.

It follows from Equation \eqref{eq:tr_cost_k}, \eqref{eq:c1_k} and \eqref{eq:c2_k} that $|\tilde c_{t+\frac{k}{M}}(x) - \tilde c_{t+\frac{k}{M}}(y)| \le \alpha |x - y|$ for all $x, y \in \mathbb{R}$, where 
$\alpha =\sum_{i=1}^{N} ( \frac{\xi_{1i}}{p_{i, t + \frac{k}{M}}} + \xi_{2i} + \nu_1 + \frac{\nu_2}{p_{i, t + \frac{k}{M}}} )|w_{i, t + \frac{k}{M}}| < 1$ by the conditions of the proposition. Since $V_{(t+\frac{k}{M})-} + h_{t+\frac{k}{M}}$ is a constant, $\Psi_2(x)$ is a contraction on $\mathbb{R}$. By the Banach Fixed Point Theorem, $\Psi_2$ has a unique fixed point $V^*_{t+\frac{k}{M}}$, i.e., 
\begin{align}\label{eq:proposition2_rebal_eq}
	V^*_{t+\frac{k}{M}} = \Psi_2(V^*_{t+\frac{k}{M}})= V_{(t+\frac{k}{M})-} + h_{t+\frac{k}{M}} - \tilde c_{t+\frac{k}{M}}(V^*_{t+\frac{k}{M}}).
\end{align}
Therefore, $V^*_{t+\frac{k}{M}}$ is the unique solution to the rebalancing equation $\Psi_2(V_{t+\frac{k}{M}}) = V_{t+\frac{k}{M}}$. Since $\tilde c_{t+\frac{k}{M}}(V^*_{t+\frac{k}{M}})\geq 0$, it follows that 
$V^*_{t+\frac{k}{M}} = V_{(t+\frac{k}{M})-} + h_{t+\frac{k}{M}} - \tilde c_{t+\frac{k}{M}}(V^*_{t+\frac{k}{M}})\leq V_{(t+\frac{k}{M})-} + h_{t+\frac{k}{M}}$. The argument of showing $V^*_{t+\frac{k}{M}} > 0$ is similar to that in the proof of Proposition \ref{prop:achieve_target_weight}. Without loss of generality, suppose there exists an integer $\tilde{N} \le N$ such that $i \le \tilde{N}$ implies
\begin{align}
	V_{i,(t+\frac{k}{M})-} \ge w_{i,t+\frac{k}{M}} V^*_{t+\frac{k}{M}}, \notag
\end{align}
and $\tilde{N} < i \le N$ implies
\begin{align}
	V_{i,(t+\frac{k}{M})-} < w_{i,t+\frac{k}{M}} V^*_{t+\frac{k}{M}}. \notag
\end{align}
Then
\begin{align}
	\tilde{c}_{t+\frac{k}{M}}(V^*_{t+\frac{k}{M}}) \le& \sum_{i=1}^{\tilde{N}} \left( \xi_{2i} + \nu_1 + \frac{\nu_2}{p_{i, t + \frac{k}{M}}} \right) 
	\left(V_{i,(t+\frac{k}{M})-} - w_{i,t+\frac{k}{M}} V^*_{t+\frac{k}{M}} \right) \notag \\
	&+ \sum_{i=\tilde{N}+1}^{N} \left( \xi_{2i} + \nu_1 + \frac{\nu_2}{p_{i, t + \frac{k}{M}}} \right) \left(w_{i,t+\frac{k}{M}} V^*_{t+\frac{k}{M}}  - V_{i,(t+\frac{k}{M})-}\right), \notag
\end{align}
which implies by Equation \eqref{eq:proposition2_rebal_eq} that
\begin{align}\label{eq:proposition2_eq3}
	&\left[1 - \sum_{i=1}^{\tilde{N}} \left( \xi_{2i} + \nu_1 + \frac{\nu_2}{p_{i, t + \frac{k}{M}}} \right) w_{i,t+\frac{k}{M}} + \sum_{i=\tilde{N}+1}^{N} \left( \xi_{2i} + \nu_1 + \frac{\nu_2}{p_{i, t + \frac{k}{M}}} \right) w_{i,t+\frac{k}{M}} \right] V^*_{t+\frac{k}{M}} \notag \\
	\ge& \left[1 - \sum_{i=1}^{\tilde{N}} \left( \xi_{2i} + \nu_1 + \frac{\nu_2}{p_{i, t + \frac{k}{M}}} \right) w_{i,(t+\frac{k}{M})-} + \sum_{i=\tilde{N}+1}^{N} \left( \xi_{2i} + \nu_1 + \frac{\nu_2}{p_{i, t + \frac{k}{M}}} \right) w_{i,(t+\frac{k}{M})-} \right] V_{(t+\frac{k}{M})-} \notag \\
	& + h_{t+\frac{k}{M}}.
\end{align}

Notice that by the conditions of the proposition,
\begin{align}\label{eq:proposition2_eq4}
	&\sum_{i=1}^{\tilde{N}} \left( \xi_{2i} + \nu_1 + \frac{\nu_2}{p_{i, t + \frac{k}{M}}} \right) w_{i,t+\frac{k}{M}} - \sum_{i=\tilde{N}+1}^{N} \left( \xi_{2i} + \nu_1 + \frac{\nu_2}{p_{i, t + \frac{k}{M}}} \right) w_{i,t+\frac{k}{M}} \notag \\
	\le& \sum_{i=1}^{N} \left(\xi_{2i} + \nu_1 + \frac{\nu_2}{p_{i, t + \frac{k}{M}}}\right) \left|w_{i, t + \frac{k}{M}}\right| \notag \\
	<&1,
\end{align}
and
\begin{align}\label{eq:proposition2_eq5}
	&\sum_{i=1}^{\tilde{N}} \left( \xi_{2i} + \nu_1 + \frac{\nu_2}{p_{i, t + \frac{k}{M}}} \right) w_{i,(t+\frac{k}{M})-} - \sum_{i=\tilde{N}+1}^{N} \left( \xi_{2i} + \nu_1 + \frac{\nu_2}{p_{i, t + \frac{k}{M}}} \right) w_{i,(t+\frac{k}{M})-} \notag \\
	\le& \sum_{i=1}^{N} \left(\xi_{2i} + \nu_1 + \frac{\nu_2}{p_{i, t + \frac{k}{M}}}\right) \left| w_{i,(t+\frac{k}{M})-} \right| \notag \\
	<& 1.
\end{align}

By the conditions of the proposition, 
\begin{align}
	h_{t+\frac{k}{M}} &\ge \ -\xi V_{(t+\frac{k}{M})-} > -\left[1 - \sum_{i=1}^{N}\left( \xi_{2i} + \nu_1 + \frac{\nu_2}{p_{i, t + \frac{k}{M}}} \right) \left| w_{i,(t+\frac{k}{M})-} \right| \right] V_{(t+\frac{k}{M})-}, \notag
\end{align}
which, together with \eqref{eq:proposition2_eq3} to \eqref{eq:proposition2_eq5} and $V_{(t+\frac{k}{M})-} > 0$, implies that $V^*_{t+\frac{k}{M}} > 0$.

We then prove part $(i)$. Suppose the conditions in part $(i)$ hold. Let $h_{t+\frac{k}{M}} = 0$ and define $\xi = \frac{1}{2}(1 - \sum_{i=1}^{N}( \xi_{2i} + \nu_1 + \frac{\nu_2}{p_{i, t + \frac{k}{M}}} ) | w_{i,(t+\frac{k}{M})-}|)$. Then the conditions in part $(ii)$ hold. Applying part $(ii)$ leads to the conclusion of part $(i)$.
\end{proof}

\subsection{Fixed and Proportional Transaction Costs}\label{appendix:fixed_proportional_cost}

In this section, we show that our method in Proposition \ref{prop:achieve_target_weight} also works for the simple form of fixed and proportional transaction cost. We have the following proposition.

\begin{proposition}\label{prop:fixed_proportional_cost}
	Consider the following proportional and fixed transaction cost
	\begin{align}
		c^\text{P}_{t+\frac{k}{M}}(V_{t+\frac{k}{M}}) & = \sum_{i=1}^{N} \left(\xi_{3i} + \xi_{4i} \left|w_{i,t+\frac{k}{M}} V_{t+\frac{k}{M}} - V_{i,(t+\frac{k}{M})-} \right|\right),
	\end{align}	
	where $\xi_{3i}$ and $\xi_{4i}$ are nonnegative constants, $i = 1, 2, \ldots, N$.
	
	($i$) If $V_{(t+\frac{k}{M})-} > 0$, $\sum_{i=1}^{N} \xi_{4i} | w_{i, t+\frac{k}{M}} | < 1$, and $\sum_{i=1}^{N} \xi_{4i} | w_{i,(t+\frac{k}{M})-} | + \frac{\sum_{i=1}^{N}\xi_{3i}}{V_{(t+\frac{k}{M})-}} < 1$, then 
	the function $\Psi_1(V_{t+\frac{k}{M}}) = V_{(t+\frac{k}{M})-} - c^\text{P}_{t+\frac{k}{M}}(V_{t+\frac{k}{M}})$ is a contraction on $\mathbb{R}$ with 
	the contraction coefficient $\sum_{i=1}^{N} \xi_{4i} |w_{i, t + \frac{k}{M}}|$ and a unique fixed point $V^*_{t+\frac{k}{M}}$, which is the unique solution to the rebalancing equation $\Psi_1(V_{t+\frac{k}{M}})=V_{t+\frac{k}{M}}$. In addition, the solution satisfies $0 < V^*_{t+\frac{k}{M}} \le V_{(t+\frac{k}{M})-}$.
	
	($ii$) If all the conditions in (i) hold, and $0 < \xi < 1 - \sum_{i=1}^{N} \xi_{4i} | w_{i,(t+\frac{k}{M})-} | - \frac{\sum_{i=1}^{N}\xi_{3i}}{V_{(t+\frac{k}{M})-}}$, and
	$h_{t+\frac{k}{M}} \ge -\xi V_{(t+\frac{k}{M})-}$, then 
	the function $\Psi_2(V_{t+\frac{k}{M}}) = V_{(t+\frac{k}{M})-} + h_{t+\frac{k}{M}} - c^\text{P}_{t+\frac{k}{M}}(V_{t+\frac{k}{M}})$ is a contraction on $\mathbb{R}$ with 
	the contraction coefficient $\sum_{i=1}^{N} \xi_{4i} |w_{i, t + \frac{k}{M}}|$ and a unique fixed point $V^*_{t+\frac{k}{M}}$, which is the unique solution to the rebalancing equation $\Psi_2(V_{t+\frac{k}{M}})=V_{t+\frac{k}{M}}$. In addition, the solution satisfies $0 < V^*_{t+\frac{k}{M}} \le V_{(t+\frac{k}{M})-} + h_{t+\frac{k}{M}}$.
\end{proposition}

\begin{proof} 
	
	We first prove part ($ii$).
	
	For any $x, y \in \mathbb{R}$,
	\begin{align}\label{eq:c1_k_P}
		\left|  c^{\text{P}}_{t + \frac{k}{M}}(x) - c^{\text{P}}_{t + \frac{k}{M}}(y) \right| &= \left| \sum_{i=1}^{N} \xi_{4i} \left[ \left|w_{i, t + \frac{k}{M}} x - V_{i, (t + \frac{k}{M})-}\right|-  \left|w_{i, t + \frac{k}{M}} y - V_{i, (t + \frac{k}{M})-}\right| \right]\right| \notag \\
		&\le  \sum_{i=1}^{N} \xi_{4i} \left| \left|w_{i, t + \frac{k}{M}} x - V_{i, (t + \frac{k}{M})-}\right|-  \left|w_{i, t + \frac{k}{M}} y - V_{i, (t + \frac{k}{M})-}\right| \right| \notag \\
		&\le |x - y| \sum_{i=1}^{N} \xi_{4i} \left|w_{i, t + \frac{k}{M}}\right|.
	\end{align}
	
	Hence, $c^{\text{P}}_{t + \frac{k}{M}}(x)$ is a contraction as $\sum_{i=1}^{N} \xi_{4i} |w_{i, t + \frac{k}{M}}| < 1$. Since $V_{(t+\frac{k}{M})-} + h_{t+\frac{k}{M}}$ is a constant, $\Psi_2(x) = V_{(t+\frac{k}{M})-} + h_{t+\frac{k}{M}} - c^\text{P}_{t+\frac{k}{M}}(x)$ is a contraction on $\mathbb{R}$. By the Banach Fixed Point Theorem, $\Psi_2$ has a unique fixed point $V^*_{t+\frac{k}{M}}$, i.e., 
	\begin{equation}\label{equ:rebal_eq_proporp}
		V^*_{t+\frac{k}{M}} = \Psi_2(V^*_{t+\frac{k}{M}})= V_{(t+\frac{k}{M})-} + h_{t+\frac{k}{M}} - c^{\text{P}}_{t+\frac{k}{M}}(V^*_{t+\frac{k}{M}}).	
	\end{equation}
	Therefore, $V^*_{t+\frac{k}{M}}$ is the unique solution to the rebalancing equation $\Psi_2(V_{t+\frac{k}{M}}) = V_{t+\frac{k}{M}}$. Since $c^{\text{P}}_{t+\frac{k}{M}}(V^*_{t+\frac{k}{M}})\geq 0$, it follows that 
	$V^*_{t+\frac{k}{M}} = V_{(t+\frac{k}{M})-} + h_{t+\frac{k}{M}} - c^{\text{P}}_{t+\frac{k}{M}}(V^*_{t+\frac{k}{M}})\leq V_{(t+\frac{k}{M})-} + h_{t+\frac{k}{M}}$. We will then show $V^*_{t+\frac{k}{M}} > 0$. Without loss of generality, suppose there exists an integer $\tilde{N} \le N$ such that $i \le \tilde{N}$ implies
	\begin{align}
		V_{i,(t+\frac{k}{M})-} \ge w_{i,t+\frac{k}{M}} V^*_{t+\frac{k}{M}}, \notag
	\end{align}
	and $\tilde{N} < i \le N$ implies
	\begin{align}
		V_{i,(t+\frac{k}{M})-} < w_{i,t+\frac{k}{M}} V^*_{t+\frac{k}{M}}. \notag
	\end{align}
	Then
	\begin{align}
		&c^{\text{P}}_{t+\frac{k}{M}}(V^*_{t+\frac{k}{M}}) = \sum_{i=1}^{\tilde{N}} \xi_{4i} \left(V_{i,(t+\frac{k}{M})-} - w_{i,t+\frac{k}{M}} V^*_{t+\frac{k}{M}} \right) + \sum_{i=\tilde{N}+1}^{N} \xi_{4i} \left(w_{i,t+\frac{k}{M}} V^*_{t+\frac{k}{M}}  - V_{i,(t+\frac{k}{M})-}\right) + \sum_{i=1}^{N} \xi_{3i}, \notag
	\end{align}
	which implies by Equation \eqref{equ:rebal_eq_proporp} that
	\begin{align}\label{eq:proposition2_eq3_prop}
		&\left( 1 - \sum_{i=1}^{\tilde{N}} \xi_{4i} w_{i,t+\frac{k}{M}} + \sum_{i=\tilde{N}+1}^{N} \xi_{4i} w_{i,t+\frac{k}{M}} \right) V^*_{t+\frac{k}{M}} \notag \\
		=& \left( 1 - \sum_{i=1}^{\tilde{N}} \xi_{4i} w_{i,(t+\frac{k}{M})-} + \sum_{i=\tilde{N}+1}^{N} \xi_{4i} w_{i,(t+\frac{k}{M})-} \right) V_{(t+\frac{k}{M})-} - \sum_{i=1}^{N}\xi_{3i} + h_{t+\frac{k}{M}}.
	\end{align}

	Notice that by the conditions of the proposition,
	\begin{align}\label{eq:proposition3_eq4}
		\sum_{i=1}^{\tilde{N}} \xi_{4i} w_{i,t+\frac{k}{M}} - \sum_{i=\tilde{N}+1}^{N} \xi_{4i} w_{i,t+\frac{k}{M}} \le \sum_{i=1}^{N} \xi_{4i} \left| w_{i,t+\frac{k}{M}} \right|<1,
	\end{align}
	and
	\begin{align}\label{eq:proposition3_eq5}
		&\left( 1 - \sum_{i=1}^{\tilde{N}} \xi_{4i} w_{i,(t+\frac{k}{M})-} + \sum_{i=\tilde{N}+1}^{N} \xi_{4i} w_{i,(t+\frac{k}{M})-} \right) V_{(t+\frac{k}{M})-} - \sum_{i=1}^{N}\xi_{3i} \notag \\
		\ge& \left( 1 - \sum_{i=1}^{N} \xi_{4i} \left| w_{i,(t+\frac{k}{M})-} \right| \right) V_{(t+\frac{k}{M})-} - \sum_{i=1}^{N}\xi_{3i} \notag \\
		>& 0.
	\end{align}

	By the conditions of the proposition, 
	\begin{align}
		&h_{t+\frac{k}{M}} \ge \ -\xi V_{(t+\frac{k}{M})-} > -\left(1 - \sum_{i=1}^{N}\xi_{4i} \left| w_{i,(t+\frac{k}{M})-} \right| \right) V_{(t+\frac{k}{M})-} + \sum_{i=1}^{N}\xi_{3i}, \notag
	\end{align}
	which, together with \eqref{eq:proposition2_eq3_prop} to \eqref{eq:proposition3_eq5}, implies that $V^*_{t+\frac{k}{M}} > 0$.
	
	We then prove part $(i)$. Suppose the conditions in part $(i)$ hold. Let $h_{t+\frac{k}{M}} = 0$ and define $\xi = \frac{1}{2}(1 - \sum_{i=1}^{N} \xi_{4i} | w_{i,(t+\frac{k}{M})-} | - \frac{\sum_{i=1}^{N}\xi_{3i}}{V_{(t+\frac{k}{M})-}})$. Then the conditions in part $(ii)$ hold. Applying part $(ii)$ leads to the conclusion of part $(i)$.
\end{proof}

\section{The RL Training Algorithm}\label{appendix:ppo_training}

We present the details of the RL training algorithm in Algorithm \ref{code:ppo_ts}. The meaning and value of hyperparameters in the algorithm are listed in Table \ref{table:parameters}. In particular, the algorithm parallelizes 
the generation of episodes in each epoch using $\tilde M$ agents in order to speed up. 
\\
\begin{breakablealgorithm}
	\caption{The RL training algorithm} 
	\begin{algorithmic}[1] \label{code:ppo_ts}
		\STATE Initialize policy network $\pi_{\theta}$ (i.e., $\left\{\mu_{\theta_1}, \sigma_{\theta_2}\right\}$), value network $V_{\phi}$, rollout buffer $\mathcal{R} \leftarrow \emptyset$, and $h \leftarrow 0$
		\WHILE {$h < H$}
		\FOR {$m = 0, \ldots, \tilde{M}-1$}
		\STATE $k \leftarrow 0$
		\WHILE {$k < K$}
		\STATE $t \leftarrow 0$, initialize state $s_t^{k, m} \sim \rho_0(\cdot)$, $T^{k,m} \leftarrow n$ 
		\WHILE {$t<n$} 
		\STATE $a_t^{k, m} \sim \pi_{\theta}(\cdot | s_t^{k, m}), \ \sigma^{k, m}_{t} \leftarrow \sigma_{\theta_2}(s_t^{k, m}), \ s_{t+1}^{k, m} \sim P(\cdot | a_t^{k, m}, s_t^{k, m}), \ r_{t}^{k, m} \leftarrow R(s_t^{k, m}, a_t^{k, m}, s_{t+1}^{k, m})$, $\ V_t^{k, m} \leftarrow V_{\phi}(s_{t}^{k, m}), \ V_{t+1}^{k, m} \leftarrow V_{\phi}(s_{t+1}^{k, m}), \ \pi_t^{k, m} \leftarrow \pi_{\theta}(a_t^{k, m} | s_t^{k, m})$
		\STATE Append $\left\{s_t^{k, m}, a_t^{k, m}, r_{t}^{k, m}, \left( V_t^{k, m}, V_{t+1}^{k, m}\right), \pi_t^{k, m}, \sigma^{k, m}_{t} \right\}$ to $\mathcal{R}$
		\IF {$s_{t+1}^{k, m}$ is a terminal state}
		\STATE $T^{k,m} \leftarrow t + 1$
		\STATE Break
		\ENDIF
		\STATE $t \leftarrow t + 1$
		\ENDWHILE
		\STATE $k \leftarrow k + 1$
		\ENDWHILE
		\ENDFOR
		\STATE Compute advantages $\hat{A}_t^{k, m}$ and cumulative discounted rewards $\hat{\eta}_t^{k, m}$ using $\left\{r_{t}^{k, m}, \left( V_t^{k, m}, V_{t+1}^{k, m}\right) \right\}$ collected in $\mathcal{R}$, for $t = 0, \ldots, T^{k, m}-1, \ m = 0, \ldots, \tilde{M}-1, k = 0, \ldots, T-1$ by the Generalized Advantage Estimation algorithm \citep{schulman2015high} described in Section \ref{subsubsec:training_strategy}. $T^{k, m}$ is the terminal time step for $m$-th agent in $k$-th episode
		\STATE Add pair $\left\{\hat{A}_t^{k, m}, \hat{\eta}_t^{k, m} \right \}$ to 
		the element $\mathcal{R}[\sum_{i=0}^{m-1} \sum_{j=0}^{K-1} T^{j,i} + \sum_{j=0}^{k-1} T^{j,m} + t]$, 
		such that the element changes from $\left\{s_t^{k, m}, a_t^{k, m}, r_{t}^{k, m}, \left( V_t^{k, m}, V_{t+1}^{k, m}\right), \pi_t^{k, m}, \sigma^{k, m}_{t} \right\}$ to \\
		$\left\{s_t^{k, m}, a_t^{k, m}, r_{t}^{k, m}, \left( V_t^{k, m}, V_{t+1}^{k, m}\right), \pi_t^{k, m}, \sigma^{k, m}_{t}, \hat{A}_t^{k, m}, \hat{\eta}_t^{k, m} \right\}$. $\mathcal{R}[j]$ denotes $j$-th element in $\mathcal{R}$, for  $j = 0, \ldots, |\mathcal{R}| -1$, where $|\mathcal{R}|$ is the size of $\mathcal{R}$
		\STATE Randomly shuffle $\mathcal{R}$
		\STATE Initialize $i \leftarrow 0$
		\WHILE{$i + \tilde{n} \le |\mathcal{R}|$}
		\STATE Collect experience $\left\{ s_{\tau}, a_{\tau}, r_{\tau}, \left(V_{\tau}, V_{\tau+1}\right), \pi_{\tau}, \sigma_{\tau}, \hat{A}_{\tau}, \hat{\eta}_{\tau} \right\}_{\tau = i}^{i+\tilde{n}-1} \leftarrow \left\{\mathcal{R}[l]: i \le l < i+\tilde{n} \right\}$
		\STATE Compute policy loss:
		\STATE $\hat{\mathcal{L}} \leftarrow - \frac{1}{\tilde{n}} \sum_{\tau = i}^{i+\tilde{n}-1} \min \left( \hat{A}_{\tau} \frac{\pi_{\theta}(a_\tau | s_\tau)}{\pi_{\tau}}, \ \hat{A}_{\tau} \text{CLIP}\left( \frac{\pi_{\theta}(a_\tau | s_\tau)}{\pi_{\tau}}, 1 - \epsilon, 1 + \epsilon \right) \right)$
		\STATE Compute value loss:
		\STATE $\hat{\mathcal{V}} \leftarrow \frac{1}{\tilde{n}} \sum_{\tau = i}^{i+\tilde{n}-1} \left(V_{\phi}(s_\tau) - \hat{\eta}_{\tau} \right)^2$
		\STATE Compute entropy loss based on the current policy distribution:
		\STATE $\hat{\mathcal{U}} \leftarrow - \frac{1}{\tilde{n}} \sum_{\tau = i}^{i+\tilde{n}-1} \sum_{j = 1}^{N_a} \left\{ \frac{1}{2} + \frac{1}{2}\ln2 \pi + \ln\sigma_{\tau, j} \right\}$ where $\sigma_{\tau} = \left(\sigma_{\tau, 1}, \ldots, \sigma_{\tau, N_a} \right)$ and $N_a$ is dimension of $a_t$
		\STATE Update $\theta, \phi$ by Adam optimizer \citep{kingma2014adam} after calculating gradients of the loss function $loss(\theta,\phi)=\hat{\mathcal{L}} + e_1 \hat{\mathcal{V}} + e_2 \hat{\mathcal{U}}$
		\STATE $i \leftarrow i + \tilde{n}$
		\ENDWHILE
		\STATE Reset $\mathcal{R}$ to $\emptyset$, and $h \leftarrow h + 1$
		\ENDWHILE
	\end{algorithmic} 
\end{breakablealgorithm}\par

\section{Learning Curves of the RL Method}\label{appendix:learning_curves}

Figure \ref{fig:rb_ewi_loss_reward} shows the learning curves with respect to training losses and cumulative rewards on the training window 02/02/2004-01/02/2018 for the return-based tracking of S\&P 500 EWI.

\begin{figure}[htbp]
	{\includegraphics[width=0.5\linewidth]{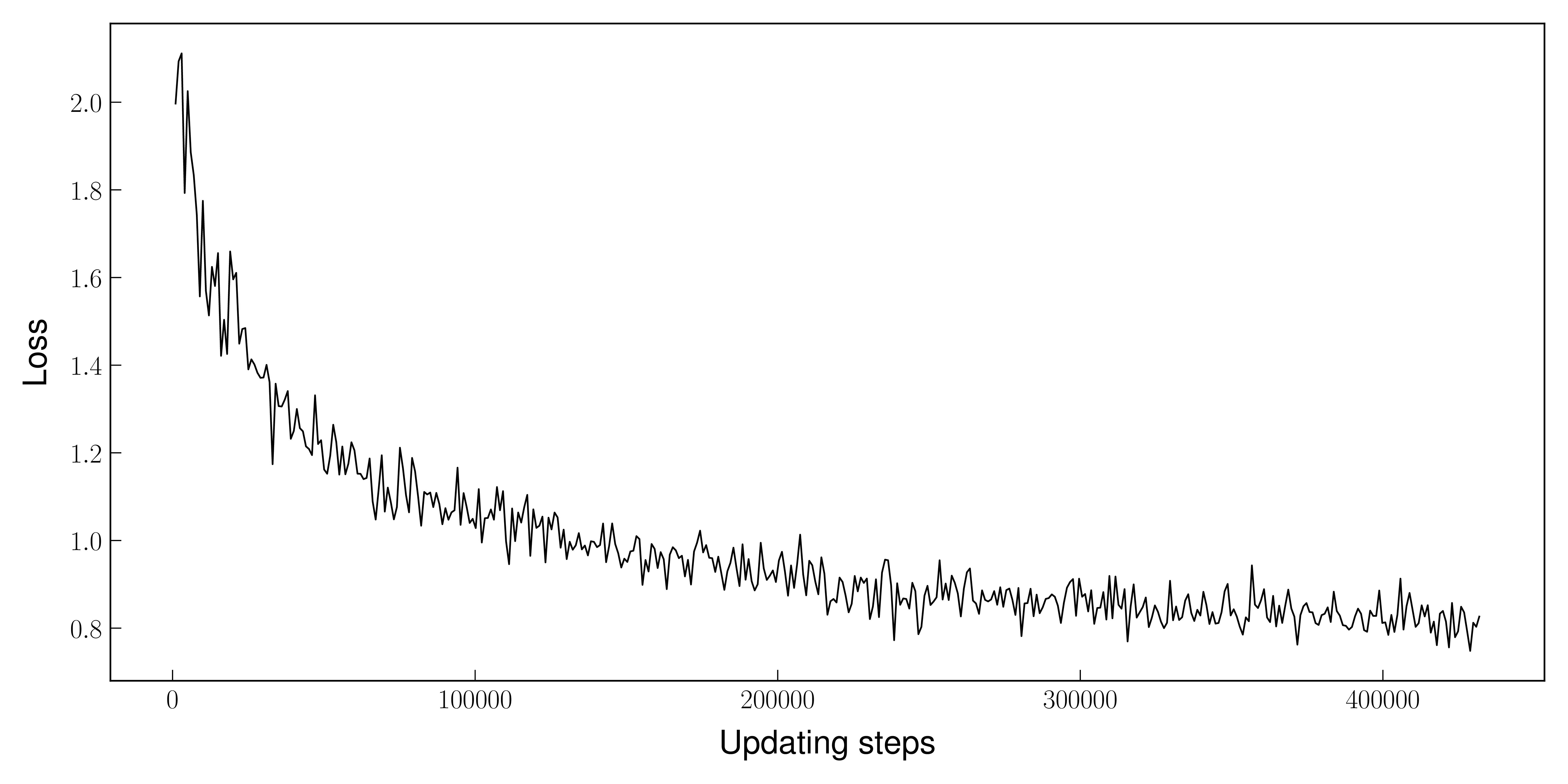}\includegraphics[width=0.5\linewidth]{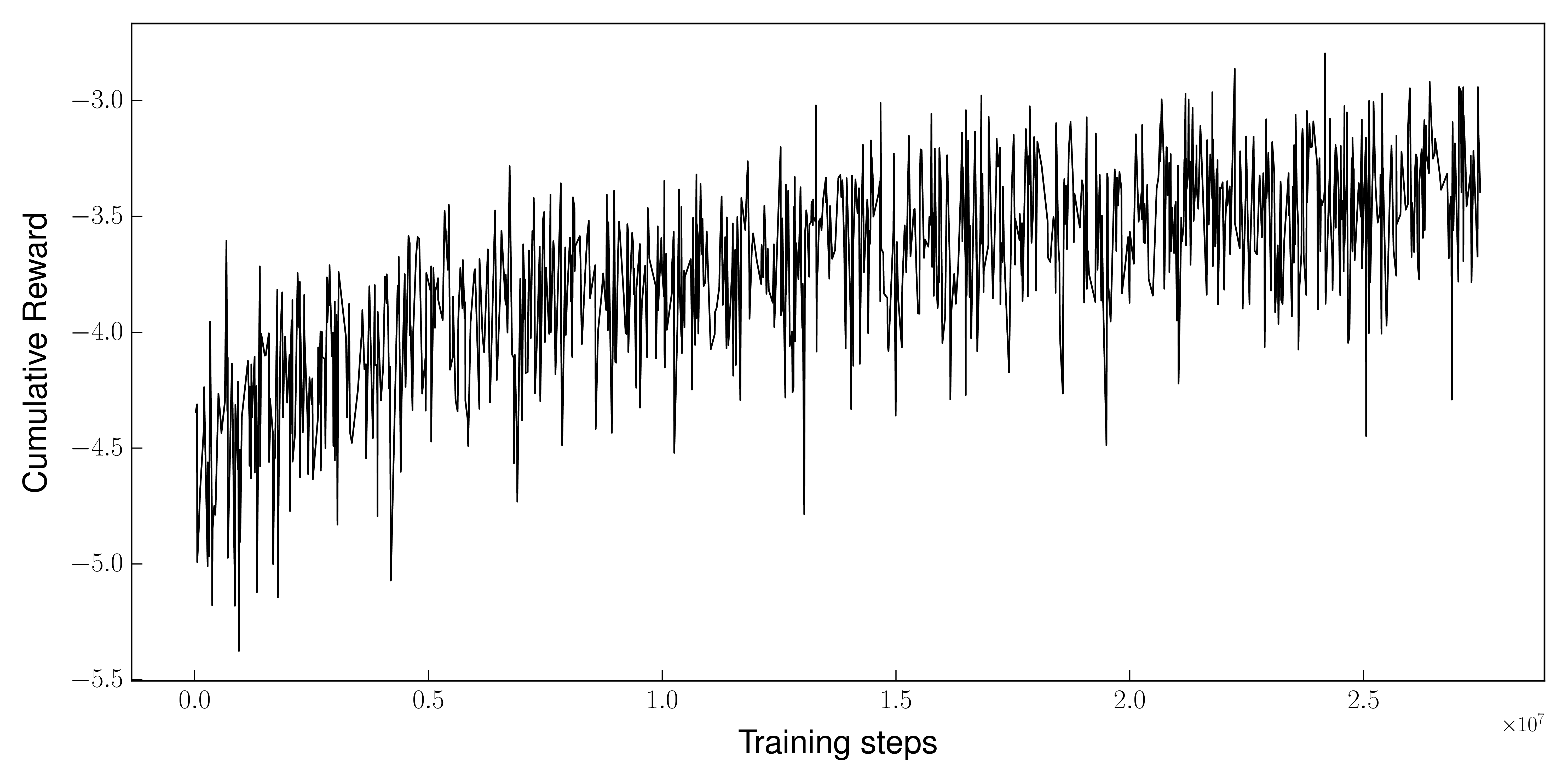}}
	\caption{\textbf{Learning curves of the proposed RL method for the return-based tracking of S\&P 500 EWI on the training window 02/02/2004-01/02/2018}.
	The $x$-axis and the $y$-axis respectively represent updating step and the loss defined in Equation \eqref{eq:total_loss} in the left subfigure, and training step and cumulative reward in the right subfigure.}
	\label{fig:rb_ewi_loss_reward}
\end{figure}

Figure \ref{fig:rb_djia_loss_reward} shows the learning curves with respect to training losses and cumulative rewards on the training window 01/02/1998-01/02/2018 for the return-based tracking of DJIA index.

\begin{figure}[htbp]
	{\includegraphics[width=0.5\linewidth]{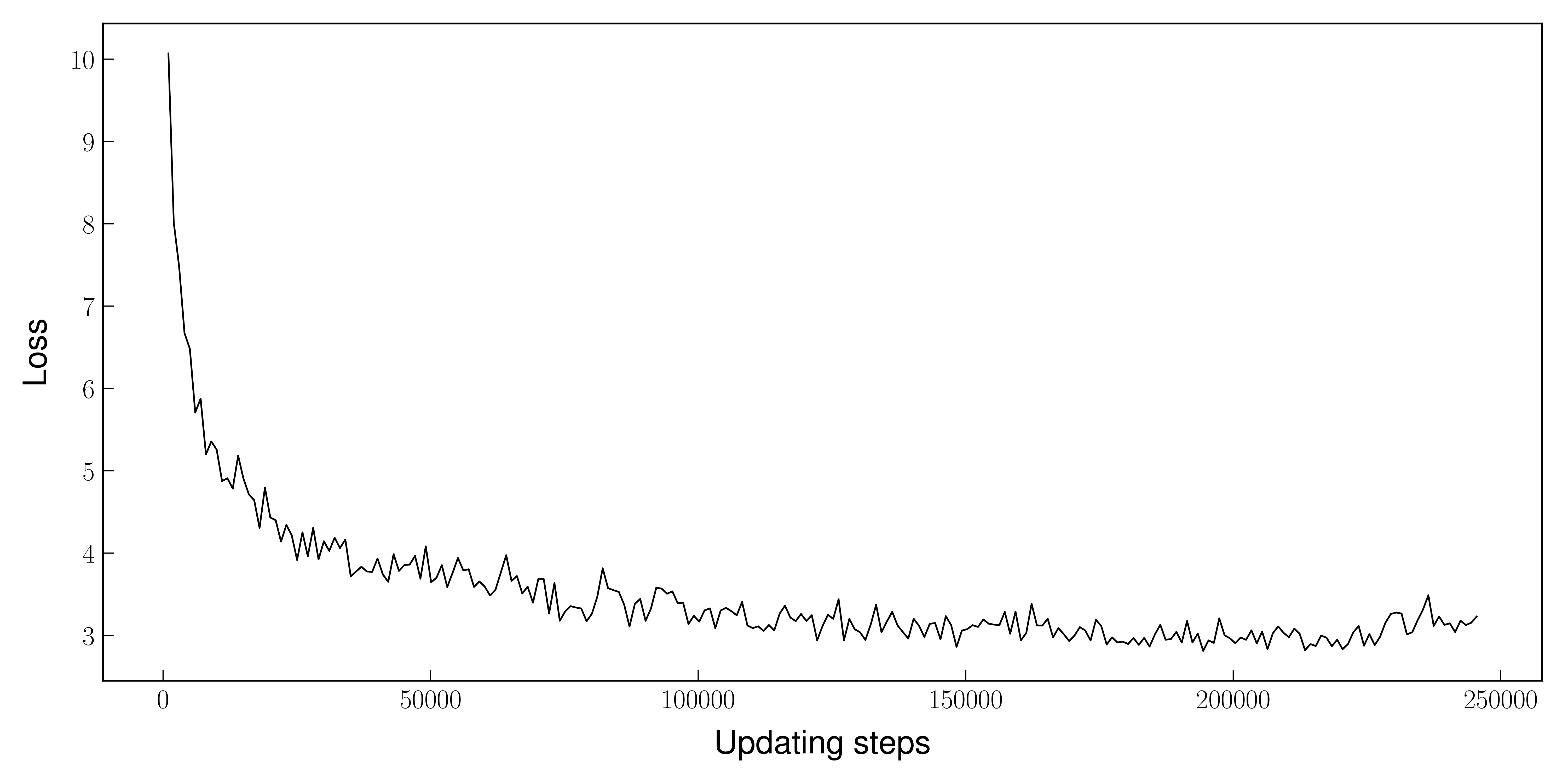}\includegraphics[width=0.5\linewidth]{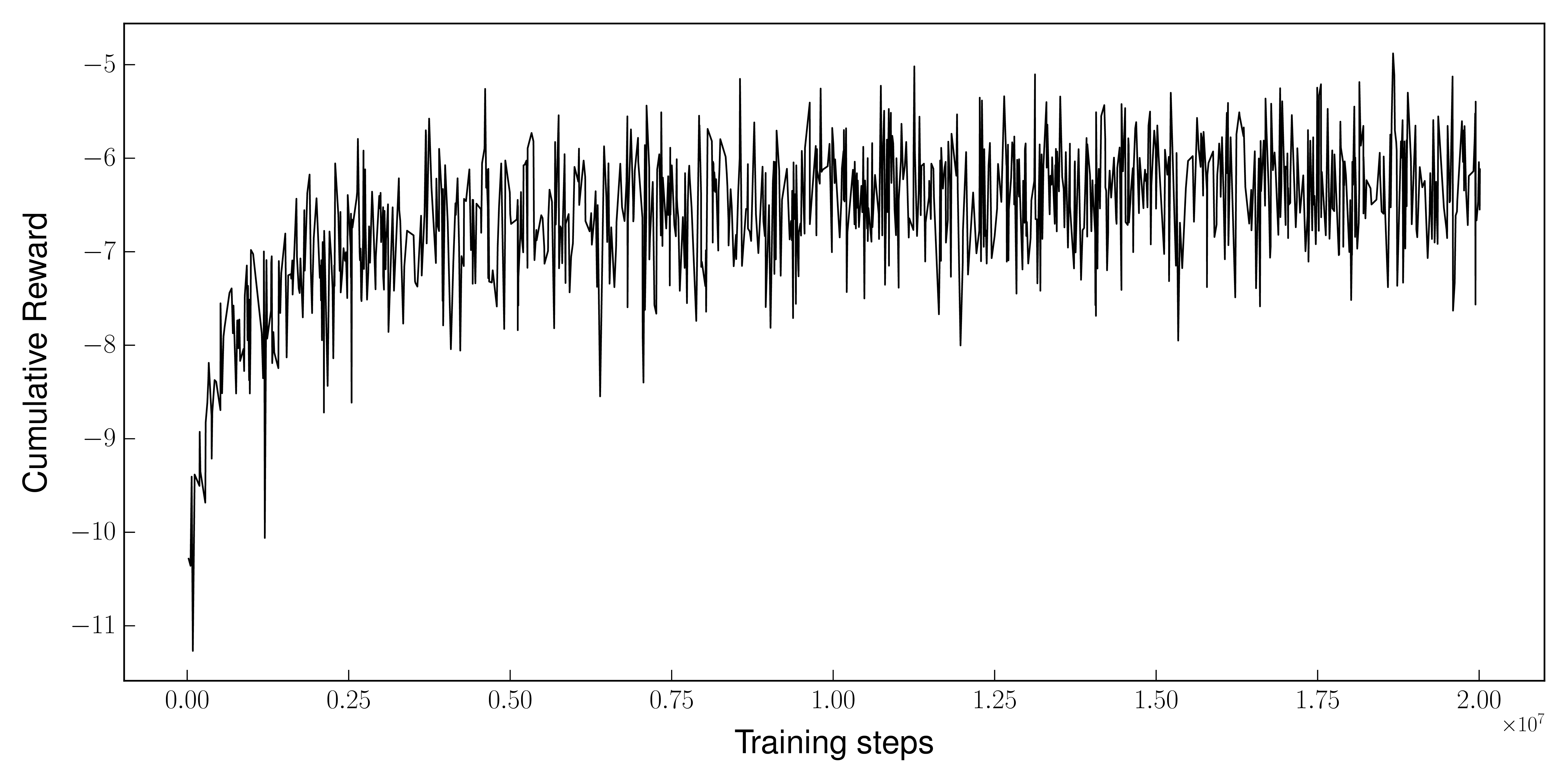}}
	\caption{\textbf{Learning curves of the proposed RL method for the return-based tracking of DJIA index on the training window 01/02/1998-01/02/2018}. The $x$-axis and the $y$-axis respectively represent the updating step and the loss defined in Equation \eqref{eq:total_loss} in the left subfigure, and the training step and the cumulative reward in the right subfigure.}
	\label{fig:rb_djia_loss_reward}
\end{figure}

Figure \ref{fig:vb_ewi_loss_reward} shows the learning curves with respect to training losses and cumulative rewards on the training window 02/02/2004-01/02/2018 for the value-based tracking of S\&P 500 EWI.

\begin{figure}[htbp]
	{\includegraphics[width=0.5\linewidth]{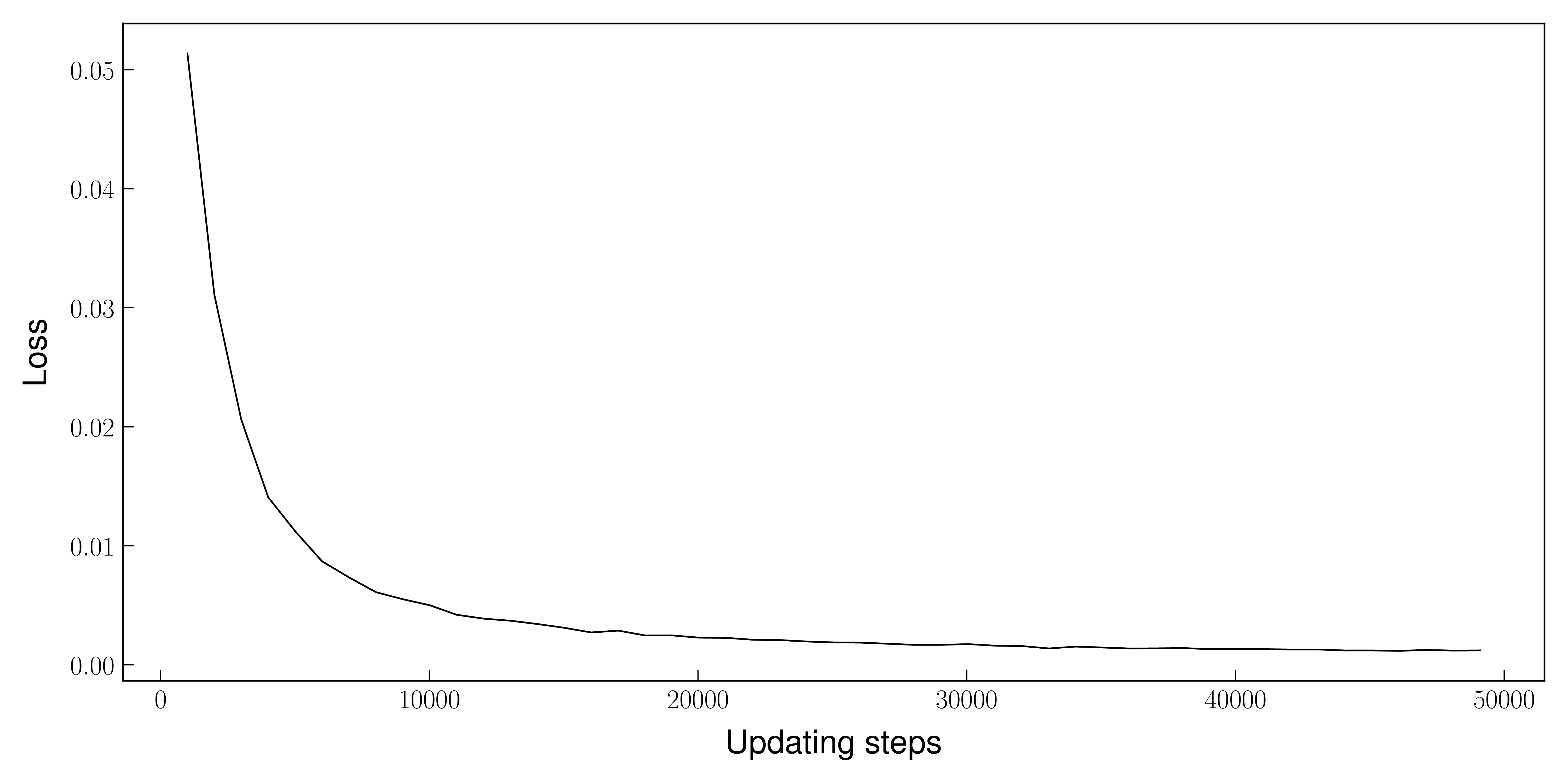}\includegraphics[width=0.5\linewidth]{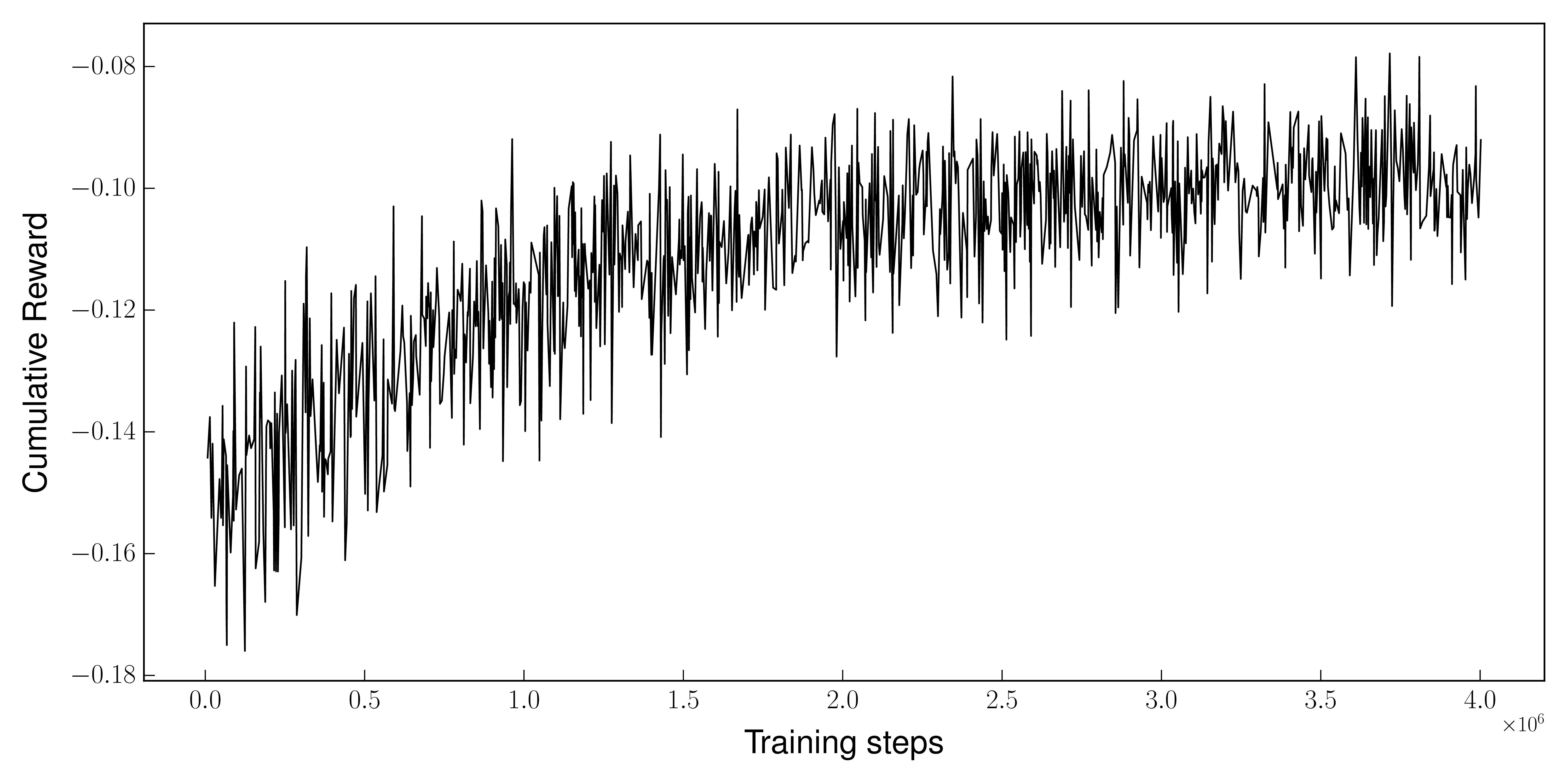}}
	\caption{\textbf{Learning curves of the RL method for the value-based tracking of S\&P 500 EWI}. The figure shows the learning curves of the proposed RL method on the training window 02/02/2004-01/02/2018 for the value-based tracking of S\&P 500 EWI. The $x$-axis and the $y$-axis respectively represent the updating step and the loss defined in Equation \eqref{eq:total_loss} in the left subfigure, and the training step and the cumulative reward in the right subfigure.}
	\label{fig:vb_ewi_loss_reward}
\end{figure}

Figure \ref{fig:vb_djia_loss_reward} shows the learning curves with respect to training losses and cumulative rewards on the training window 01/02/1998-01/02/2018 for the value-based tracking of DJIA index.

\begin{figure}[htbp]
	{\includegraphics[width=0.5\linewidth]{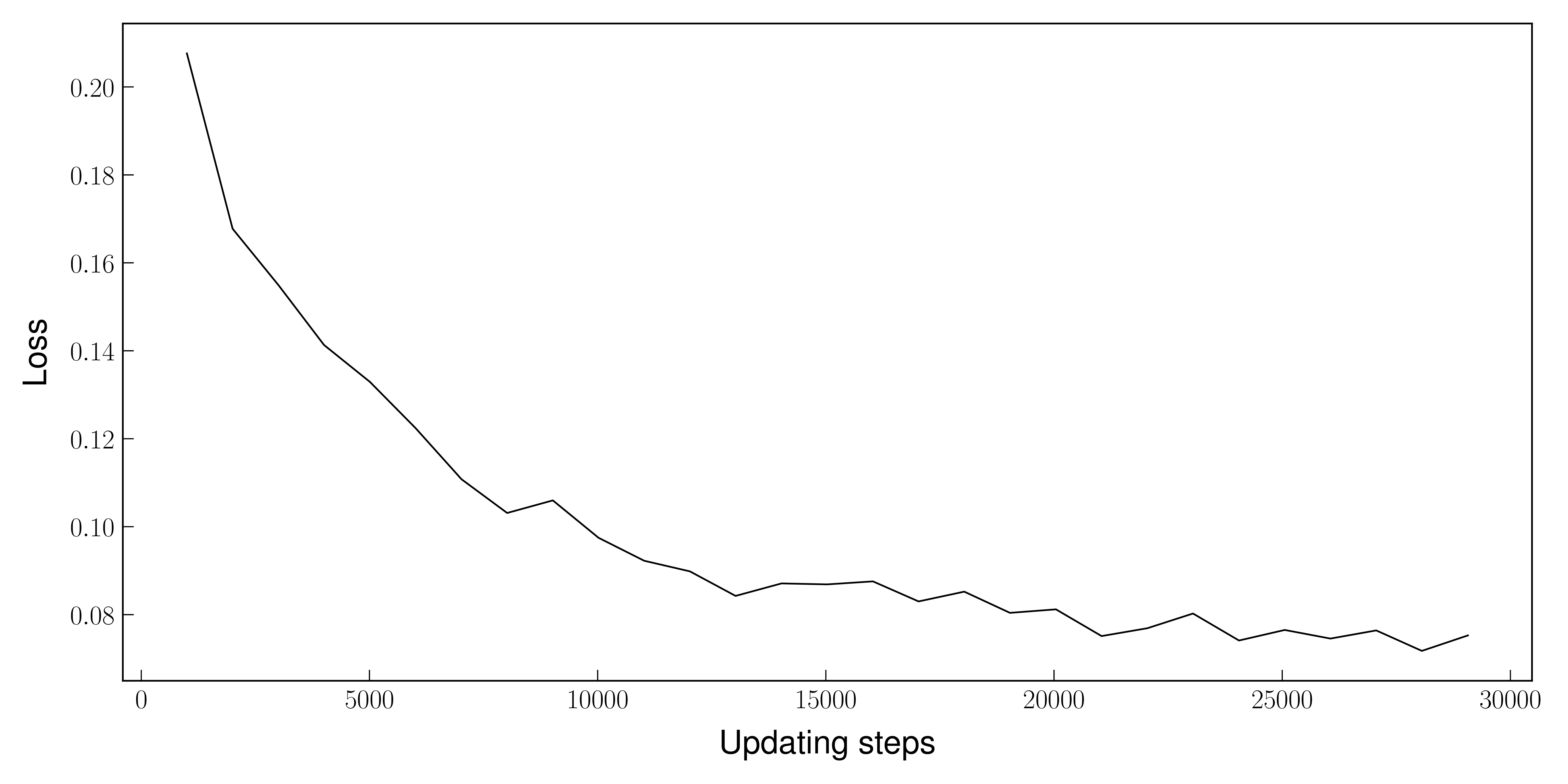}\includegraphics[width=0.5\linewidth]{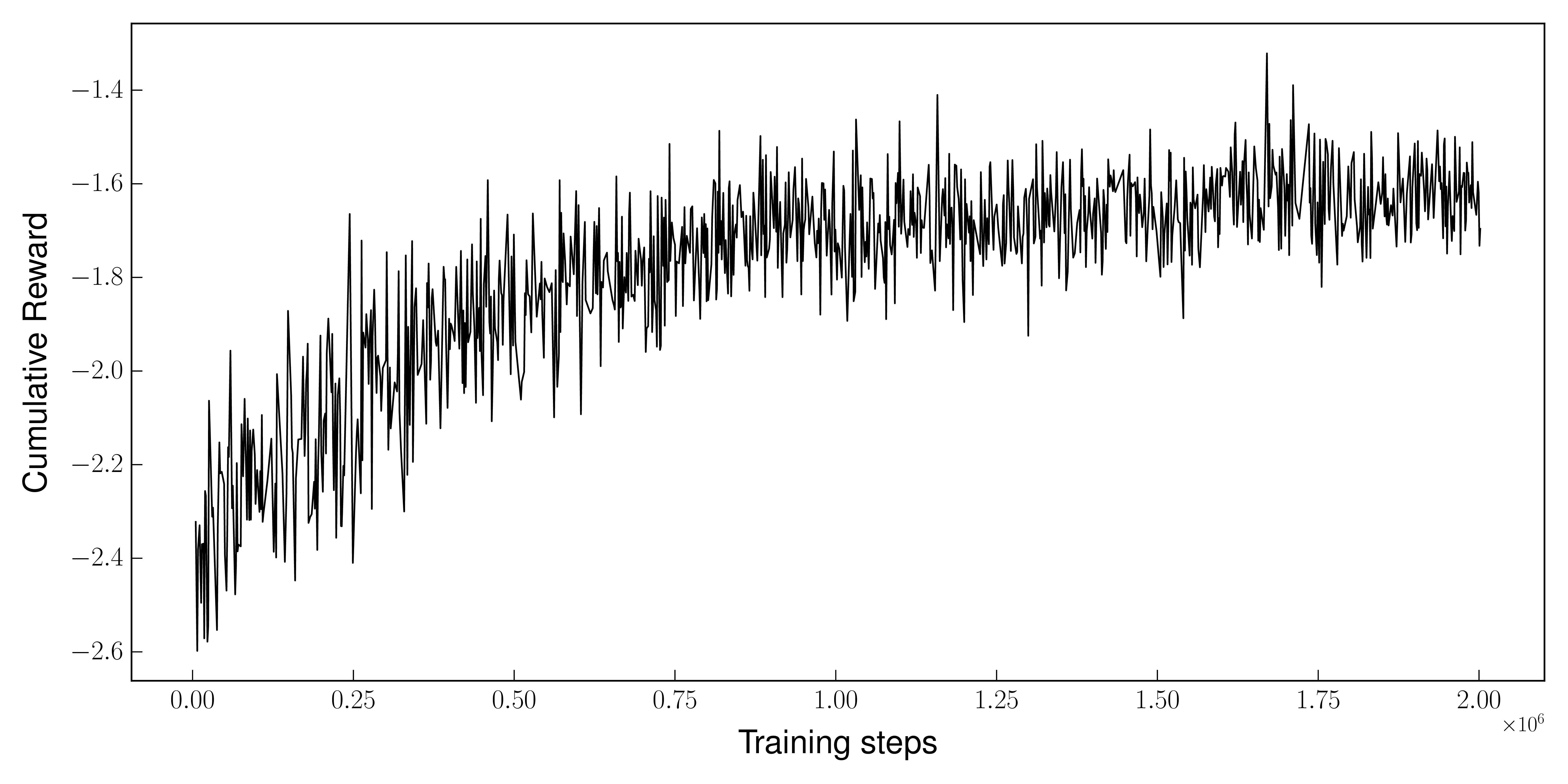}}
	\caption{\textbf{Learning curves of the RL method for the value-based tracking of DJIA index on the training window 01/02/1998-01/02/2018}. The $x$-axis and the $y$-axis respectively represent the updating step and the loss defined in Equation \eqref{eq:total_loss} in the left subfigure, and the training step and the cumulative reward in the right subfigure.}
	\label{fig:vb_djia_loss_reward}
\end{figure}

\section{Trading Volumes and Transaction Costs}\label{appendix:tc_ewi_djia}

Table \ref{tab:rb_ewi_tc} present the out-of-sample trading volumes and transaction costs of the proposed RL method and the MM method for the return-based tracking of S\&P 500 EWI.

\begin{table}[htbp]
	\caption
	{\textbf{Out-of-sample trading volume (Vol) and transaction costs (TC) of the proposed RL method and the MM method for the return-based tracking of S\&P 500 EWI}. The trading volume is in terms of number of shares. The initial wealth $V_{0-}$ of each testing year is USD 20 billion. In the last two rows, mean and stderr respectively refer to the mean and standard error of the values across the testing years from 2006 to 2021 in each column.
	}
	{\begin{tabular}{crrrrrr}
			\toprule
			\multirow{2}{*}{Testing year} & \multicolumn{2}{c}{Vol} & \multicolumn{2}{c}{TC} & \multicolumn{2}{c}{TC/$V_{0-}$} \\
			\cmidrule(lr){2-3}\cmidrule(lr){4-5}\cmidrule(lr){6-7}
			& \multicolumn{1}{c}{RL} & \multicolumn{1}{c}{MM} & \multicolumn{1}{c}{RL} & \multicolumn{1}{c}{MM} & \multicolumn{1}{c}{RL} & \multicolumn{1}{c}{MM} \\
			\midrule
			2006  & 2.686E+09 & 2.117E+09 & 1.343E+07 & 1.058E+07 & 6.716E-04 & 5.292E-04 \\
			2007  & 1.963E+09 & 2.408E+09 & 9.818E+06 & 1.204E+07 & 4.909E-04 & 6.021E-04 \\
			2008  & 5.437E+09 & 5.470E+09 & 2.656E+07 & 2.680E+07 & 1.328E-03 & 1.340E-03 \\
			2009  & 9.255E+09 & 8.806E+09 & 4.554E+07 & 4.308E+07 & 2.277E-03 & 2.154E-03 \\
			2010  & 2.938E+09 & 2.621E+09 & 1.469E+07 & 1.311E+07 & 7.345E-04 & 6.554E-04 \\
			2011  & 2.633E+09 & 2.574E+09 & 1.317E+07 & 1.287E+07 & 6.583E-04 & 6.435E-04 \\
			2012  & 2.635E+09 & 2.657E+09 & 1.318E+07 & 1.329E+07 & 6.588E-04 & 6.644E-04 \\
			2013  & 2.060E+09 & 2.012E+09 & 1.030E+07 & 1.006E+07 & 5.149E-04 & 5.031E-04 \\
			2014  & 1.486E+09 & 1.454E+09 & 7.430E+06 & 7.271E+06 & 3.715E-04 & 3.636E-04 \\
			2015  & 1.643E+09 & 1.623E+09 & 8.213E+06 & 8.117E+06 & 4.107E-04 & 4.058E-04 \\
			2016  & 1.875E+09 & 1.685E+09 & 9.376E+06 & 8.427E+06 & 4.688E-04 & 4.213E-04 \\
			2017  & 1.551E+09 & 1.385E+09 & 7.758E+06 & 6.926E+06 & 3.879E-04 & 3.463E-04 \\
			2018  & 1.485E+09 & 1.601E+09 & 7.279E+06 & 7.811E+06 & 3.640E-04 & 3.906E-04 \\
			2019  & 1.888E+09 & 1.607E+09 & 8.942E+06 & 8.039E+06 & 4.471E-04 & 4.019E-04 \\
			2020  & 2.469E+09 & 2.416E+09 & 1.228E+07 & 1.200E+07 & 6.138E-04 & 6.001E-04 \\
			2021  & 1.454E+09 & 1.400E+09 & 7.273E+06 & 6.999E+06 & 3.636E-04 & 3.500E-04 \\
			\midrule
			mean  & 2.716E+09 & 2.615E+09 & 1.345E+07 & 1.296E+07 & 6.726E-04 & 6.482E-04 \\
			stderr & 4.992E+08 & 4.805E+08 & 2.447E+06 & 2.340E+06 & 1.224E-04 & 1.170E-04 \\
			\bottomrule
	\end{tabular}}
	\label{tab:rb_ewi_tc}
\end{table}

Table \ref{tab:rb_djia_tc} present the out-of-sample trading volumes and transaction costs of the proposed RL method and the MM method for the return-based tracking of DJIA index.

\begin{table}[htbp]
	\caption
	{\textbf{Out-of-sample trading Volume (Vol) and transaction costs (TC) of the proposed RL method and the MM method for the return-based tracking of DJIA index}. The trading volume is in terms of number of shares. The initial wealth $V_{0-}$ of each testing year is USD 20 billion. In the last two rows, mean and stderr respectively refer to the mean and standard error of the values across the testing years from 2005 to 2021 in each column.	
	 }
	{\begin{tabular}{crrrrrr}
			\toprule
			\multirow{2}{*}{Testing year} & \multicolumn{2}{c}{Vol} & \multicolumn{2}{c}{TC} & \multicolumn{2}{c}{TC/$V_{0-}$} \\
			\cmidrule(lr){2-3}\cmidrule(lr){4-5}\cmidrule(lr){6-7}
			& \multicolumn{1}{c}{RL} & \multicolumn{1}{c}{MM} & \multicolumn{1}{c}{RL} & \multicolumn{1}{c}{MM} & \multicolumn{1}{c}{RL} & \multicolumn{1}{c}{MM} \\
			\midrule
			2005  & 1.501E+09 & 1.198E+09 & 7.505E+06 & 5.988E+06 & 3.752E-04 & 2.994E-04 \\
			2006  & 1.626E+09 & 1.313E+09 & 8.130E+06 & 6.566E+06 & 4.065E-04 & 3.283E-04 \\
			2007  & 1.307E+09 & 1.126E+09 & 6.533E+06 & 5.629E+06 & 3.267E-04 & 2.815E-04 \\
			2008  & 3.848E+09 & 3.831E+09 & 1.924E+07 & 1.916E+07 & 9.621E-04 & 9.579E-04 \\
			2009  & 3.416E+09 & 3.456E+09 & 1.708E+07 & 1.728E+07 & 8.539E-04 & 8.640E-04 \\
			2010  & 1.356E+09 & 1.145E+09 & 6.781E+06 & 5.725E+06 & 3.390E-04 & 2.862E-04 \\
			2011  & 1.373E+09 & 1.293E+09 & 6.865E+06 & 6.466E+06 & 3.433E-04 & 3.233E-04 \\
			2012  & 1.293E+09 & 1.161E+09 & 6.465E+06 & 5.803E+06 & 3.232E-04 & 2.902E-04 \\
			2013  & 1.093E+09 & 1.013E+09 & 5.466E+06 & 5.065E+06 & 2.733E-04 & 2.532E-04 \\
			2014  & 6.779E+08 & 6.477E+08 & 3.390E+06 & 3.238E+06 & 1.695E-04 & 1.619E-04 \\
			2015  & 7.107E+08 & 6.324E+08 & 3.554E+06 & 3.162E+06 & 1.777E-04 & 1.581E-04 \\
			2016  & 7.280E+08 & 6.998E+08 & 3.640E+06 & 3.499E+06 & 1.820E-04 & 1.750E-04 \\
			2017  & 6.436E+08 & 6.092E+08 & 3.218E+06 & 3.046E+06 & 1.609E-04 & 1.523E-04 \\
			2018  & 6.806E+08 & 6.488E+08 & 3.403E+06 & 3.244E+06 & 1.702E-04 & 1.622E-04 \\
			2019  & 5.189E+08 & 4.774E+08 & 2.594E+06 & 2.387E+06 & 1.297E-04 & 1.193E-04 \\
			2020  & 6.800E+08 & 6.708E+08 & 3.400E+06 & 3.354E+06 & 1.700E-04 & 1.677E-04 \\
			2021  & 4.504E+08 & 4.283E+08 & 2.252E+06 & 2.142E+06 & 1.126E-04 & 1.071E-04 \\
			\midrule
			mean  & 1.288E+09 & 1.197E+09 & 6.442E+06 & 5.985E+06 & 3.221E-04 & 2.993E-04 \\
			stderr & 2.329E+08 & 2.349E+08 & 1.165E+06 & 1.174E+06 & 5.823E-05 & 5.872E-05 \\
			\bottomrule
	\end{tabular}}
	\label{tab:rb_djia_tc}
\end{table}

Table \ref{tab:vb_ewi_tc} present the out-of-sample trading volume and transaction costs of the proposed RL method and the MM method for the value-based tracking of S\&P 500 EWI.

\begin{table}[htbp]
	\caption
	{\textbf{Out-of-sample trading Volume (Vol) and transaction costs (TC) of the proposed RL method and the MM method  for the value-based tracking of S\&P 500 EWI}. The trading volume is in terms of number of shares. The initial wealth $V_{0-}$ of each testing year is USD 20 billion. In the last two rows, mean and stderr respectively refer to the mean and standard error of the values across the testing years from 2006 to 2021 in each column.			
	 }
	{\begin{tabular}{crrrrrr}
			\toprule
			\multirow{2}{*}{Testing year} & \multicolumn{2}{c}{Vol} & \multicolumn{2}{c}{TC} & \multicolumn{2}{c}{TC/$V_{0-}$} \\
			\cmidrule(lr){2-3}\cmidrule(lr){4-5}\cmidrule(lr){6-7}
			& \multicolumn{1}{c}{RL} & \multicolumn{1}{c}{MM} & \multicolumn{1}{c}{RL} & \multicolumn{1}{c}{MM} & \multicolumn{1}{c}{RL} & \multicolumn{1}{c}{MM} \\
			\midrule
			2006  & 2.708E+09 & 2.117E+09 & 1.353E+07 & 1.058E+07 & 6.766E-04 & 5.292E-04 \\
			2007  & 2.315E+09 & 2.408E+09 & 1.158E+07 & 1.204E+07 & 5.788E-04 & 6.021E-04 \\
			2008  & 5.416E+09 & 5.471E+09 & 2.557E+07 & 2.680E+07 & 1.279E-03 & 1.340E-03 \\
			2009  & 8.306E+09 & 8.807E+09 & 4.043E+07 & 4.308E+07 & 2.022E-03 & 2.154E-03 \\
			2010  & 3.147E+09 & 2.622E+09 & 1.574E+07 & 1.311E+07 & 7.872E-04 & 6.554E-04 \\
			2011  & 2.832E+09 & 2.574E+09 & 1.416E+07 & 1.287E+07 & 7.081E-04 & 6.435E-04 \\
			2012  & 2.968E+09 & 2.658E+09 & 1.485E+07 & 1.329E+07 & 7.424E-04 & 6.644E-04 \\
			2013  & 2.372E+09 & 2.012E+09 & 1.186E+07 & 1.006E+07 & 5.929E-04 & 5.031E-04 \\
			2014  & 1.699E+09 & 1.454E+09 & 8.495E+06 & 7.271E+06 & 4.247E-04 & 3.636E-04 \\
			2015  & 1.769E+09 & 1.623E+09 & 8.845E+06 & 8.117E+06 & 4.423E-04 & 4.058E-04 \\
			2016  & 2.448E+09 & 1.685E+09 & 1.224E+07 & 8.427E+06 & 6.119E-04 & 4.213E-04 \\
			2017  & 2.090E+09 & 1.385E+09 & 1.045E+07 & 6.926E+06 & 5.226E-04 & 3.463E-04 \\
			2018  & 2.029E+09 & 1.601E+09 & 9.705E+06 & 7.811E+06 & 4.853E-04 & 3.906E-04 \\
			2019  & 2.824E+09 & 1.607E+09 & 1.271E+07 & 8.039E+06 & 6.354E-04 & 4.019E-04 \\
			2020  & 2.624E+09 & 2.416E+09 & 1.299E+07 & 1.200E+07 & 6.494E-04 & 6.001E-04 \\
			2021  & 1.525E+09 & 1.400E+09 & 7.607E+06 & 6.999E+06 & 3.804E-04 & 3.500E-04 \\
			\midrule
			mean  & 2.942E+09 & 2.615E+09 & 1.442E+07 & 1.296E+07 & 7.211E-04 & 6.482E-04 \\
			stderr & 4.212E+08 & 4.805E+08 & 2.019E+06 & 2.340E+06 & 1.010E-04 & 1.170E-04 \\
			\bottomrule
	\end{tabular}}
    \label{tab:vb_ewi_tc}	
\end{table}

Table \ref{tab:vb_djia_tc} present the out-of-sample trading volume and transaction costs of the proposed RL method and the MM method for the value-based tracking of DJIA index.

\begin{table}[htbp]
	\caption
	{\textbf{Out-of-sample trading Volume (Vol) and transaction costs (TC) of the proposed RL method and the MM method  for the value-based tracking of DJIA index}. The trading volume is in terms of number of shares. The initial wealth $V_{0-}$ of each testing year is USD 20 billion. In the last two rows, mean and stderr respectively refer to the mean and standard error of the values across the testing years from 2005 to 2021 in each column.		
	}
	{\begin{tabular}{crrrrrr}
			\toprule
			\multirow{2}{*}{Testing year} & \multicolumn{2}{c}{Vol} & \multicolumn{2}{c}{TC} & \multicolumn{2}{c}{TC/$V_{0-}$} \\
			\cmidrule(lr){2-3}\cmidrule(lr){4-5}\cmidrule(lr){6-7}
			& \multicolumn{1}{c}{RL} & \multicolumn{1}{c}{MM} & \multicolumn{1}{c}{RL} & \multicolumn{1}{c}{MM} & \multicolumn{1}{c}{RL} & \multicolumn{1}{c}{MM} \\
			\midrule
			2005  & 2.439E+09 & 1.198E+09 & 1.220E+07 & 5.988E+06 & 6.098E-04 & 2.994E-04 \\
			2006  & 2.652E+09 & 1.313E+09 & 1.326E+07 & 6.566E+06 & 6.628E-04 & 3.283E-04 \\
			2007  & 2.370E+09 & 1.126E+09 & 1.185E+07 & 5.629E+06 & 5.924E-04 & 2.815E-04 \\
			2008  & 5.825E+09 & 3.832E+09 & 2.933E+07 & 1.916E+07 & 1.467E-03 & 9.579E-04 \\
			2009  & 5.571E+09 & 3.456E+09 & 2.800E+07 & 1.728E+07 & 1.400E-03 & 8.640E-04 \\
			2010  & 3.195E+09 & 1.145E+09 & 1.600E+07 & 5.725E+06 & 8.000E-04 & 2.862E-04 \\
			2011  & 3.511E+09 & 1.293E+09 & 1.764E+07 & 6.466E+06 & 8.818E-04 & 3.233E-04 \\
			2012  & 3.012E+09 & 1.161E+09 & 1.511E+07 & 5.803E+06 & 7.557E-04 & 2.902E-04 \\
			2013  & 2.586E+09 & 1.013E+09 & 1.297E+07 & 5.065E+06 & 6.483E-04 & 2.532E-04 \\
			2014  & 1.647E+09 & 6.477E+08 & 8.236E+06 & 3.238E+06 & 4.118E-04 & 1.619E-04 \\
			2015  & 1.450E+09 & 6.324E+08 & 7.250E+06 & 3.162E+06 & 3.625E-04 & 1.581E-04 \\
			2016  & 1.671E+09 & 6.998E+08 & 8.357E+06 & 3.499E+06 & 4.178E-04 & 1.750E-04 \\
			2017  & 1.443E+09 & 6.092E+08 & 7.228E+06 & 3.046E+06 & 3.614E-04 & 1.523E-04 \\
			2018  & 1.974E+09 & 6.488E+08 & 9.898E+06 & 3.244E+06 & 4.949E-04 & 1.622E-04 \\
			2019  & 1.076E+09 & 4.774E+08 & 5.394E+06 & 2.387E+06 & 2.697E-04 & 1.193E-04 \\
			2020  & 1.295E+09 & 6.708E+08 & 6.495E+06 & 3.354E+06 & 3.247E-04 & 1.677E-04 \\
			2021  & 1.027E+09 & 4.283E+08 & 5.150E+06 & 2.142E+06 & 2.575E-04 & 1.071E-04 \\
			\midrule
			mean  & 2.514E+09 & 1.197E+09 & 1.261E+07 & 5.985E+06 & 6.305E-04 & 2.993E-04 \\
			stderr & 3.423E+08 & 2.349E+08 & 1.724E+06 & 1.174E+06 & 8.619E-05 & 5.872E-05 \\
			\bottomrule 
	\end{tabular}}	
    \label{tab:vb_djia_tc}	
\end{table}
	
%
%
%

\end{appendices}

\section*{Acknowledgement} Xianhua Peng is partially supported by the Natural Science Foundation of Shenzhen (Grant No. JCYJ20190813104607549).

\bibliographystyle{dcu}
\bibliography{index_tracking_ref} 


\end{document}